\documentclass[11pt]{article}
\pdfoutput=1

\usepackage[margin=1in]{geometry}

\DeclareMathAlphabet{\mathbbold}{U}{bbold}{m}{n}
\usepackage{amsmath,amsfonts,amssymb,amsthm}
\usepackage{mathtools}
\usepackage[usenames,dvipsnames,svgnames,table]{xcolor}
\usepackage{thm-restate}
\usepackage{braket}
\usepackage[pagebackref]{hyperref}
\hypersetup{
    pdftitle={Classical lower bounds from quantum upper bounds}, 
    pdfauthor={Shalev Ben-David, Adam Bouland, Ankit Garg, and Robin Kothari}, 
    colorlinks=true, 
    linkcolor=blue, 
    citecolor=blue, 
    urlcolor=blue 
}

\renewcommand{\backref}[1]{}

\renewcommand{\backrefalt}[4]{%
\ifcase #1 %
\or
[p.\ #2]%
\else
[pp.\ #2]%
\fi}

\usepackage{tikz,tikz-qtree}
\usetikzlibrary{shapes.misc,shapes.geometric,arrows,positioning,calc,backgrounds}
\usepackage{wrapfig}
\usepackage{enumitem}
\usepackage{tocloft}
\usepackage{microtype}
\usepackage{stmaryrd}

\makeatletter
\newcommand*\rel@kern[1]{\kern#1\dimexpr\macc@kerna}
\newcommand*\widebar[1]{%
  \begingroup
  \def\mathaccent##1##2{%
    \rel@kern{0.8}%
    \overline{\rel@kern{-0.8}\macc@nucleus\rel@kern{0.2}}%
    \rel@kern{-0.2}%
  }%
  \macc@depth\@ne
  \let\math@bgroup\@empty \let\math@egroup\macc@set@skewchar
  \mathsurround\z@ \frozen@everymath{\mathgroup\macc@group\relax}%
  \macc@set@skewchar\relax
  \let\mathaccentV\macc@nested@a
  \macc@nested@a\relax111{#1}%
  \endgroup
}
\makeatother

\makeatletter
\newcommand{\para}{%
  \@startsection{paragraph}{4}%
  {\z@}{2ex \@plus 3.3ex \@minus .2ex}{-1em}%
  {\normalfont\normalsize\bfseries}%
}
\makeatother

\newtheorem{theorem}{Theorem}
\newtheorem{lemma}[theorem]{Lemma}
\newtheorem{proposition}[theorem]{Proposition}
\newtheorem{corollary}[theorem]{Corollary}
\newtheorem{definition}[theorem]{Definition}

\newtheorem{conjecture}{Conjecture}

\theoremstyle{definition}

\newcommand{\eq}[1]{\hyperref[eq:#1]{(\ref*{eq:#1})}}
\renewcommand{\sec}[1]{\hyperref[sec:#1]{Section~\ref*{sec:#1}}}
\newcommand{\thm}[1]{\hyperref[thm:#1]{Theorem~\ref*{thm:#1}}}
\newcommand{\lem}[1]{\hyperref[lem:#1]{Lemma~\ref*{lem:#1}}}
\newcommand{\defn}[1]{\hyperref[def:#1]{Definition~\ref*{def:#1}}}
\newcommand{\prop}[1]{\hyperref[prop:#1]{Proposition~\ref*{prop:#1}}}
\newcommand{\cor}[1]{\hyperref[cor:#1]{Corollary~\ref*{cor:#1}}}
\newcommand{\fig}[1]{\hyperref[fig:#1]{Figure~\ref*{fig:#1}}}
\newcommand{\tab}[1]{\hyperref[tab:#1]{Table~\ref*{tab:#1}}}
\newcommand{\alg}[1]{\hyperref[alg:#1]{Algorithm~\ref*{alg:#1}}}
\newcommand{\app}[1]{\hyperref[app:#1]{Appendix~\ref*{app:#1}}}

\newcommand{\MultiLineComment}[1]{}

\newcommand{\be}{\begin{equation}}
\newcommand{\ee}{\end{equation}}

\newcommand{\B}{\{0,1\}}
\newcommand{\Ba}{\{0,1,*\}}

\renewcommand{\th}[1]{$#1^\mathrm{th}$}

\DeclareMathOperator{\Bernoulli}{Bernoulli}
\DeclareMathOperator{\E}{\mathbb{E}}
\DeclareMathOperator{\adeg}{\widetilde{\deg}}
\DeclareMathOperator{\bdeg}{\widetilde{\mathrm{bdeg}}}

\DeclareMathOperator{\R}{\mathbb{R}}

\DeclareMathOperator{\Q}{Q}

\DeclareMathOperator{\Dom}{Dom}
\DeclareMathOperator{\AND}{\mathsf{AND}}
\DeclareMathOperator{\NAND}{\mathsf{NAND}}

\DeclareMathOperator{\NOTEQ}{\widebar{\mathsf{EQ}}}
\DeclareMathOperator{\OR}{\mathsf{OR}}
\DeclareMathOperator{\XOR}{\mathsf{XOR}}

\DeclareMathOperator{\polylog}{polylog}
\DeclareMathOperator{\poly}{poly}
\DeclareMathOperator{\rank}{rank}
\DeclareMathOperator{\row}{row}
\DeclareMathOperator{\col}{col}

\DeclareMathOperator{\arank}{\widetilde{rank}}
\DeclareMathOperator{\agamma}{\widetilde{\gamma}}
\newcommand{\alogrank}{\log\arank}

\newcommand{\QCC}{\Q^{*}_{\textrm{\scshape cc}}}

\DeclareMathOperator{\QIC}{QIC}
\DeclareMathOperator{\CGT}{\mathsf{CGT}}
\DeclareMathOperator{\SCGT}{\mathsf{SCGT}}
\DeclareMathOperator{\DISJ}{\mathsf{DISJ}}
\DeclareMathOperator{\IP}{\mathsf{IP}}

\DeclareMathOperator{\PrOR}{\mathsf{PrOR}}
\DeclareMathOperator{\Parity}{\mathsf{XOR}}
\DeclareMathOperator{\MAJ}{\mathsf{MAJ}}

\DeclareMathOperator{\PrTH}{\mathsf{PrTH}}
\DeclareMathOperator{\PrAND}{\mathsf{PrAND}}

\newcommand{\tO}{\widetilde{O}}
\newcommand{\tOmega}{\widetilde{\Omega}}

\newcommand{\X}{\mathcal{X}}
\newcommand{\Y}{\mathcal{Y}}

\newcommand{\cl}[1]{\mathsf{#1}}
\newcommand{\cc}{\mathrm{cc}}
\newcommand{\CC}{\mathrm{cc}}

\begin{document}
\title{\bfseries Classical lower bounds from quantum upper bounds}

\author{
Shalev Ben-David\\[.5ex]
University of Maryland\\
\texttt{shalev@umd.edu}
\and
\qquad
\and
Adam Bouland\\[.5ex]
University of California, Berkeley\\
\texttt{abouland@berkeley.edu}
\and
\\
Ankit Garg\\[.5ex]
Microsoft Research\\
\texttt{garga@microsoft.com}
\and
\qquad
\and
\\
Robin Kothari\\[.5ex]
Microsoft Research\\
\texttt{robin.kothari@microsoft.com}
}

\date{}
\maketitle
\thispagestyle{empty}
\begin{abstract}
We prove lower bounds on complexity measures, such as the approximate degree of a Boolean function and the approximate rank of a Boolean matrix, using quantum arguments. We prove these \emph{lower bounds} using a quantum query \emph{algorithm} for the combinatorial group testing problem.

We show that for any function $f$, the approximate degree of computing the $\OR$ of $n$ copies of $f$ is $\Omega(\sqrt{n})$ times the approximate degree of $f$, which is optimal. No such general result was known prior to our work, and even the lower bound for the $\OR$ of $\AND$s function was only resolved in 2013.

We then prove an analogous result in communication complexity, showing that the logarithm of the approximate rank (or more precisely, the approximate $\gamma_2$ norm) of $F:\X\times\Y\to\B$ grows by a factor of $\tOmega(\sqrt{n})$ when we take the $\OR$ of $n$ copies of $F$, which is also essentially optimal.
As a corollary, we give a new proof of Razborov's celebrated $\Omega(\sqrt{n})$ lower bound on the quantum communication complexity of the disjointness problem.

Finally, we generalize both these results from composition with the $\OR$ function to composition with arbitrary symmetric functions, yielding nearly optimal lower bounds in this setting as well.
\end{abstract}

\clearpage
{\tableofcontents}

\clearpage

\section{Introduction}
\label{sec:intro}

Quantum computing promises to allow the efficient solution of certain problems believed to be intractable for classical computers, and is therefore of great practical interest. 
From a mathematical perspective, another important contribution of quantum computing is the rise of the ``quantum method'' as a proof technique. 
That is, often one can prove purely classical (i.e., not quantum) mathematical statements using techniques from quantum information for which no classical proof is known, or where the quantum proof is substantially simpler than its classical counterpart.\footnote{This is analogous to how it is sometimes easier to prove a statement about real numbers using complex numbers, as expressed in the following quote usually attributed to Jacques Hadamard~\cite{Kah91}: ``The shortest path between two truths in the real domain passes through the complex domain''.}
For example, the non-existence of efficient 2-locally-decodable codes was first proven using quantum arguments~\cite{kerenidis2003exponential}.
The closure of the classical complexity class $\textsf{PP}$ under intersection was first shown using classical techniques by Beigel, Reingold, and Spielman \cite{beigel1995pp}, but Aaronson showed it could be reproven using quantum techniques in a simpler way~\cite{aaronsonpostbqp}.
The survey by Drucker and de Wolf provides more examples of this proof technique~\cite{DW11}. 

\para{$\OR$ composition.}
In this work, we apply the quantum method to resolve several composition questions for classical complexity measures in query complexity and communication complexity. 
A quintessential example of this type of question is the $\OR$-composition question, which asks the following:
Given a function $f$, how hard is it to compute the function $\OR_n \circ f$,  the $\OR$ of $n$ copies of $f$? 
One particular strategy for computing $\OR_n \circ f$ is to compose the best algorithms for $\OR_n$ and $f$ in the given model of computation. 
For many complexity measures (including all the measures studied in this paper), the product of the complexities of $\OR_n$ and $f$ will yield an upper bound on the complexity of $\OR_n \circ f$. 
Typically, we conjecture that this upper bound is optimal, but it is not obvious that this must be the case, and hence establishing such a lower bound is usually difficult (or possibly even false for some complexity measures).
For example, it is known that this upper bound is optimal for deterministic~\cite{Tal13,Mon14} and quantum query complexity~\cite{Rei11,LMR+11}, but was only recently established for randomized query complexity~\cite{GJPW17}.

In this paper we show an optimal $\OR$-composition result for approximate degree, a complexity measure in query complexity first studied by Nisan and Szegedy~\cite{nisan1994degree}, which lower bounds quantum query complexity~\cite{BBC+01}, and a nearly optimal $\OR$-composition theorem for approximate rank (or approximate $\gamma_2$-norm or generalized discrepancy), a measure in communication complexity which lower bounds quantum communication complexity~\cite{BdW01,LS09a}.

Our results significantly generalize previous $\OR$-composition results for these measures.
For instance, $\OR$-composition for approximate degree was open for close to 20 years \emph{just for the special case that $f$ is the $\AND$ function}! After several incremental improvements (see \tab{andor}) by Shi~\cite{shi2002approximating}, Ambainis~\cite{ambainis2005polynomial}, and Sherstov~\cite{sherstov2009intersection}, the problem was recently resolved by Sherstov~\cite{She13a} and Bun and Thaler~\cite{BT13} using a linear programming characterization of approximate degree.

In contrast, we show a tight $\OR$-composition theorem for approximate degree for \emph{arbitrary} functions $f$, generalizing these works and newer results on constant-depth compositions of the $\AND$ and $\OR$ functions~\cite{BT15}. (In fact, we also provide an optimal lower bound on the approximate degree of the $\OR$ of possibly different functions $f_i$.) 

In communication complexity, to the best of our knowledge no $\OR$-composition result was known for approximate rank. 
Indeed, such a result would directly imply Razborov's celebrated $\Omega(\sqrt{n})$  lower bound on the quantum communication complexity of the disjointness function~\cite{Raz03}.
To highlight the power of our techniques, we provide a short proof of the $\Omega(\sqrt{n})$ lower bound for disjointness.
We also provide a more direct proof of the recent lower bound on the quantum information complexity of disjointness~\cite{BGKMT15}.

\begin{table}
    \centering
        \begin{tabular}{c|c}
             Bound & Citation \\
             \hline $O(n)$ & H\o yer, Mosca and de Wolf \cite{HMdW03} \\
             $\Omega\left(\sqrt{n}\right)$ & Nisan and Szegedy \cite{nisan1994degree} \\
            $\Omega\left(\sqrt{n\log n}\right)$ & Shi \cite{shi2002approximating} \\
             $\Omega\left(n^{0.66\ldots}\right)$ & Ambainis~\cite{ambainis2005polynomial} \\
             $\Omega\left(n^{0.75}\right)$ & Sherstov \cite{sherstov2009intersection} \\
             $\Omega(n)$ & Sherstov~\cite{She13a} and Bun and Thaler~\cite{BT13}
        \end{tabular}
    \caption{History of lower bounds on the approximate degree of $\OR_n\circ \AND_n$ (from \cite{She13a})}
    \label{tab:andor}
\end{table}

\para{Symmetric function composition.} 
We then generalize our $\OR$-composition results to hold for compositions with arbitrary symmetric functions, which are functions that only depend on the Hamming weight of the input. Other than $\OR$, compositions with symmetric functions like parity and majority have been studied in complexity theory. For instance, the question of how difficult it is to compute  $\XOR_n \circ f$ was already studied in 1982 in Yao's seminal paper on the $\XOR$ lemma~\cite{yao1982theory} (see \cite{ODONNELL200468} for a general composition theorem for $g \circ f$ in this setting.). 
Since the class of symmetric functions includes the $\OR$ function, proving composition theorems for arbitrary symmetric functions is even harder. Such composition theorems are known for deterministic~\cite{Tal13,Mon14} and quantum query complexity~\cite{Rei11,LMR+11}. 
But it remains open to show a similar theorem for randomized query complexity, where only partial results are known \cite{gavinsky2018randomised,San18}.

\para{Techniques.} Although the final results for approximate degree and approximate rank are purely classical, our proofs use quantum algorithms in a crucial way, and there is no known classical proof of these results.
We therefore believe this to be a powerful example of the ``quantum method" \cite{DW11}.
However, we only use quantum algorithms in a black-box manner and the reader is not required to be familiar with quantum query complexity.
We only use its relationship with polynomials due to Beals et al.~\cite{BBC+01} and the existence of a quantum algorithm for the combinatorial group testing problem due to Belovs~\cite{Bel15}. 

Another salient feature of our proofs is that our lower bounds on various measures like approximate degree are proven using the existence of very fast quantum \emph{algorithms} for related problems. 
This is part of a recent trend in complexity theory, sometimes called ``ironic complexity theory"~\cite{aaronsonpnpsurvey}, in which lower bounds are proven using upper bounds. 
For instance, Williams' celebrated circuit lower bound for $\cl{ACC}$ uses this approach~\cite{WilliamsACC}.

Our approach of using a fast quantum algorithm (by Belovs~\cite{Bel15}) to prove lower bounds is inspired by the recent work of Hoza~\cite{hoza2017quantum}, who showed that fast quantum algorithms for certain query problems imply lower bounds in communication complexity. Hoza's work was, in turn, inspired by work of Cleve, van Dam, Nielsen, and Tapp~\cite{cleve1999quantum}, who used the Bernstein--Vazirani algorithm~\cite{BV97} to prove the first lower bound on the quantum communication complexity (with unlimited shared entanglement) of the inner product function.
Similar proof techniques were also used by Buhrman and de Wolf \cite{buhrman1998lower} to show a lower bound on the quantum query complexity of searching a sorted list by a reduction to the hardness of computing parity. 
%

\subsection{Our results}

We now describe our results in more detail.

\subsubsection{Approximate degree}

For any Boolean function $f:\B^n \to \B$, the approximate degree of $f$, denoted $\adeg(f)$, is the minimum degree of any real polynomial $p$ over the variables $x_1,\ldots,x_n$, such that $|f(x)-p(x)|\leq 1/3$ for all $x\in\B^n$. 
Note that $\adeg(f)\leq n$ for all Boolean functions since any Boolean function can be represented exactly with a polynomial of degree $n$. 
Also note that negating the output of a function does not change its approximate degree, and neither does negating input bits. Hence $\adeg(\OR_n) = \adeg(\AND_n) = \adeg(\NAND_n)$ and results for one function carry over to the others.

Approximate degree was first studied by Nisan and Szegedy~\cite{nisan1994degree}.
Since then, it has been used to prove oracle separations, design learning algorithms, and show lower bounds on quantum query complexity, formulas size, and communication complexity. 
(See \cite{She13a,She13b,BT13} and the references therein for more information.) 
It can be used to prove lower bounds on quantum query complexity because for all (total or partial) functions $f$, we have $\Q(f) \geq \frac{1}{2} \adeg(f)$~\cite{BBC+01}, where $\Q(f)$ denotes the bounded-error quantum query complexity of $f$. 

Although approximate degree has a simple definition in terms of polynomials, several simple questions about this measure remain open. Surprisingly, even the approximate degree of the depth-2 AND-OR tree $\AND_n \circ \OR_m$ 
remained open for close to 20 years! 
In 2013, after several incremental improvements (described in \tab{andor}), Sherstov~\cite{She13a} and Bun and Thaler~\cite{BT13} showed that 
\begin{equation}
    \adeg(\AND_n \circ \OR_m) = \Omega(\sqrt{nm}),
\end{equation} 
which is optimal~\cite{HMdW03}. These lower bounds were proved using a linear programming formulation of approximate degree, and exploited certain properties of the dual polynomial for the $\OR$ function.
In contrast to these approaches using dual polynomials, our $\OR$-composition result for approximate degree uses completely different techniques and is more general:

\begin{restatable}{theorem}{approxdeg}
\label{thm:approxdeg}
For any Boolean function $f:\B^m\to\B$, we have 
\begin{equation}
    \adeg(\OR_n \circ f) 
    = \Omega(\sqrt{n}\,\adeg(f)).
\end{equation}
\end{restatable}

This lower bound is tight due to a matching upper bound of Sherstov~\cite{She13b}. This resolves the $\OR$-composition question for approximate degree. As an example, this now allows us to show the optimal  bound $\adeg(\OR_n \circ \MAJ_n) = \Omega(n^{3/2})$, where $\MAJ$ is the majority function. Prior to our work, the best lower bound that could be proved with known techniques was $\adeg(\OR_n \circ \MAJ_n) = \Omega(n)$.

After characterizing the approximate degree of the depth-2 $\AND$-$\OR$ tree, Bun and Thaler~\cite{BT15} also proved that the approximate degree of the depth-$d$ $\AND$-$\OR$ tree on $n$ inputs is $\Omega(\sqrt{n}/\log^{d/2-1}n)$. \thm{approxdeg} straightforwardly implies the optimal bound of $\Omega(\sqrt{n})$.

We then generalize \thm{approxdeg} to a composition theorem for arbitrary symmetric functions $g$. Our $\OR$-composition theorem plays a central role in the proof of our symmetric-function composition theorem, which we discuss in \sec{prooftechniques}. 

\begin{restatable}{theorem}{approxdegsym}
\label{thm:approxdegsym}
For any symmetric Boolean function $g:\B^n\to\B$ and any Boolean function $f:\B^m\to\B$, we have
\begin{equation}
    \adeg(g \circ f) = \tOmega(\adeg(g)\adeg(f)).
\end{equation}
\end{restatable}

This lower bound is also tight up to log factors due to a matching upper bound of Sherstov~\cite{She13b}. This resolves the symmetric-composition question for approximate degree. 

\subsubsection{Approximate rank or \texorpdfstring{$\gamma_2$}{gamma2} norm}

In communication complexity, we have two players Alice and Bob, who hold inputs $x\in\X$ and $y\in\Y$ respectively. 
Their goal is to compute a  function $F:\X \times \Y \to \B$ on their inputs while minimizing the communication between them. 
One of the most studied functions in communication complexity is the set disjointness problem $\DISJ_n:\B^n\times\B^n\to\B$, defined as $\DISJ_n(x,y)=\bigvee_{i=1}^n (x_i \wedge y_i)$ for all $x,y\in\B^n$. 

The quantum communication complexity of the disjointness problem was one of the early open problems in quantum communication complexity. 
Let $\QCC(F)$ denote the bounded-error quantum communication complexity of a function $F$ with unlimited preshared entanglement.
Then it follows from Grover's algorithm~\cite{Gro96} and the query-to-communication simulation algorithm of Buhrman, Cleve, and Wigderson~\cite{BCW98} that $\QCC(\DISJ_n) = O(\sqrt{n}\log n)$, which was later improved to  $\QCC(\DISJ_n) = O(\sqrt{n})$~\cite{AA03}.
However the lower bound remained open until a breakthrough by Razborov~\cite{Raz03}, who showed that $\QCC(\DISJ_n) = \Omega(\sqrt{n})$.

Razborov's result actually lower bounds a smaller complexity measure. With any communication problem $F:\X\times\Y\to\B$, we can associate a $\{-1,+1\}$ matrix, called the sign matrix of $F$, whose $(x,y)$ entry is $(-1)^{F(x,y)}$. Informally, the approximate rank of $F$, denoted $\arank(F)$ is the least rank of any matrix that is entry-wise close to the sign matrix of $F$. (See \sec{comm} for a more precise definition.)
Another measure that is essentially equivalent to approximate rank 
is the approximate $\gamma_2$-norm of the sign matrix of $F$, which we denote $\agamma_2(F)$, also defined in \sec{comm}.
For any function $F$, $\log \agamma_2(F)$ lower bounds its quantum communication complexity, and Razborov's result proves the stronger statement that $\log\agamma_2(\DISJ_n) = \Omega(\sqrt{n})$.

We first show that our techniques yield a new proof of Razborov's celebrated $\Omega(\sqrt{n})$ lower bound for disjointness.

\begin{restatable}{theorem}{disjointness}
\label{thm:disjointness}
Let $\DISJ_n:\B^n\times\B^n \to \B$ be the set disjointness function defined as $\DISJ_n(x,y)=\bigvee_{i=1}^n (x_i \wedge y_i)$ for all $x,y\in\B^n$. Then
\begin{equation}
\log \arank (\DISJ_n) =  \Omega(\sqrt{n}) \text{ and } \log \agamma_2(\DISJ_n) = \Omega(\sqrt{n}).    
\end{equation}

\end{restatable}

Note that this lower bound is tight due to the matching quantum algorithm of Aaronson and Ambainis~\cite{AA03}. Building on this, we generalize our result to an $\OR$-composition theorem.\footnote{An astute reader may worry that an $\OR$-composition theorem cannot possibly hold in communication complexity because some functions do not become harder as we take the $\OR$ of many copies of the function. For example, the function $\NOTEQ_n:\B^n\times\B^n\to\B$, defined as $\NOTEQ_n(x,y)=0$ if and only if $x=y$, can be solved with $O(1)$ communication using a randomized or quantum protocol. Taking the $\OR$ of many copies of $\NOTEQ_n$ only yields a larger instance of $\NOTEQ$, which is no harder than before. However, \thm{polylog_gamma} still holds because $\log \agamma_2(\NOTEQ_n) \leq 0$.}

\begin{restatable}{theorem}{approxrank}
\label{thm:approxrank}
\label{thm:polylog_gamma}
For any function $F:\X \times \Y \to \B$, we have
$\log\agamma_2(\OR_n\circ F)    =\tOmega\left(\sqrt{n}\log\agamma_2(F)\right)$.
\end{restatable}

We then generalize this proof to show a nearly optimal composition theorem for an arbitrary symmetric function $g$ and an arbitrary communication problem $F$.

\begin{restatable}{theorem}{approxranksym}
\label{thm:approxranksym}
For any Boolean function $F:\X \times \Y \to \B$, and any symmetric function $g:\{0,1\}^n\to \{0,1\}$, we have
\begin{equation}
\log\agamma_2(g\circ F)
    \geq \adeg(g)^{1-o(1)} \log\agamma_2(F).    
\end{equation}
\end{restatable}

Note that these lower bounds are also essentially tight, as a matching upper bound of 
$\log\agamma_2(g\circ F) = \tO(\adeg(g) \log\agamma_2(F))$ can be proved by composing a polynomial for $g$ with a matrix for $F$. (For example, this can be done using the construction in \lem{poly_matrix}).

\subsubsection{Further Extensions}

We also prove two further extensions of our result. In \sec{unbalanced}, we generalize our tight $\OR$-composition theorem for $\adeg(\OR_n \circ f)$ to the case of different functions $f_i$. 
We show a tight lower bound on the approximate degree of the $\OR$ of $n$ possibly different functions $f_i$ (which may possibly even have different input sizes). 
This completely characterizes the approximate degree of this function, and furthermore implies that the approximate degree of any constant-depth read-once formula is $\Omega(\sqrt{n})$. This lower bound is optimal, since an upper bound of $O(\sqrt{n})$ is known for the approximate degree of arbitrary read-once formulas (not just constant-depth) via the $O(\sqrt{n})$ upper bound on quantum query complexity \cite{Rei11} and it is an interesting open question if this upper bound is tight for arbitrary read-once formulas.

In \sec{QIC}, we show a lower bound on the quantum information complexity of the disjointness function. 
Quantum information complexity \cite{Touchette} is a information relaxation of quantum communication complexity, in the same sense that information complexity is a relaxation of communication complexity. 
Intuitively, instead of charging for the number of bits (or qubits) of communication if a protocol, information complexity only charges for the information transmitted by these bits (or qubits). 
We use our techniques to reprove the $\Omega(\sqrt{n})$ lower bound on the quantum information complexity of disjointness \cite{BGKMT15} up to log factors. This lower bound is already known, but the known proof uses an alternate characterization of quantum information complexity as amortized quantum communication complexity. 
In contrast, our proof is more direct and works with the information theoretic definition of quantum information complexity.

\subsection{High-level overview of techniques}
\label{sec:prooftechniques}
While we prove several different lower bounds against measures in query and communication complexity, our proofs share several common techniques. 
In particular, all our proofs use Belovs' algorithm for the combinatorial group testing problem~\cite{Bel15}, which we now describe.
Combinatorial group testing has a long history originating in the testing of World War II draftees for Syphilis \cite{du2000combinatorial}, where the goal was to minimize the number of tests used to screen recruits.
The basic idea was to pool multiple blood samples together before testing them; the blood test then reveals if anyone in the pool has the disease.
In other words the test reveals the $\OR$ of the draftee's disease statuses within the group.
One can easily see that if only one person has the disease, then one can use binary search to use only $\log n$ tests to identify which of $n$ people has the disease; similarly one can show that if $k$ people have the disease then $k\log n$ tests suffice.

More formally, in this problem there is a hidden string $x\in\{0,1\}^n$.
One is allowed to query any subset $S\subseteq [n]$, and querying a subset $S$ returns the $\OR$ of the bits of $x$ in the subset, i.e., $\bigvee_{i\in S} x_i$.
The goal is to use these subset queries to learn all of the bits of $x$.
Clearly this can be achieved with $n$ queries in almost any reasonable measure of query complexity, by querying each bit of the input separately, i.e., by querying the subsets $\{1\}, \ldots, \{n\}$.
And as previously mentioned, for sparse inputs one can use fewer than $n$ queries.
But for worst-case inputs this trivial $O(n)$ query algorithm
 is optimal for classical (deterministic or randomized) query complexity.
This is because if the string $x$ contains a single $0$, then this problem reduces to search.
 Therefore even a quantum algorithm for this query problem would require $\Omega(\sqrt{n})$ queries by the lower bound for Grover search \cite{BBBV97}.
Surprisingly, Belovs \cite{Bel15} showed that the quantum query complexity of this problem is at most $O(\sqrt{n})$ as well.
This algorithm will play a key role in our proofs.

\para{Approximate degree $\OR$-composition.} We first describe the ideas required to lower bound the approximate degree of functions of the form $\OR_n \circ f$, making note of the parts of the proof that fail in communication complexity.

Suppose by way of contradiction that $\adeg(\OR_n \circ f)=T$, where $T$ is smaller than expected. 
This means we can compute the $\OR$ of $n$ copies of a function $f$ more easily than expected.
But this also implies we can compute the $\OR$ of any subset $S\subseteq [n]$ of these $n$ copies of $f$, since we can apply this algorithm to any subset $S$ of our choice. (This argument already does not work in communication complexity when only one player knows the subset $S$, since that player would have to communicate $S$ to the other player.)

Now we view the $n$ outputs to the functions $f$ as the hidden string $x\in\B^n$ in the combinatorial group testing problem. 
In the combinatorial group testing problem, we assume we have the ability to query the $\OR$ of any subset of the bits, which is exactly what the assumed polynomial for $\OR_n\circ f$ gives us.
From Belovs' quantum algorithm, we can construct an approximating polynomial for combinatorial group testing using the results of Beals et al.~\cite{BBC+01}. 
More precisely, since combinatorial group testing has an $n$-bit output, which is the hidden string $x\in\B^n$, we use a decision version of this problem that simply outputs the parity of all the bits.
We would now like to compose this polynomial with the assumed polynomials that allow us to compute the $\OR$ of a subset of the functions $f$. 
However, since the polynomials we wish to compose are approximating polynomials, they do not straightforwardly compose as expected, and to make this work, we use Sherstov's robust polynomial construction~\cite{She13b}. Finally, by composing these polynomials of degree $T$ and degree $O(\sqrt{n})$, we get a polynomial of degree $O(T\sqrt{n})$ for computing the parity of all the functions $f$, i.e., we have shown that $\adeg(\XOR_n \circ f)=O(T\sqrt{n})$.

Computing the parity of $n$ copies of a function $f$ is usually $n$ times as hard as computing $f$ in most models of computation. 
Such a result is known for all the measures considered in this paper. 
The argument is now completed by combining the fact that $\adeg(\XOR_n \circ f) = \Omega(n \adeg(f))$~\cite{She12} and $\adeg(\XOR_n \circ f)=O(T\sqrt{n})$. Combining these gives us $T=\Omega(\sqrt{n}\adeg(f))$, as desired.

Our results in communication complexity and the extension to arbitrary symmetric functions build on the ideas presented here. The flowchart in \fig{Dependencies} describes the flow of ideas as well as the dependencies between various sections.

\begin{figure}[t]
\centering
\begin{tikzpicture}
[->,>=stealth',shorten >=1pt,auto,  thick,yscale=0.8,
main node/.style={circle,draw}, node distance = 0.8cm and 1.8cm,
block/.style   ={rectangle, draw, text width=15em, text centered, rounded corners, minimum height=2.5em, fill=white, align=center, font={\footnotesize}, inner sep=5pt}]
    \node[main node,block,line width=2pt] (OR) at (0,0) {Approximate degree OR composition (\sec{queryOR})};
    \node[main node,block] (PrOR) at (-4,-6) {Approximate degree PrOR composition (\sec{queryPrOR})};
    \node[main node,block] (DISJ) at (4,-2) {Approximate rank of Disjointness (\sec{commDISJ})};
   \node[main node,block] (UnbOR) at (-4,2) {Approximate degree unbalanced OR composition (\sec{unbalanced})};
    \node[main node,block] (OR-QIC) at (4,2) {Quantum information complexity OR composition (\sec{QIC})};
    \node[main node,block] (sym) at (-4,-8) {Approximate degree composition for symmetric functions (\sec{querysym})} ;
    \node[main node,block] (OR-CC) at (4,-4) {Approximate rank OR composition (\sec{commOR})} ;
    \node[main node,block] (PrOR-CC) at (4,-6) {Approximate rank PrOR composition (\sec{commPrOR})} ;
    \node[main node,block] (sym-CC) at (4,-8) {Approximate rank composition for symmetric functions (\sec{commsym})} ;

    \path [->] (OR) edge node {} (PrOR);
    \path [->](OR) edge node {} (DISJ);
    \path [->](OR) edge node {} (UnbOR);
    \path [->](OR) edge node {} (OR-QIC);
    \path [->](PrOR) edge node {} (sym);
    \path [->](DISJ) edge node {} (OR-CC);
    \path [->](OR-CC) edge node {} (PrOR-CC);
    \path [->](PrOR-CC) edge node {} (sym-CC);
    \path [->](PrOR) edge node {} (PrOR-CC);
    \path [->](sym) edge node {} (sym-CC);
\end{tikzpicture}
    \caption{Reading order for the results shown in this paper. An arrow from $A$ to $B$ indicates that $A$ is a prerequisite for reading $B$.}\label{fig:Dependencies}
\end{figure}
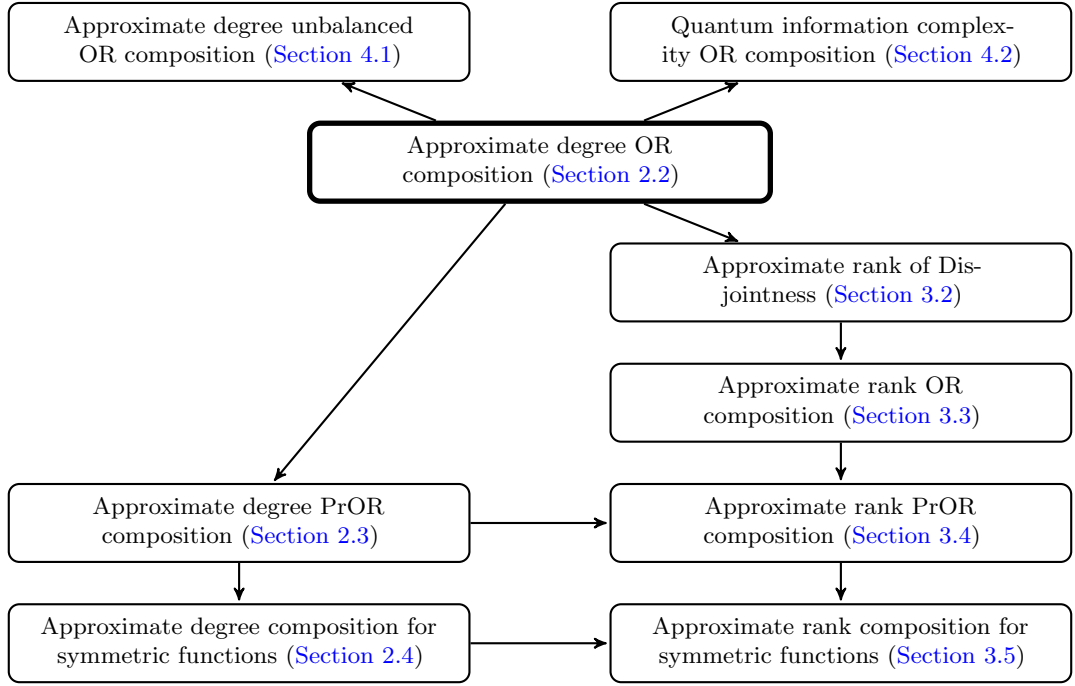

\para{$\OR$ composition in communication complexity.}
The  general strategy outlined above also works in communication complexity for the measures approximate rank and approximate gamma 2 norm, but we need to make additional arguments to make some steps work.

First, as noted above if one player knows a subset $S$ of the shared input, but the other does not, it is not in general possible for them to run a communication protocol on that subset of their shared input. Thus our communication results have some overhead for dealing with this situation. Naively it would seem this overhead is too expensive, since Alice would need to communicate the entire subset $S$ to Bob, which might be more  expensive than the rest of the protocol. However, a recursive argument based on self-reducibility of the $\OR$ function allows the conversion of the additive $O(n)$ loss into a multiplicative polylogarithmic loss. 

The other technically challenging part of porting this argument to communication complexity is in composing approximating polynomials with approximating matrices. This composition does not work as cleanly as in query complexity, and in some cases leads to an additional log factor loss.

\para{Approximate degree $\PrOR$-composition.} 
To lower bound the approximate degree of functions of the form $g \circ f$, where $g$ is a symmetric Boolean function, we first show an intermediate lower bound which will play a key role in our symmetric composition theorem. 
In particular we consider the Promise-OR function, denoted $\PrOR$.
The $\PrOR_n:\B^n \to \Ba$ function is the same as the $\OR$ function with the additional promise that the input has Hamming weight either 0 or 1. We first extend our lower bound on the approximate degree of $\OR_n \circ f$ to the partial function $\PrOR_n \circ f$.
(For partial functions, we require that an approximating polynomial be close to the function on inputs in the domain, and be bounded in $[0,1]$ on all inputs including those outside the domain.)

The main insight that allows us to extend our lower bounds from $\OR$ to $\PrOR$ is that Belovs' algorithm actually solves a more general problem than combinatorial group testing, or more precisely, assumes a weaker access model to the input. 
In particular Belovs' algorithm only requires that the queries that are supposed to return the $\OR$ of a subset $S$, i.e., the value $\bigvee_{i\in S} x_i$, return the correct answer when $\sum_{i\in S} x_i \in \B$. The queries may return incorrect answers on those subsets $S$ for which $\sum_{i\in S} x_i > 1$.

While this is the key conceptual step needed for the generalization, working with partial functions presents several technical challenges.
One of the main challenges corresponds to the robustness of polynomials, for which we used Sherstov's robust polynomial construction~\cite{She13b} previously. 

Recall that our proof strategy is to compose the polynomial induced by Belovs' algorithm with a too-good-to-be-true (approximating) polynomial for $\PrOR_n \circ f$. 
This requires Belovs' polynomial to be robust, i.e. handle noisy inputs, which it may not be. 
For $\OR$ composition, we applied Sherstov's construction to obtain a robust version of Belovs' polynomial, which tolerates $1/3$ noise in the input bits. 
However, Sherstov's robust polynomial construction has a downside---it constructs polynomials that are not multilinear and whose value may blow up when an input variable is not close to being Boolean. This is exactly what can happen when we plug in an approximating polynomial for a partial function. 

For this reason we use a different strategy for composing polynomials without using Shertov's technique.
In particular we prove that the polynomials constructed from quantum algorithms are already mildly robust, i.e., they can handle $1/\poly(q)$ noise in the inputs, where $q$ is the query complexity of the quantum algorithm. Since the polynomials induced by quantum algorithms are multilinear (and hence they are bounded whenever the inputs are in $[0,1]$), this allows us to extend our composition framework to the setting of partial functions (at the expense of losing a logarithmic factor).

We note we are not the first to prove intrinsic robustness of polynomials derived from quantum algorithms. For instance, Buhrman \emph{et al.} \cite{buhrman2007robust} show that the polynomials derived from a quantum algorithm computing a \emph{total} function $f$ can tolerate $O(1/C(f))$ noise in the inputs, where $C(f)$ is the certificate complexity of $f$. However this result is insufficient for our application as we are applying it to a partial function.
It is also not hard to show that all multilinear polynomials are $O(1/n)$ robust to noise, where $n$ is the number of input bits. 
However, this does not suffice for our application either because the number of input bits for $\CGT$ is exponentially larger than its quantum query complexity.
To the best of our knowledge this particular robustness property of polynomials derived from quantum algorithms was not known before and might be of independent interest.

\para{Approximate degree composition for symmetric functions.} 
We now describe how an approximate degree lower bound for $\PrOR_n \circ f$ can be used to lower bound the approximate degree of $g \circ f$, where $g$ is a symmetric function. By a result of Paturi~\cite{Pat92}, it is known that the approximate degree of a symmetric $n$-bit function is completely determined by the Hamming weight closest to $n/2$ where the function $g$ changes value. This implies that it suffices to prove the composition theorem for the case when the outer function $g$ is $\PrTH^k_n$ which is defined as follows:
\begin{equation*}
\PrTH^k_n(x) = \begin{cases} 0 &\mbox{if } |x|=k \\
1 & \mbox{if } |x|=k+1\\ 
* & \mbox{otherwise}
\end{cases}. 
\end{equation*}

Now the elementary, but crucial, observation is that $\PrTH_{2k}^k \circ \PrOR_{n/2k}$ is a sub-function of $\PrTH_n^k$. Hence we can obtain an approximate degree lower bound for $\PrTH^k_n \circ f$ using our composition theorem for $\PrOR$, and a prior composition theorem of Sherstov that works for $\PrTH_{2k}^k$. (Sherstov's result yields optimal composition theorems whenever the outer function has linear approximate degree \cite{She12}). This yields the composition theorem for arbitrary symmetric functions.

\subsection{Open problems}

We end with a discussion of main open problems left open by our work. The foremost open problem is whether the following conjecture is true.

\begin{conjecture}
For all Boolean functions $g:\B^n \to \B$ and $f:\B^m \to \B$, we have $\adeg(g \circ f) = \Omega(\adeg(g) \adeg(f))$.
\end{conjecture}

Our result resolves this (up to log factors) when $g$ is symmetric and $f$ is arbitrary. A related question is whether any of our results can be reproved using the dual polynomials framework that has been used to show recent lower bounds for approximate degree~\cite{She13a,BT13,BT15}. In particular, is there a way to convert a dual witness for $\adeg(f)$ into a dual witness for $\adeg(\OR\circ f)$?

A more open ended question is whether this technique can be generalized to functions other than $\OR$ by developing new quantum algorithms. Belovs' algorithm used $\OR$-queries to a hidden string $x$ to learn all of $x$. What other quantum algorithms of this form exist? Are there nontrivial quantum algorithms that use $g$-queries to learn $x$ for some function $g \notin \{\OR,\XOR\}$? Are there nontrivial quantum algorithms that use $g$-queries to compute some other function $h(x)$ of the input? This motivates the study of a whole class of quantum algorithms, which to the best of our knowledge has not been systematically studied other than in the work of Belovs~\cite{Bel15}.

\section{Approximate degree}
\label{sec:query}

In this section we prove our composition theorems for approximate degree. We start with some definitions and known results in \sec{queryprelim}. In \sec{queryOR} we prove the $\OR$-composition theorem (\thm{approxdeg}), which is the starting point for the more general results proved in this paper. \sec{queryPrOR} generalizes the composition theorem to a partial function related to $\OR$, and then \sec{querysym} proves the final result with arbitrary symmetric functions.

\subsection{Preliminaries}
\label{sec:queryprelim}

In this section we collect some basic definition and known results about partial functions, approximate degree, and quantum query complexity.
Partial functions will play a key role in our proofs, even though the main results are about total Boolean functions. Hence it is necessary to formally define partial functions and extend the definitions of approximate degree and quantum query complexity to partial functions.

\para{Definitions.} A partial Boolean function on $m$ bits is a function that is only defined on a subset of $\B^m$. There are two common ways to talk about partial functions. We can either view it as a function from $D$ to $\B$, where $D\subseteq \B^m$, or as a function $f:\B^m \to \Ba$, where the function evaluates to $*$ outside $D$. We will mostly use the second definition and refer to the subset of $x\in\B^m$ with $f(x)\neq *$ as the ``promise'' and denote it $\Dom(f)$. We can now define the composition of two partial functions more formally.

\begin{definition}
Let $g:\B^n \to \Ba$ and $f:\B^m \to \Ba$ be partial Boolean functions. Then we define $g \circ f:\B^{nm}:\Ba$ to be the partial function $g \circ f (x_1,\ldots, x_n) = g(f(x_1),\ldots,f(x_n))$ on those inputs for which all $ x_i \in \Dom(f)$ and 
$(f(x_1),\ldots,f(x_n))\in\Dom(g)$. The function evaluates to $*$ on all other inputs.
\end{definition}

Most algorithmic models are easily generalized to partial functions. A (classical or quantum) algorithm for a partial function $f$ is only required to be correct on inputs in $\Dom(f)$ and can have arbitrary behavior on inputs outside $\Dom(f)$. 
Extending the definition of approximate degree to partial functions is more subtle, and we motivate it by using an example of a partial function.

Recall that the $\OR$ function on $n$ bits is defined as $\OR_n(x)=0$ if $|x|=0$ and $\OR_n(x) = 1$ if $|x|>0$, where $|x|$ denotes the Hamming weight of $x$ or the number of $1$s in $x$. 
Let us define a partial function related to $\OR$, which we call PromiseOR, as follows. $\PrOR_n:\B^n\to\Ba$ is the $\OR$ function with the additional promise that the input has Hamming weight $0$ or $1$. In other words,
\begin{equation}
\PrOR_n(x) = \begin{cases} 0 &\mbox{if } |x|=0 \\
1 & \mbox{if } |x|=1\\ 
* & \mbox{otherwise}
\end{cases}. 
\end{equation}
Intuitively $\PrOR$ contains the hardest instances of the $\OR$ function, and hence lower bounds for the $\OR$ function should hold against the $\PrOR$ function as well. For example, the quantum query complexity of $\PrOR$ is still $\Theta(\sqrt{n})$, and the deterministic and randomized query complexities of $\PrOR$ are $\Theta(n)$.

The approximate degree of $\PrOR$ is also $\Theta(\sqrt{n})$ as one might expect, as long as we define approximate degree for partial functions appropriately. For a partial function we clearly want the polynomial to approximate the function value on inputs in the promise. But we additionally want the polynomial to be \emph{bounded}.  We say a polynomial $p$ on $m$ variables is bounded if for all $x\in\B^m$, $p(x)\in[0,1]$.

We use the following standard generalization of approximate degree to partial functions that is sometimes called ``bounded approximate degree'' in the literature~\cite{BKT18}.

\begin{definition}[Bounded approximate degree]
For any partial Boolean function $f:\B^m \to \Ba$, the bounded approximate degree of $f$, denoted $\bdeg(f)$, is the minimum degree of any real polynomial $p$ over the variables $x_1,\ldots,x_m$, such that \begin{itemize}
    \item ($p$ is bounded) for all $x\in\B^m$, $p(x) \in [0,1]$, and
    \item ($p$ approximates $f$) for all $x\in\Dom(f)$, $|f(x)-p(x)|\leq 1/3$.
\end{itemize}
\end{definition}

With this generalization of approximate degree, it is indeed true that $\bdeg(\PrOR) = \Theta(\sqrt{n})$, as expected. Note that if we did not require that the polynomial be bounded on all inputs in the domain, then there would be a degree-$1$ polynomial that exactly represents the $\PrOR$ function, which is the polynomial $\sum_{i=1}^m x_i$. 

Finally, we define what it means for a polynomial approximating a Boolean function to be $\delta$-robust to input noise. Informally it means the polynomial continues to approximate the Boolean function even if the input bits are $\delta$-far from being Boolean.

\begin{definition}[$\delta$-robustness to input noise]
\label{def:robust}
Let $h:\B^n\to\Ba$ be a partial Boolean function, and let $p:\B^n \to \R$ be a polynomial. We say that $p$
approximately computes $h$ with robustness $\delta\in[0,1/2)$ if for any $x\in\Dom(h)$ and any $\Delta \in [-\delta,\delta]^n$, we have $|h(x)-p(x+\Delta)|\le 1/3$.
\end{definition}

\para{Known results.} We now collect some facts about bounded polynomials and bounded approximate degree that we need to prove our results.

The first result we use is Sherstov's result on making polynomials robust to noise~\cite[Theorem 1.1]{She13b}. This result states that any polynomial $p$ can be made $1/3$-robust to input noise by only increasing the degree of the polynomial by a constant factor.

\begin{theorem}[Sherstov]
\label{thm:Sherstovrobust}
Let $q:\B^n\to[0,1]$ be a given polynomial. Then there exists a polynomial $q':\mathbb{R}^n\to\mathbb{R}$ of degree $O(\deg(q)+\log(1/\epsilon))$ such that
\begin{equation}
    |q(x)-q'(x+\Delta)| < \epsilon,
\end{equation}
for all $x\in \B^n$ and $\Delta \in [-1/3,1/3]^n$.
\end{theorem}

We will also need a result of Sherstov that establishes the hardness of computing the parity of $n$ copies of a function $f$, or more generally of $n$ different functions $f_1,f_2,\ldots,f_n$. We denote the parity of these $n$ functions  $\XOR_n \circ\,(f_1,f_2,\ldots,f_n)$. Sherstov shows that the approximate degree of the parity of $n$ functions is at least the sum of their approximate degrees~\cite[Theorem 5.9]{She12}.

\begin{theorem}[Sherstov]
\label{thm:Sherstovparity}
For any partial Boolean functions $f_1,f_2,\ldots,f_n$, we have 
\begin{equation}
\bdeg\Bigl(\XOR_n \circ\,(f_1,f_2,\ldots,f_n)\Bigr) = \Omega\Bigl(\sum_i \bdeg(f_i)\Bigr).
\end{equation}
In particular, for any partial Boolean function $f$, we have $\bdeg(\XOR_n \circ f) = \Omega(n \bdeg(f))$.
\end{theorem}

Finally, we also need the following result of Sherstov~\cite[Theorem 6.6]{She12} that proves a composition theorem for bounded approximate degree when the outer function has high degree. 

\begin{theorem}[Sherstov]\label{thm:Sherstov_compose}
Let $g:\B^n\to\Ba$ and $f:\B^m\to\Ba$ be partial Boolean functions. 
Then $\bdeg\left(g\circ f\right)=\Omega\left(\bdeg(g)^2\bdeg(f)/n\right)$.
\end{theorem}

We now formally state the connection between quantum algorithms and approximating polynomials. Beals et al.~\cite{BBC+01} showed that the acceptance probability of a quantum query algorithm that makes few queries can be expressed as a low degree polynomial.

\begin{theorem}[Beals et al.]
\label{thm:Beals}
Let $A$ be a quantum query quantum algorithm that makes $T$ queries to an oracle string $x\in\B^n$ and outputs $1$ with probability $A(x)$. Then there exists a real polynomial $p$ of degree $2T$ over the variables $x_1,\ldots,x_n$ such that for all $x\in\B^n$, $p(x)=A(x)$.
\end{theorem}

By choosing $A$ to be a quantum algorithm that computes a partial function $f$ to bounded error, we get the following corollary.

\begin{corollary}
For any partial Boolean function $f$, $\Q(f)\geq \frac{1}{2} \bdeg(f)$.
\end{corollary}

\subsection{\texorpdfstring{$\OR$}{OR} composition}
\label{sec:queryOR}

In this section we prove our first main result, \thm{approxdeg}.
We start by formally defining the combinatorial group testing problem, whose quantum query complexity was first studied by Ambainis and Montanaro~\cite{AM14}.

\para{Combinatorial group testing problem.} Let $\CGT_{2^n}$ be the following problem.
There is a hidden $n$-bit string $x$, which we have to determine using $\OR$-queries to $x$. In an $\OR$-query, we query the oracle with a subset $S\subseteq[n]$ and the oracle outputs $1$ if there exists an $i\in S$ such that $x_i = 1$. In other words, the oracle's output is the function $\bigvee_{i\in S} x_i$.
Formally, combinatorial group testing is a partial function
\begin{equation}
    \CGT_{2^n}:\B^{2^n}\to\B^n \cup \{*\},
\end{equation} 
where the input is a $2^n$-bit string corresponding to the $\OR$s of all possible subsets of $x$, and the promise is that all bits are indeed the $\OR$ of some string $x \in \B^n$. When the promise is satisfied the desired output is the hidden string $x$. In other words, $y\in\B^{2^n}$ is in $\Dom(\CGT_{2^n})$ if there exists an $x\in\B^n$ such that for all $S\subseteq\B^n$, $y_S = \bigvee_{i\in S} x_i$. For such a $y$, $\CGT_{2^n}(y)=x$. Note that for any $y\in\Dom(\CGT_{2^n})$, the string $x$ is uniquely defined by $x_i = y_{\{i\}}$.

\setlength{\intextsep}{0pt}%
\setlength{\columnsep}{10pt}%
\begin{wrapfigure}{r}{0.3\textwidth}
\vspace{0.5em}
\centering
\begin{tikzpicture}[scale=1]
\node[draw,rounded corners,inner sep=4pt,minimum width=3cm](CGT) at (2,2) {$\CGT_{2^n}$};
\node (y1) at (0,0.25) {$y_{\emptyset}$};
\node (y2) at (1,0.25) {$y_{\{1\}}$};
\node (y3) at (2,0.25) {$y_{\{2\}}$};
\node (y4) at (3,0.25) {$\cdots$};
\node (y5) at (4,0.25) {$y_{[n]}$};
\node (x1) at (1,3) {$x_1$};
\node (x2) at (2,3) {$\cdots$};
\node (x3) at (3,3) {$x_n$};
\draw (y1) -- (CGT);
\draw (y2) -- (CGT);
\draw (y3) -- (CGT);
\draw (y4) -- (CGT);
\draw (y5) -- (CGT);
\draw ([xshift=-1cm]CGT.north) -- ([xshift=-1cm,yshift=0.5cm]CGT.north);
\draw ([xshift=-0.5cm]CGT.north) -- ([xshift=-0.5cm,yshift=0.5cm]CGT.north);
\draw ([xshift=0cm]CGT.north) -- ([xshift=0cm,yshift=0.5cm]CGT.north);
\draw ([xshift=0.5cm]CGT.north) -- ([xshift=0.5cm,yshift=0.5cm]CGT.north);
\draw ([xshift=1cm]CGT.north) -- ([xshift=1cm,yshift=0.5cm]CGT.north);
\end{tikzpicture}
\vspace{-0.5em}
\end{wrapfigure}

Note that although the problem has an input size of $2^n$ bits, the problem is easily solved with $n$ queries as we can simply query all the singleton subsets $y_{\{i\}}$ for $i\in[n]$ to learn all the bits of $x$. Surprisingly, Belovs showed that the quantum query complexity of this problem is quadratically better than this~\cite[Theorem 3.1]{Bel15}.

\begin{theorem}[Belovs]
\label{thm:BelovsCGT}
The bounded-error quantum query complexity of $\CGT_{2^n}$ is $\Theta(\sqrt{n})$. 
\end{theorem}

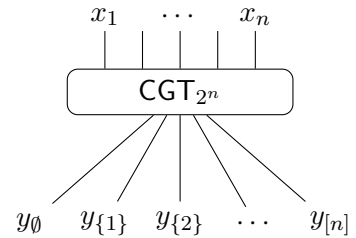
\para{Decision problem associated with $\CGT$.}
Since we want to work with polynomials (and Boolean matrices in \sec{comm}), it will be more convenient to consider a decision problem corresponding to combinatorial group testing. 
To do so, we define the problem 
\begin{equation}
\XOR_n \circ \CGT_{2^n}:\B^{2^n}\to \Ba,    
\end{equation}
which computes the parity of all the output bits of the $\CGT$ function. 

\setlength{\intextsep}{0pt}%
\setlength{\columnsep}{10pt}%
\begin{wrapfigure}{r}{0.3\textwidth}
\vspace{0.25em}
\centering
\begin{tikzpicture}[scale=1]
\node[draw,rounded corners,inner sep=4pt,minimum width=3cm](CGT) at (2,2) {$\CGT_{2^n}$};
\node (y1) at (0,0.25) {$y_{\emptyset}$};
\node (y2) at (1,0.25) {$y_{\{1\}}$};
\node (y3) at (2,0.25) {$y_{\{2\}}$};
\node (y4) at (3,0.25) {$\cdots$};
\node (y5) at (4,0.25) {$y_{[n]}$};
\node[draw,rounded corners,inner sep=4pt,minimum width=3cm](XOR) at (2,3.6) {$\XOR_{n}$};
\draw (y1) -- (CGT);
\draw (y2) -- (CGT);
\draw (y3) -- (CGT);
\draw (y4) -- (CGT);
\draw (y5) -- (CGT);
\draw ([xshift=-1cm]CGT.north) node [above left] {$x_1$} -- ([xshift=-1cm,yshift=1cm]CGT.north);
\draw ([xshift=-0.5cm]CGT.north) -- ([xshift=-0.5cm,yshift=1cm]CGT.north);
\draw ([xshift=0cm]CGT.north) -- ([xshift=0cm,yshift=1cm]CGT.north);
\draw ([xshift=0.5cm]CGT.north) -- ([xshift=0.5cm,yshift=1cm]CGT.north);
\draw ([xshift=1cm]CGT.north) node [above right] {$x_n$} -- ([xshift=1cm,yshift=1cm]CGT.north);
\draw ([xshift=0cm]XOR.north) -- ([xshift=0cm,yshift=0.3cm]XOR.north) node [above] {$\bigoplus_i x_i$};
\end{tikzpicture}
\vspace{1.5em}
\end{wrapfigure}

In other words, $\XOR_n \circ \CGT_{2^n} (y) = \XOR_n(\CGT_{2^n}(y))$, which is the $\XOR$ of all the bits of $x$, the hidden string in the $\CGT$ problem. Of course, any quantum algorithm that solves $\CGT_{2^n}$ and outputs $x$ can instead output the parity of all the bits of $x$.

We can now construct a polynomial that approximates this Boolean function. Using \thm{Beals} and \thm{BelovsCGT}, we can get a polynomial of degree $O(\sqrt{n})$ that approximates $\XOR_n \circ \CGT_{2^n}$ on all inputs in the promise and is bounded in $[0,1]$ outside the promise. 

For our application we need a more robust version of this polynomial. We need a polynomial that also works when the input variables are close to being Boolean. Combining this polynomial with \thm{Sherstovrobust}, we get the following.

\begin{restatable}{theorem}{magicpoly}
\label{thm:magic_poly}
There is a real polynomial $p$ of degree $O(\sqrt{n})$ acting on $2^n$ variables
$\{y_S\}_{S\subseteq[n]}$ such that for any input $y\in\B^{2^n}$ with $\XOR_n\circ\CGT_{2^n}(y)\neq *$, and any $\Delta\in[-1/3,1/3]^{2^n}$,
\begin{equation}
    |p(y+\Delta)-\XOR_n\circ\CGT_{2^n}(y)|\leq 1/3,
\end{equation}
and for all $y\in\B^{2^n}$, $p(y)\in[0,1]$.
\end{restatable}

\begin{proof}
We start with \thm{BelovsCGT} which gives us a quantum algorithm that makes $O(\sqrt{n})$ queries and approximates $\XOR_n\circ\CGT_{2^n}$ to bounded error. Given a quantum algorithm computing a function with probability at least $2/3$, we can always boost the success probability to any constant in $(1/2,1)$ by repeating the quantum algorithm and taking the majority vote of the outcomes. This only increases the quantum query complexity by a constant factor. Hence we can assume the quantum algorithm of \thm{BelovsCGT} has error at most $1/6$ and apply \thm{Beals} to get a polynomial $p'$ of degree $O(\sqrt{n})$ such that for all $y\in\B^{2^n}$ with $\XOR_n\circ\CGT_{2^n}(y)\neq *$,
\begin{equation}
|p'(y)-\XOR_n\circ\CGT_{2^n}(y)|\leq 1/6.    
\end{equation}
Furthermore, because $p'$ arises from a quantum algorithm, we know that even on inputs outside the promise, i.e., inputs with $\XOR_n\circ\CGT_{2^n}(y)= *$, $p'(y) \in [0,1]$. Since $p'$ is bounded in $[0,1]$, we can apply \thm{Sherstovrobust} to it with $\delta=1/6$ to obtain a new polynomial $p$ that is robust to input noise.
Note that since $|p'(y)-p(y)|\leq1/6$ on all inputs $y\in\B^{2^n}$, including those outside the promise, $p'(y) \in [-1/6,7/6]$ for all $y\in\B^{2^n}$.
By rescaling and shifting the polynomial, we can map the interval $[-1/6,7/6]$ to the interval $[0,1]$. Explicitly, we map $p(y)$ to $\frac{3}{4}(p(y)+\frac{1}{6})$. Since the original polynomial was in $[0,1/6]$ for $0$-inputs, this rescaled and shifted polynomial lies in $[0,1/4]$, and similarly for $1$-inputs it lies in $[3/4,1]$, satisfying the conditions of the theorem.
\end{proof}

Using these results we can now prove the main result of this section.

\approxdeg*

\noindent \emph{Proof.}
Assuming  $\adeg(f)\neq 0$ (otherwise the result is trivial), there is an input $w^*$ such that $f(w^*)=0$. Let $w^*$ be any such input.

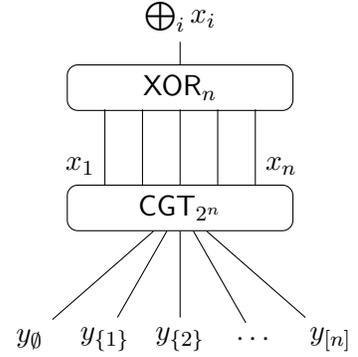
\setlength{\intextsep}{0pt}%
\setlength{\columnsep}{10pt}%
\begin{wrapfigure}{r}{0.42\textwidth}
\vspace{0.25em}
\centering
   \begin{tikzpicture}[-, scale=1,
    level distance = 3em,
    sibling distance = 2pt,
   every internal node/.style={vert}]
    \tikzstyle{vert}=[ellipse,draw=black,minimum size=0.8cm,inner sep=3pt]
    \tikzset{edge from parent/.style=
        {draw, edge from parent path={(\tikzparentnode) -- (\tikzchildnode)}}}
    \Tree
    [.{$\OR_n$}
        \edge node [above left] {$x_1$};
        [.\node {$f$}; 
            {$w_{11}$}
			\edge[draw=white];{$\cdots$}
			{$w_{1m}$}
        ]
		\edge[draw=white];
        [.\node[draw=white]{$\cdots$};
            \edge[draw=white];{$\cdots$}
        ]
        \edge node [above right] {$x_n$};
        [.{$f$}
            {$w_{n1}$}
			\edge[draw=white];{$\cdots$}
			{$w_{nm}$}
        ]
    ]
    \end{tikzpicture}
\vspace{0.5em}
\end{wrapfigure}
Let $q$ be a polynomial of degree $T:=\adeg(\OR_n \circ f)$ that approximates $\OR_n \circ f$. 
Let the input variables of the \th{i} copy of $f$, for $i\in[n]$, be called $w_{i1}, w_{i2}, \ldots, w_{im}$. 
Let us also define for all $i\in[n]$, $x_i := f(w_{i1},w_{i2}\ldots,w_{im})$ to be the output of the \th{i} function $f$. 

Thus $q$ is a polynomial over the variables $w_{11}$ to $w_{nm}$ that approximately computes the Boolean function $\bigvee_{i=1}^n x_i$.
From $q$, we can define for any $S\subseteq [n]$, a new polynomial $q_S$ over the same set of variables $\{w_{ij}:i\in[n], j\in [m]\}$ that approximately computes the Boolean function $\bigvee_{i\in S} x_i$. The polynomial $q_S$ is obtained from $q$ by setting all the inputs to $f$ for which $i\notin S$ equal to the special input $w^*$ for which $f(w^*)=0$.
Thus the polynomial $q_S$ is a polynomial over the same variables as $q$ and has degree at most $T$ and approximates the function $\bigvee_{i \in S} x_i$.

Now from \thm{magic_poly}, we have a real polynomial $p$ of degree $O(\sqrt{n})$ acting on $2^n$ variables $\{y_S\}_{S\subseteq[n]}$ such that for any input $y\in\B^{2^n}$ with $\XOR_n\circ\CGT_{2^n}(y)\neq *$, and any $\Delta\in[-1/3,1/3]^{2^n}$,
\begin{equation}\label{eq:robust21}
    |p(y+\Delta)-\XOR_n\circ\CGT_{2^n}(y)|\leq 1/3.
\end{equation}

Now we define a polynomial $r$ in the variables $w_{11}$ to $w_{nm}$  by taking the polynomial $p$ over variables $y_S$ and replacing each occurrence of the variable $y_S$ with the polynomial $q_S$. 

Then because of equation \eq{robust21} and the fact that the polynomial $q_S$ approximates $\bigvee_{i \in S} x_i$, the polynomial $r$ approximates the parity of the bits $x_i$ (recall $x_i = f(w_{i1},w_{i2}\ldots,w_{im})$). Also note that $r$ is of degree $O(\sqrt{n}T)$. Thus we have
\begin{equation}\label{eq:xoror}
    \adeg(\XOR_n \circ f) = O(\sqrt{n}T) = O(\sqrt{n}\, \adeg(\OR_n \circ f)).
\end{equation}
Since $\adeg(\XOR_n \circ f)) = \Omega(n \, \adeg(f))$ (\thm{Sherstovparity}), we get 
$\adeg(\OR_n \circ f) = \Omega(\sqrt{n}\, \adeg(f))$.
\qed
\bigskip

Note that essentially the same proof yields a weak lower bound for the $\OR$ of $n$ different functions $f_1,f_2,\ldots,f_n$. The proof would follow similarly, except instead of \eq{xoror}, we would arrive at 
\begin{equation}
    \adeg(\XOR_n \circ \, (f_1,f_2,\ldots,f_n)) = O(\sqrt{n}\, \adeg(\OR_n \circ \, (f_1,f_2,\ldots,f_n))).
\end{equation}
Now from \thm{Sherstovparity}, we have that $\adeg(\XOR_n \circ \, (f_1,f_2,\ldots,f_n))=\Omega(n\min_i \adeg(f_i))$, and hence
\begin{equation}\label{eq:ORmin}
    \adeg(\OR_n \circ \, (f_1,f_2,\ldots,f_n))=\Omega(\sqrt{n} \min_i \adeg(f_i)). 
\end{equation}
We will use this weak result in \sec{unbalanced} to establish an optimal bound on the approximate degree of the $\OR$ of $n$ different functions.

\subsection{\texorpdfstring{$\PrOR$}{PrOR} composition}
\label{sec:queryPrOR}

We now extend the main result of the previous section to work with the partial function  $\PrOR_n:\B^n\to\Ba$ introduced in \sec{queryprelim}. $\PrOR$ is just the $\OR$ function with the additional promise that the input has Hamming weight $0$ or $1$. As discussed in \sec{queryprelim}, $\bdeg(\PrOR) = \Theta(\sqrt{n})$. We now generalize the result of the previous section to work for $\PrOR$.

\begin{theorem}
\label{thm:approxdegPrOR}
For any Boolean function $f:\B^m\to\B$, we have
\begin{equation}
    \bdeg(\PrOR_n \circ f) = \Omega(\sqrt{n}\adeg(f)/\log n).
\end{equation}
\end{theorem}

Before proving this, we need to establish some properties of polynomials that arise from quantum algorithms using \thm{Beals}. Recall from \defn{robust} that a polynomial approximating a Boolean function is $\delta$-robust to input noise if the polynomial correctly approximates the function even if the input bits are $\delta$-far from being Boolean.

We now prove that approximating polynomials for Boolean functions that are constructed via quantum algorithms (using \thm{Beals}) naturally possess some robustness to input noise. While it is not hard to show that all multilinear polynomials on $n$ bits are robust to input noise smaller than $O(1/n)$, this is not good enough for our
applications, as the polynomials we are interested in
act on exponentially many variables (see \sec{prooftechniques} for a detailed discussion). Instead,
we show that polynomials that arise from quantum algorithms are much more robust to input noise if the function has small quantum query complexity.

\begin{theorem}\label{thm:multilinear_robust}
Let $h$ be a partial Boolean function. Then there is a bounded multilinear polynomial $q$
of degree $O(\Q(h))$ that approximately computes $h$ with robustness $\Omega(1/\Q(h)^2)$.
\end{theorem}

\begin{proof}
Consider the quantum query algorithm that computes $h$ to error $1/3$
using at most $\Q(h)$ queries. By repeating the algorithm $3$ times
and taking a majority vote, we get an algorithm $A$ computing $h$ to
error $7/27$ using at most $3\Q(h)$ queries.
Let $q$ be the polynomial associated with this algorithm;
this polynomial has degree at most $6\Q(h)$, and for each $x\in\B^n$,
$q(x)$ equals the acceptance probability of the quantum algorithm.
Furthermore, make $q$ multilinear by replacing any squared variable
$x_i^2$ with $x_i$ until no variables have power higher than $1$;
this does not change the behavior of $q$ on $\B^n$, since $1^2=1$
and $0^2=0$.

The resulting polynomial $q$ is multilinear. In addition, since
it evaluates to an acceptance probability on every input in $\B^n$,
it is bounded within $[0,1]$ on $\B^n$.
Together with multilinearity, we conclude $q$ is bounded in $[0,1]$
on all inputs in $[0,1]^n$. On inputs in
$\Dom(h)$, the polynomial $q$ evaluates to the acceptance probability
of a quantum algorithm $A$ computing $h$, so $q$ computes $h$; it remains only
to show that $q$ computes $h$ robustly.

Fix $x\in\Dom(h)$, and fix $y\in[0,1]^n$ that is entry-wise within
$\delta$ of $x$ for $\delta=10^{-5}\Q(h)^{-2}$.
We must show $|p(y)-h(x)|\le 1/3$.

Let $B$ be the probability distribution over $\B^n$
given by sampling each bit $i$ independently from $\Bernoulli(y_i)$,
that is, bit $i$ is $0$ with probability $1-y_i$ and $1$ with probability
$y_i$. By the multilinearity of $q$, it is not hard to see that
\begin{equation}q(y)=\E_{z\sim B}[q(z)].\end{equation}
Moreover, since $q(z)=A(z)$ for $z\in\B^n$
(where $A(z)$ is the acceptance probability of $A$ when run on $z$),
we get $q(y)=\E_{z\sim B}A(z)$.

Split the strings $z\in\B^n$ into two groups: call the strings with
$|A(z)-A(x)|\le4/81$ ``close'' to $x$, and call the strings with
$|A(z)-A(x)|>4/81$ ``far'' from $x$. Let $C$ be the set of close strings.
Then
\begin{equation}q(y)=\Pr_{z\sim B}[z\in C]\cdot\E_{z\sim B|z\in C}[A(z)]
    +\Pr_{z\sim B}[z\notin C]\cdot\E_{z\sim B|z\notin C}[A(z)]\end{equation}
\begin{equation}=\E_{z\sim B|z\in C}[A(z)]+
    \Pr_{z\sim B}[z\notin C](\E_{z\sim B|z\notin C}[A(z)]-\E_{z\sim B|z\in C}[A(z)]).\end{equation}
The expectation $\E_{z\sim B|z\in C}[A(z)]$ is within $4/81$
of $A(x)$. The term $\E_{z\sim B|z\notin C}[A(z)]-\E_{z\sim B|z\in C}[A(z)]$ is has magnitude most $1$, so if we upper bound $\Pr_{z\sim B}[z\notin C]$
by $2/81$, we will conclude that $q(y)$ is within $6/81=2/27$ of
$A(x)$, and hence within $2/27+7/27=1/3$ of $h(x)$.
For this reason, it suffices to show
$\Pr_{z\sim B}[z\notin C]\le 2/81$.

From the hybrid argument~\cite{BBBV97}, let $m_{i,t}$ be the probability that
if $A$ is run on $x$ for $t-1$ queries and its query register is subsequently measured (right before the \th{t} query),
it is found to be querying position $i$.
Let $m_i=m_{i,1}+m_{i,2}+\dots+m_{i,T}$ for all $i\in[n]$,
where $T\le3\Q(h)$ is the total number of queries made by $A$.
Then $\sum_{i=1}^n m_i=T$. Moreover, the hybrid argument tells us
that for any $z\in\B^n$, we have
\begin{equation}\sum_{i:x_i\ne z_i} m_i\ge\frac{|A(x)-A(z)|^2}{4T}.\end{equation}
For strings $z\notin C$, this sum is at least $4/3^8T$.

Note that $\E_{z\sim B}\sum_{i:x_i\ne z_i} m_i\le \delta\sum_{i=1}^n m_i=\delta T$, since each bit $z_i$ differs from $x_i$ with probability
at most $\delta$. By Markov's inequality, the probability
over $z\sim B$ of the event $\sum_{i:x_i\ne z_i}\ge 4/3^8T$ is less
than $3^8\delta T^2/4$. This upper bounds the probability that
$z\notin C$ for $z\sim B$. Since we have $\delta\le 10^{-5}\Q(h)^{-2} \le (8/3^{12})T^{-2}$,
we conclude $\Pr_{z\sim B}[z\notin C]\le 2/81$, as desired.
\end{proof}

We now generalize the combinatorial group testing problem defined in \sec{queryOR} to a problem we call  ``singleton combinatorial group testing.''
In this problem instead of being able to query the $\OR$ of any subset of bits of the unknown input $x$, we can only query the $\PrOR$ of this string $x$. In other words, we are only guaranteed to receive the correct answer if the subset queried contains exactly zero or one $1$s. In all other cases, the query response may be arbitrary.

\begin{restatable}[Singleton CGT]{definition}{scgt}
\label{def:SCGT}
Let $\SCGT_{2^n}:D \to \B^n$ be a partial function with $D \subseteq \B^{2^n}$. 
Let $D$ be the set of all $z\in\B^{2^n}$ for which there exists an $x\in\B^n$ with the property that for all $S\subseteq[n]$ satisfying $\sum_{i\in S} x_i\in\B$, we have $\sum_{i\in S} x_i=z_S$. 
Note that for all $z\in D$, the string $x$ is uniquely defined by $x_i = z_{\{i\}}$, and we denote this string $x(z)$. 
We define the partial Boolean function $\SCGT_{2^n}:D\to\B^n$ by $\SCGT_{2^n}(z):=x(z)$.
\end{restatable}

Note that both $\CGT_{2^n}$ and $\SCGT_{2^n}$ are partial functions that agree on the inputs that are in both their domains, but $\SCGT_{2^n}$ is a  more general problem in the sense that its promise strictly contains the promise of $\CGT_{2^n}$. In symbols, for all $z\in\Dom(\CGT_{2^n})$, $\CGT_{2^n}(z)=\SCGT_{2^n}(z)$, and $\Dom(\CGT_{2^n}) \subsetneq \Dom(\SCGT_{2^n})$.

Remarkably, Belovs' algorithm for $\CGT$ also works for this more general problem.

\begin{restatable}{theorem}{strongBelovs}
\label{thm:strong_Belovs}
The bounded-error quantum query complexity of $\SCGT_{2^n}$ is $\Theta(\sqrt{n})$.
\end{restatable}

This claim essentially follows from Belovs' construction, but for completeness we include a full proof in \app{belovsSCGT}. We also provide a proof sketch here for readers familiar with Belovs' proof.

\begin{proof}[Proof sketch.]
Belovs upper bounds the quantum query complexity of $\CGT$ by exhibiting a solution to the dual of the adversary SDP that is known to characterize quantum query complexity~\cite{Rei11,LMR+11}. Instead of reproducing his proof and observing that it works just as well for $\SCGT$, we explain why the proof goes through.

The SDP solution that Belovs constructs is fully described the the vectors constructed at the beginning of the proof, which in the notation of his paper are called $\psi\llbracket  A \rrbracket$. 
His set $A\subseteq[n]$ corresponds to the set of bits in the hidden input $x$ that are equal to $1$. 
For every input $A$ (or $x\in\B^n$ in our notation), he constructs a vector $\psi\llbracket  A \rrbracket$ indexed by $S\subseteq[n]$, which corresponds to querying the subset $S$. 

However, the only nonzero entries of this vector are those for which $|S \cap A| \in \B$, i.e., where the subset queried has intersection $0$ or $1$ with the hidden input $x \in \B^n$. 
Hence this vector is easy to generalize to all inputs $z\in \Dom(\SCGT_{2^n})$. 
Note that in $\CGT$, $z$ is uniquely determined by the hidden string $x\in \B^n$, but in $\SCGT$ only some of the bits of $z$ are fixed by $x$. 
However, the vector $\psi$ is fixed by $x$ alone, and hence the solution for the SDP corresponding to an input $z$ only depends on the hidden input $x$ and not on the irrelevant variables of $z$ (i.e., the ones not fixed by the hidden input $x$).

Since the vectors are the same as in the original SDP, the value of the objective function remains the same. We only need to check that the constraint is satisfied. 
The constraint is on a pair of inputs $z$ and $z'$, and sums over bits on which they differ. 
But this sum will again not depend on the irrelevant bits of $z$ and $z'$, since the vectors only depend on the underlying hidden inputs, which completes the proof.
\end{proof}

Before proving \thm{approxdegPrOR}, we will prove an
analogue of \thm{magic_poly}, showing the existence of
a robust polynomial for $\XOR_n\circ\SCGT_{2^n}$.
However, unlike in \thm{magic_poly},
we will require the polynomial
to be multilinear.
The reason for this extra requirement is that we will
need plug in polynomials for $\PrOR$ into the variables,
and those might not approximate $\B$ values on all inputs.

This multilinearity requirement means
that we cannot use Sherstov's robustification construction
(\thm{Sherstovrobust}).
Instead, we use \thm{multilinear_robust},
which shows that polynomials coming from quantum algorithms
are always slightly robust. The weaker robustness condition
will later cause us to lose a logarithmic factor.

\begin{theorem}\label{thm:multilinear_magic_poly}
There is a real polynomial $p$ of degree $O(\sqrt{n})$
acting on $2^n$ variables $\{y_S\}_{S\subseteq[n]}$
and a constant $c\ge 10^{-5}$ such that for any input
$y\in\B^{2^n}$ with $\XOR_n\circ\SCGT_{2^n}(y)\neq *$, and any 
$\Delta\in[-c/n,c/n]^{2^n}$,
\begin{equation}
    |p(y+\Delta)-\XOR_n\circ\SCGT_{2^n}(y)|\leq 1/3,
\end{equation}
and for all $y\in\B^{2^n}$, $p(y)\in[0,1]$. In addition,
$p$ is multilinear.
\end{theorem}

\begin{proof}
We know that Belovs' quantum algorithm works for $\SCGT_{2^n}:\B^{2^n} \to \B^n$ (\thm{strong_Belovs}).  This algorithm can be modified to output the parity of the  $n$ output bits, yielding a quantum algorithm for $\XOR_n \circ \SCGT_{2^n}$. \thm{multilinear_robust} then gives
us the desired polynomial $p$.
\end{proof}

We are now ready to prove the main result of this section, \thm{approxdegPrOR}. 
This proof is similar in structure to the proof of \thm{approxdeg}, with some crucial differences. 
We retain the proof structure and variable names of \thm{approxdeg} to highlight the similarity in structure.

\bigskip \noindent \emph{Proof of \protect{\thm{approxdegPrOR}}.}
Assuming  $\adeg(f)\neq 0$ (otherwise the result is trivial), there is an input $w^*$ such that $f(w^*)=0$. Let $w^*$ be any such input.

\setlength{\intextsep}{0pt}%
\setlength{\columnsep}{10pt}%
\begin{wrapfigure}{r}{0.42\textwidth}
\centering
   \begin{tikzpicture}[-, scale=1,
    level distance = 3em,
    sibling distance = 2pt,
   every internal node/.style={vert}]
    \tikzstyle{vert}=[ellipse,draw=black,minimum size=0.8cm,inner sep=3pt]
    \tikzset{edge from parent/.style=
        {draw, edge from parent path={(\tikzparentnode) -- (\tikzchildnode)}}}
    \Tree
    [.{$\PrOR_n$}
        \edge node [above left] {$x_1$};
        [.\node {$f$}; 
            {$w_{11}$}
			\edge[draw=white];{$\cdots$}
			{$w_{1m}$}
        ]
		\edge[draw=white];
        [.\node[draw=white]{$\cdots$};
            \edge[draw=white];{$\cdots$}
        ]
        \edge node [above right] {$x_n$};
        [.{$f$}
            {$w_{n1}$}
			\edge[draw=white];{$\cdots$}
			{$w_{nm}$}
        ]
    ]
    \end{tikzpicture}
\vspace{0.5em}
\end{wrapfigure}
Let $q$ be a polynomial of degree $T:=\bdeg(\PrOR_n \circ f)$ that approximates $\PrOR_n \circ f$ and is bounded outside the promise. 
Let the input variables of the \th{i} copy of $f$, for $i\in[n]$, be called $w_{i1}, w_{i2}, \ldots, w_{im}$. 
Let us also define for all $i\in[n]$, $x_i := f(w_{i1},w_{i2}\ldots,w_{im})$ to be the output of the \th{i} function $f$. Thus $q$ is a polynomial over the variables $w_{11}$ to $w_{nm}$ that approximately computes the
partial Boolean function $\PrOR(x_1,x_2,\dots,x_n)$.

We amplify the success probability of the polynomial $q$ until it agrees with $\PrOR_{n}\circ f$ to within error $c/n$ on all inputs in the promise\footnote{\label{footnote:amppoly} This can be done by composing $q$ with the univariate ``amplification polynomial'' of Buhrman \emph{et al.} (see proof of Lemma 1 in \cite{buhrman2007robust}). This is a polynomial of degree $O(\log (1/\epsilon))$ that maps $[0,1/3]$ to $[0,\epsilon]$ and $[2/3,1]$ to $[1-\epsilon,1]$, and furthermore preserves the boundedness of the polynomial, i.e. $[1/3,2/3]$ gets mapped to points in the interval $[0,1]$.},
where $c$ is the constant from \thm{multilinear_magic_poly}.
This amplified polynomial, which we call $q'$, has degree $O(T \log n)$.
Furthermore, $q'$ remains bounded in $[0,1]$ at all Boolean points outside the promise as well.

From $q'$, we can define for any $S\subseteq [n]$, a new polynomial $q'_S$ over the same set of variables $\{w_{ij}:i\in[n], j\in [m]\}$ that approximately computes the $\PrOR$ function only on those $x_i$ with $i\in S$.
The polynomial $q'_S$ is obtained from $q'$ by setting all the inputs to $f$ for which $i\notin S$ equal to the special input $w^*$ for which $f(w^*)=0$.
Thus the polynomial $q'_S$ is a polynomial over the same variables as $q'$ and has degree at most $\deg(q')$ and approximates the function $\sum_{i\in S} x_i$ when $\sum_{i\in S} x_i \in \B$.

As before, we define a polynomial $r$ in the variables $w_{11}$ to $w_{nm}$  by taking the polynomial $p$
(from \thm{multilinear_magic_poly}) over variables $y_S$ and replacing each occurrence of the variable $y_S$ with the polynomial $q'_S$. 

Then because of the robust approximation condition of $p$ in
\thm{multilinear_magic_poly} and the fact that the polynomial $q'_S$ approximates $\sum_{i\in S} x_i$ when $\sum_{i\in S} x_i \in \B$, we would like to argue that the polynomial $r$ approximates the parity of the bits $x_i$. 
We know that the polynomials $q'_S$ approximates (to error $c/n$) the function  $\sum_{i\in S} x_i$ when $\sum_{i\in S} x_i \in \B$.
Now if $q'_S$ was additionally also $c/n$-close to having Boolean output on all other inputs, then we would be done. 
This is because this perfectly fits the input of $\SCGT$, which expects the correct answer on subsets $S$  with $\sum_{i\in S} x_i \in \B$, but still expects a Boolean answer on the remaining inputs with $\sum_{i\in S} x_i > 1$.
However, the fact that the polynomial $q'_S$ may output any value in $[0,1]$ when $\sum_{i\in S} x_i > 1$ is not a problem because of the multilinearity of the polynomial $p$. We can simply view $q'_S$ as the convex combination of polynomials that always output a value $c/n$-close to being Boolean, and since $p$ is multilinear, the value of $p$ on a non-Boolean input is the same as its expected value over some Boolean inputs. Hence if $p$ works correctly when all inputs are in $\B$, it must also work correctly on inputs in $[0,1]$.

In more detail, fix an input $w \in \B^{nm}$ to $\PrOR_{n}\circ f$, and
consider the vector $v_w$ (of length $2^n$) of all the real numbers $q'_S(w)$ for $S\subseteq[n]$. 
We have $r(w)=p(v_w)$. For sets $S$ for which
$\sum_{i\in S} x_i \in \B$, the entry of $v_w$ at $S$
is $c/n$-close to $\sum_{i\in S} x_i$. Call those entries of $v_w$ the ``good'' entries.
For sets $S$ without this property, the entry of $v_w$
at $S$ is a number in $[0,1]$. We can write $v_w$ as a convex combination
of vectors that agree with $v_w$ on the good entries and have $\B$
values on the bad entries. Since $p$ is multilinear, its value
on a convex combination of vectors is a convex combination of its values
on each vector. Since $p$ robustly computes $\SCGT$, and since all
vectors in the convex combination are extremely close to points
in $\B^{2^{n}}$ that satisfy the $\SCGT$ promise with hidden string
$x_1 x_2 \dots x_{n}$, the value of $p(v_w)$ is within $1/3$
of $\Parity_{n}\circ f(w) = \bigoplus_{i=1}^n x_i$.

Since $\deg(r)\leq \deg(p)\deg(q')$, $r$ is of degree $O(\sqrt{n}T\log n)$. Thus we have
\begin{equation}
    \adeg(\XOR_n \circ f) = O(\sqrt{n}T\log n) = O(\sqrt{n} \log n\, \bdeg(\PrOR_n \circ f) ).
\end{equation}
Combining this with $\adeg(\XOR_n \circ f)) = \Omega(n \, \adeg(f))$ (\thm{Sherstovparity}), we get 
$\bdeg(\PrOR_n \circ f) = \Omega(\sqrt{n} \adeg(f)/\log n)$.
\qed

\subsection{Symmetric function composition}
\label{sec:querysym}

We can now prove the final result of this section, which generalizes the results of \sec{queryOR} from $\OR$ to arbitrary symmetric functions.

\approxdegsym*

\begin{proof}
By the result of Paturi~\cite{Pat92}, we know that the approximate degree of a symmetric function $g:\B^n \to \B$ is completely determined by the Hamming weight closest to $n/2$ where the function changes value.
More precisely, let $k$ be the closest number to $n/2$ such that $g$ gives different
values to strings of Hamming weight $k$ and $k+1$. Let us assume that $k \leq n/2$ (otherwise let $k$ be defined as $n-k$). Then Paturi showed that  $\adeg(g)=\Theta(\sqrt{nk})$. 

So to prove our result it suffices to show
$\adeg(g\circ f)=\Omega(\adeg(f)\sqrt{nk}/\log n)$.
Let us now define a partial function that captures the hard inputs of the function $g$.
Let $\PrTH^k_n:\B^n\to\Ba$ be the partial function
\begin{equation}
\PrTH^k_n(x) = \begin{cases} 0 &\mbox{if } |x|=k \\
1 & \mbox{if } |x|=k+1\\ 
* & \mbox{otherwise}
\end{cases}. 
\end{equation}
Now we have 
\begin{equation}
\label{eq:1}
    \adeg(g \circ f) \geq \bdeg(\PrTH_n^k \circ f),
\end{equation}
since $g$ contains either $\PrTH_n^k$ or its negation as a sub-function, which implies $g\circ f$ contains either $\PrTH_n^k\circ f$ or its negation as a sub-function.

Let $\PrOR_{n/2k}:\B^{n/2k}\to\Ba$ be the partial function that maps
the all-zeros string to $0$, the Hamming weight $1$ strings to $1$, and the remaining strings to $*$. Observe that $\PrTH_{2k}^k \circ \PrOR_{n/2k}$ is a sub-function of $\PrTH_n^k$, so we have
\begin{equation}
\label{eq:2}
    \bdeg(\PrTH_n^k \circ f) \geq \bdeg(\PrTH_{2k}^k \circ \PrOR_{n/2k}\circ f).
\end{equation}

By \thm{Sherstov_compose}, we know that 
\begin{align}
    \bdeg(\PrTH_{2k}^k \circ \PrOR_{n/2k}\circ f) &= \Omega(\bdeg(\PrTH_{2k}^k)^2\bdeg(\PrOR_{n/2k}\circ f)/2k) \nonumber \\ 
    &=\Omega(k\adeg(\PrOR_{n/2k}\circ f)),\label{eq:3}
\end{align} 
since $\bdeg(\PrTH_{2k}^{k})=\Theta(k)$~\cite{Pat92}.

Combining equations \eq{1}, \eq{2}, and \eq{3}, we get $\adeg(g\circ f)=\Omega(k\adeg(\PrOR_{n/2k}\circ f))$. Finally, using \thm{approxdegPrOR}, we have 
\begin{equation}
\adeg(g\circ f)=\Omega(k\sqrt{n/2k}\adeg(f)/\log n)=\Omega(\sqrt{nk}\adeg(f)/\log n)=\Omega(\adeg(g)\adeg(f)/\log n),
\end{equation} 
which proves the claim.
\end{proof}

We note that all the composition theorems in this section also extend to partial functions $f$.

\section{Approximate rank or \texorpdfstring{$\gamma_2$}{gamma 2} norm}
\label{sec:comm}

In this section we prove our composition theorems for approximate rank and approximate $\gamma_2$.
\sec{commprelim} starts with some definitions and known results in communication complexity.
In \sec{commDISJ} we prove the $\Omega(\sqrt{n})$ lower bound on the approximate rank and approximate $\gamma_2$ of the disjointness problem (\thm{disjointness}), first proved by Razborov~\cite{Raz03}. We also show that the same proof technique yields an $\OR$-composition theorem in the special case that the function $F$ contains an all-zeros row or column (\cor{clean_compose}).
In \sec{commOR} we prove a general $\OR$-composition theorem for approximate $\gamma_2$ (\thm{approxrank}). We then extend this composition result to the $\PrOR$ function in \sec{commPrOR}, and finally to arbitrary symmetric functions in \sec{commsym}.

\subsection{Preliminaries}
\label{sec:commprelim}

In this section we describe the setting of communication complexity and define the complexity measures we are interested in, approximate rank and approximate $\gamma_2$ norm, their properties, and the relationships between these measures. We then end with some results related to quantum communication and query complexity that we use. 

\para{Rank and $\gamma_2$.}
In communication complexity, we have a known function
$F:\X\times\Y\to\B$ for some finite sets $\X$
and $\Y$. Two players, typically called Alice and Bob, receive
inputs $x\in\X$ and $y\in\Y$ respectively, and their
goal is to compute $F(x,y)$ using as little communication as possible.

We identify a communication function $F$ with a
\emph{sign matrix} whose
rows are indexed by $\X$ and whose columns are indexed by $\Y$,
and where the entry corresponding to $(x,y)$ is $(-1)^{1-F(x,y)}\in\{-1,1\}$.
Note that this matrix completely specifies the communication problem.

We are interested in two complexity measures of matrices, the rank of a matrix, and the $\gamma_2$ norm of a matrix. The latter is defined by
\begin{equation}
\gamma_2(A):=\min_{B,C:BC=A}\|B\|_{\row}\|C\|_{\col},
\end{equation}
where $\|B\|_{\row}$ and $\|C\|_{\col}$ denote the largest $\ell_2$
norm of a row of $B$ and the largest $\ell_2$ norm of a column of $C$,
respectively. 

The $\gamma_2$ norm is in many ways similar to rank and has several useful properties \cite{She12,LSS08}, which we list here and compare with rank. The following hold for all matrices $A$ and $B$:
\begin{enumerate}
\item $\gamma_2(A+B)\le\gamma_2(A)+\gamma_2(B)$ and
$\rank(A+B)\le\rank(A)+\rank(B)$. \label{item:subadd}
\item For any scalar $\lambda\ne 0$,
$\gamma_2(\lambda A)=|\lambda|\gamma_2(A)$ and
$\rank(\lambda A)=\rank(A)$.
\item If $B$ is a submatrix of $A$, $\gamma_2(B)\le\gamma_2(A)$ and
$\rank(B)\le\rank(A)$.
\item $\gamma_2(A)$ and $\rank(A)$ are invariant under duplicating, rearranging,
or negating rows or columns of $A$.
\item $\gamma_2(A\otimes B)= \gamma_2(A)\gamma_2(B)$ and 
$\rank(A\otimes B)=\rank(A)\rank(B)$.
\item $\gamma_2(A\circ B)\le \gamma_2(A)\gamma_2(B)$ and
$\rank(A\circ B)\le\rank(A)\rank(B)$. \label{item:submulthad}
\item $\gamma_2(J)=1$ and $\rank(J)=1$.
\item $\|A\|_\infty\le\gamma_2(A)\le\|A\|_{\infty}\sqrt{\rank(A)}$.\label{item:rel}
\end{enumerate}
In the above, $A\otimes B$ denotes the Kronecker (tensor) product,
$A\circ B$ denotes the Hadamard product, $J$ denotes the all-ones
matrix, and $\|A\|_{\infty}$ denotes the maximum absolute value
of an entry of $A$. Since we will only deal with matrices
whose entries are bounded in $[-1,1]$, item \ref{item:rel} will
give us $\gamma_2(A)\le\sqrt{\rank(A)}$.

\para{Approximate rank and $\gamma_2$.}
We are also interested in matrices that approximate the sign matrix of a function $F:\X\times\Y\to\B$, where
we say that a matrix $A$ of real numbers approximates
(the sign matrix of) $F$ to error $\epsilon$
if $|A_{xy}-(-1)^{1-F(x,y)}|\le \epsilon$ for all $(x,y)$
and $|A_{xy}|\le 1$ for all $(x,y)$.
If $A$ approximates $F$ to error $\epsilon$, we denote this by
$A\approx_\epsilon F$.

If $M:\mathbb{R}^{\X\times\Y}\to\mathbb{R}^+$
is a function on matrices, we define the $\epsilon$-approximate
version of $M$ as
\begin{equation}
M_\epsilon(F):=\inf_{A:A\approx_\epsilon F} M(A).
\end{equation}
We are particularly interested in the case where $M(A)=\rank(A)$ and $M(A)=\gamma_2(A)$.
For both gamma $2$ and rank, the infimum above
is achieved and can be replaced by a minimum.

We now list some useful properties of approximate rank and $\gamma_2$
shown in the literature. The first result shows that the exact value of $\epsilon$ chosen is not really too important.

\begin{lemma}[Amplification]\label{lem:amp}
For any $F:\X\times\Y\to\B$ and any
$0<\epsilon'<\epsilon<1$, there exist constants $c,d>0$ that only depend on $\epsilon$ and $\epsilon'$, such that
\begin{align}
\log\rank_\epsilon(F) &\le \log\rank_{\epsilon'}(F) \le c \, (\log\rank_\epsilon(F)),\textrm{ and} \\
\log\gamma_{2,\epsilon}(F) &\le \log\gamma_{2,\epsilon'}(F) \le d \, (\log\gamma_{2,\epsilon}(F)+1).
\end{align}
\end{lemma}

\begin{proof}
The left two inequalities follow from the definition
of $\epsilon$-approximation. For the right inequalities,
the idea is to find a univariate polynomial $p$ that
maps the range $[-1,-1+\epsilon]$ to $[-1,-1+\epsilon']$
and maps the range $[1-\epsilon,1]$ to $[1-\epsilon',1]$.
(See footnote \ref{footnote:amppoly} for more information.)
Then take a matrix $A$ that approximates $F$ to error $\epsilon$
and apply $p$ to $A$, using the Hadamard product for the matrix
product. The result is that $p$ gets applied to each entry of $A$,
producing a matrix $A'$ that approximates $F$ to error $\epsilon'$.
Finally, the rank or gamma $2$ norm of $A'$ can be upper bounded
by appealing to sub-additivity (item \ref{item:subadd} above)
and to sub-multiplicativity under the Hadamard product (item \ref{item:submulthad} above).
The rank and $\gamma_2$ will increase polynomially, and taking logarithms on both sides will
yield a result of the form $\log\rank_{\epsilon'}(F) \le c \, (\log\rank_\epsilon(F)+1)$.

To remove the additive constant, we need to worry about
what happens when $\rank_{\epsilon}(F)\le 1$ or
$\gamma_{2,\epsilon}(F)\le 1$. For rank, this can only happen
if all rows of the approximating matrix are constant multiples
of each other; but in that case, even the original sign matrix of $F$
is rank $1$. 
It follows that $\rank_{\epsilon'}(F)=1$ as well, so
a multiplicative constant is sufficient.
For gamma $2$, however, this is not the case; we therefore
have to lose an additive constant as well as a multiplicative constant.
\end{proof}

In light of \lem{amp}, we do not need to worry about the exact error
$\epsilon$ so long as it is bounded between $0$ and $1$. We will
pick $\epsilon=2/3$ to be the default, and use
$\agamma_2(F)$ and $\arank(F)$ to denote the approximate gamma $2$
and rank of $F$ to error $2/3$. (Note that since
we are approximating the values $\{-1,1\}$, an error of $2/3$
is analogous to the usual $1/3$ error for approximating
$\B$ values).

While these two measures seem somewhat different, Lee and Shraibman~ \cite{LS09a} showed that $\log\agamma_2(F)$ and $\alogrank(F)$
are actually closely related.

\begin{theorem}[Lee and Shraibman]\label{thm:gamma_rank}
Let $F:\X\times\Y\to\B$ be a communication problem and define $|F|:=|\X||\Y|$. Then
\begin{equation}2\log\agamma_2(F)\le\alogrank(F)
\le 6\log\agamma_2(F)+O(\log\log|F|).\end{equation}
\end{theorem}

Finally, we will need the following result of Sherstov~\cite[Theorem 4.18]{She12} that says that taking the $\XOR$ of $n$ copies of a function $F$ increases its approximate $\gamma_2$ by a factor of $n$ as long as $\agamma_2(F)$ is larger than some universal constant.

\begin{theorem}[Sherstov]\label{thm:xor}
There is a universal constant $c$ such that for all $F:\X\times\Y\to\B$,
\begin{equation}\log\agamma_2(\XOR_n\circ F)=\Omega(n(\log\agamma_2(F)-c)).
\end{equation}
\end{theorem}

\para{Quantum results.}
Next, we note the relevance of $\agamma_2$ for lower bounding
quantum communication complexity. The relation between these measures follows from the following theorem of Linial and Shraibman \cite{LS09b}.

\begin{theorem}[Linial and Shraibman]
Let $\Pi$ be a $T$-qubit quantum communication protocol with arbitrary preshared entanglement
on inputs from $\X$ and $\Y$.
Let $P\in\mathbb{R}^{\X\times\Y}$
be the matrix of acceptance probabilities of $\Pi$; i.e.,
$P_{xy}$ is the probability that $\Pi$ accepts when Alice and Bob
start with input $(x,y)\in\X\times\Y$. Then $\gamma_2(P)\le 2^T$.
\end{theorem}

This theorem naturally yields the following corollary.

\begin{corollary}[Linial and Shraibman]
Let $F:\X\times\Y\to\B$ be a communication problem.
Then $\QCC(F)\ge\log\agamma_2(F)-1$.
\end{corollary}

\begin{proof}
Let $\Pi$ be a communication protocol computing $F$ to error
$1/3$ using $\QCC(F)$ qubits of communication.
Let $P$ be the matrix of its acceptance probabilities.
Note that by negating the answers at the end of $\Pi$,
we get a new protocol $\Pi'$ with acceptance probability
matrix $J-P$ that has the same communication cost.
Then
$\gamma_2(2P-J)\le\gamma_2(P)+\gamma_2(J-P)\le 2\cdot 2^{\QCC(F)}$
so $\QCC(F)\ge \log\gamma_2(2P-J)-1$.
Finally, observe that the entries of $2P-J$ lie in $[-1,1]$ and
approximate the sign matrix of $F$ to error $2/3$.
Thus $\gamma_2(2P-J)$ upper bounds
$\agamma_2(F)$, and the result follows.
\end{proof}

We will additionally need the existence of a low-degree robust polynomial for combinatorial group testing described in \sec{query}, which we restate here
in the $\{-1,1\}$ basis. Note that in the theorem below, we abuse notation
by assuming that the functions $\XOR_n$ and $\CGT_{2^n}$ take
input bits in $\{-1,1\}$ and give output bits in $\{-1,1\}$.

\begin{theorem}\label{thm:magic_poly_pm}
There is a real polynomial $p$ of degree $O(\sqrt{n})$ acting on $2^n$
variables $\{y_S\}_{S\subseteq[n]}$ such that for any input
$y\in\{-1,1\}^{2^n}$ with $\XOR_n\circ\CGT_{2^n}(y)\ne *$ and any
$\Delta\in[-2/3,2/3]^{2^n}$,
\begin{equation}|p(y+\Delta)-\XOR_n\circ\CGT_{2^n}(y)|\le 2/3,\end{equation}
and for all $y\in\{-1,1\}^{2^n}$, $p(y)\in[-1,1]$.
Moreover, the sum of absolute values of coefficients of $p$
is at most $(2^n)^{O(\sqrt{n})}\leq 2^{O(n^{1.5})}$.
\end{theorem}

\begin{proof}
This follows immediately from \thm{magic_poly} by a variable substitution
of $2y_S-1$ for each variable $y_S$,
mapping $\B$ inputs to $\{-1,1\}$ inputs.
We also apply this linear mapping to the output of
the polynomial.
The condition on the sum of absolute values of the coefficients
follows from the fact that a robustified multilinear polynomial of degree
$d$ on $N$ variables has sum of absolute coefficients at most $N^{O(d)}$,
which we prove in \app{coeff_bounded}
(see \thm{coeff_bounded}).
\end{proof}

\subsection{Disjointness lower bound}
\label{sec:commDISJ}

We start by reproving the disjointness lower bound for
approximate log rank, first proven by Razborov \cite{Raz03}.
We note that the techniques used in our proof are quite different
from the previous proofs by Razborov \cite{Raz03}
and Sherstov \cite{She11}. Our result for disjointness, \thm{disjointness}, is restated below:

\disjointness*

\begin{proof}
Let $G=\DISJ_n=\OR_n\circ\AND^{\CC}$, where $\AND^{\CC}$ is the
communication function where Alice and Bob each get one bit,
and they must compute the $\AND$ of their bits. In later
theorems, we will replace $\AND^{\CC}$ by an arbitrary
communication function $F$.

Let $A$ be the optimal approximating matrix for $G$,
so $A$ approximates $G$ and $\rank(A)=\arank(G)$.
For each set $S\subseteq [n]$, let $G_S$ be the disjointness
problem restricted to that set of bits. That is, in $G_S$,
Alice and Bob still get $n$-bit strings as input, but now
they ignore the bits whose position is outside of $S$, and must
compute disjointness of the bits whose position is in $S$.
We use $A$ to construct an approximating matrix for $G_S$
for every $S\subseteq [n]$. For a fixed $S$, we define $(A_S)_{xy}$
to be $A_{xy^S}$, where $y^S$ is the string $y$ will all the bits
outside of $S$ replaced with $0$; that is, $y^S_i=y_i$ for $i\in S$,
and $y^S_i=0$ for $i\notin S$. Since $A$ approximates
$\OR_n\circ \AND^{\CC}$ entrywise, each entry $(A_S)_{xy}$ approximates
the disjointness function on $x$ and $y$ restricted to the subset
$S$; thus $A_S$ approximates $G_S$.
Moreover, for each $S\subseteq[n]$, the matrix $A_S$
is simply the matrix $A$ with some columns replaced by others;
that is, you can get from $A$ to $A_S$ by deleting and duplicating
columns. Hence $\rank(A_S)\le\rank(A)=\arank(G)$.

Now, from \thm{magic_poly_pm}, there is a polynomial $p$ of degree
$d=O(\sqrt{n})$ acting on $2^n$ variables $\{z_S\}_{S\subseteq [n]}$
that is bounded between $-1$ and $1$ on all inputs in $\{-1,1\}^{2^n}$.
Moreover, this polynomial $p$ has the property that when
the variables $z_S$ approximate the $\OR$ of the subset $S$
of $n$ fixed bits $w_1,w_2,\dots,w_n$, the polynomial outputs
an approximation of the parity of those bits.

We apply $p$ to the matrices $A_S$, plugging in $z_S=A_S$
for all $S\subseteq[n]$. We use the Hadamard product as the matrix
product in this polynomial evaluation. This gives a matrix
$B$ with the property that $B_{xy}$ is the result of evaluating
$p$ on $z_S=(A_S)_{xy}$. Since $(A_S)_{xy}$ approximates
$G_S(x,y)$, which is the $\OR$ of the bitwise $\AND^{\CC}$
of the positions in $S$, we conclude that $B_{xy}$
approximates the parity of the bitwise $\AND^{\CC}$
of the strings $x$ and $y$. In other words, $B$ is a matrix
approximating the function $\IP_n$, which means
$\rank(B)\ge\arank(\IP_n)$.

We now upper bound the rank of $B$ using the rank of $A$.
We can write
\begin{equation}B=\sum_{m}\alpha_m\prod_{S\in m} A_S,\end{equation}
where $m$ are the monomials of $p$ (represented as sets of size
at most $d$ of subsets $S\subseteq[n]$), $\alpha_m$
are the coefficients, and the product
refers to the Hadamard product of the matrices $A_S$.
Consider a single term of the sum, and let $T$ be the resulting
matrix product in that term. Then each column of $T$ is the Hadamard
product of at most $d$ columns of the $A_S$ matrices. But recall that
all the columns of the $A_S$ matrices are columns of $A$;
therefore, each column of $T$ is the Hadamard product of at most $d$
columns of $A$. If $A'$ denotes the matrix $A$ with an added
all-ones column, then each column of $T$ is the Hadamard product
of exactly $d$ columns of $A'$.

Let $C$ be the matrix consisting of all $d$-wise Hadamard products
of columns of $A'$ (in any order). Then $C$ is a submatrix
of $(A')^{\otimes d}$, the $d$-fold tensor product of $A'$.
We saw that each term $T$ in the sum has columns from $C$.
The columns of the final sum $B$ are therefore linear combinations
of the columns of $C$, meaning the column span of $B$
is a subspace of the column span of $C$. We conclude
\begin{equation}\arank(\IP_n)\le\rank(B)\le\rank(C)\le\rank((A')^{\otimes d})= \rank(A')^d\le (1+\rank(A))^d= (1+\arank(G))^d.\end{equation}
Taking logarithms and using $d=O(\sqrt{n})$, we get
\begin{equation}\alogrank(\DISJ_n))=\Omega(\alogrank(\IP_n)/\sqrt{n}),\end{equation}
where we used the easy-to-check fact that $\arank(\DISJ_n)>1$
for all $n\ge 1$ to replace the additive $1$ term with
a multiplicative factor.
Finally, it is known that $\alogrank(\IP_n)=\Omega(n)$ (e.g., it follows from Forster's lower bound on the sign rank of inner product~\cite{For02}). 
A more elementary way of showing this is to lower bound approximate rank using discrepancy~\cite{LS09b} and then show that the discrepancy of inner product is $\Omega(n)$~\cite{KN06}.
The result for $\gamma_2(\DISJ_n)$ follows from \thm{gamma_rank}.
\end{proof}

The disjointness lower bound lets us conclude that
for any communication problem $F$ containing $\AND^{\CC}$
as a sub-problem, the approximate gamma $2$ norm of
$\OR_n\circ F$ is at least $\Omega(\sqrt{n})$
(so in particular, it goes to infinity as $n\to\infty$).
We also observe that $\NOTEQ$, the negation of the equality
problem, is essentially the only communication problem
which does not contain $\AND$ as a subproblem; it is therefore
the only communication task that does not grow like $\Omega(\sqrt{n})$
when we take the $\OR$ of $n$ copies of it (indeed,
the $\OR$ of $\NOTEQ$ is a larger $\NOTEQ$ instance, which
can be solved with bounded error using constant communication
complexity).

\begin{corollary}\label{cor:NOTEQ}
Let $F:\X\times\Y\to\B$ be a communication problem that is not equivalent (i.e. identical up to permuting rows/columns and deleting repeated rows/columns of the communication matrix) to $\NOTEQ$.
Then $\log\agamma_2(\OR_n\circ F)=\Omega(\sqrt{n})$.
\end{corollary}

\begin{proof}
Remove rows and columns from $F$
until all rows and columns are distinct; this doesn't
change the communication problem. Suppose $F$ has a row
or column with at least two zeros. Without loss of generality,
let that be a row, and name it $x_1\in\X$.
Then there are $y_1,y_2\in\Y$ with
$F(x_1,y_1)=F(x_1,y_2)=0$. Since the columns $y_1$ and $y_2$
are distinct, there is some row $x_2$ where they disagree.
By exchanging $y_1$ and $y_2$ if necessary, we get
$F(x_2,y_1)=0$ and $F(x_2,y_2)=1$, which means
that on the inputs $\{x_1,x_2\}\times\{y_1,y_2\}$, the
function $F$ acts like $\AND^{\cc}$. In particular,
if we restrict the sign matrix of $F$ to the rows and columns
corresponding to these inputs, we get the sign matrix of $\AND^{\cc}$.
It follows that if $A$ is an approximating matrix for $\OR_n\circ F$,
a submatrix of $A$ is an approximating matrix for
$\OR_n\circ\AND^{\cc}=\DISJ_n$. Since gamma $2$ is nonincreasing
under submatrices,
\begin{equation}
\log\agamma_2(\OR_n\circ F)\ge\log\agamma_2(\DISJ_n)=\Omega(\sqrt{n}),
\end{equation}
using \thm{disjointness}.

If $F$ does not have a row or column with
at least two zeros, sort its rows by the position of the zero
in that row. The resulting matrix is simply $J-I$, with the possible
addition of an all-ones row or column. But $J-I$ is the $\NOTEQ$
function. An addition of an all ones row
is equivalent to adding an element to $\X$ that's not in $\Y$
and solving $\NOTEQ$ on $(\X,\Y)$; an addition of an all ones
column is similarly equivalent to adding a new element to $\Y$.
In all cases, we get that $F$ is equivalent to a version of $\NOTEQ$.
\end{proof}

Observe that since the sign matrix of $\NOTEQ$ is $J-2I$
(or a submatrix of $J-2I$ in the case that the sets $\X$ and $\Y$
are not identical), we have
$\gamma_2(\NOTEQ)\le\gamma_2(J)+2\gamma_2(I)=3$.
One approximating matrix for any sign matrix $A$ (to error $2/3$)
is simply $(1/3)A$; hence
$\agamma_2(\NOTEQ)\le (1/3)\gamma_2(\NOTEQ)\le 1$,
and $\log\agamma_2(\NOTEQ)\le 0$. This conveniently allows
us to hope that
$\log\agamma_2(\OR_n\circ F)=\Omega(\sqrt{n}\log\agamma_2(F))$
for all functions $F$. We don't quite manage to prove this,
but we get quite close in \cor{clean_compose} and \thm{polylog_gamma}.

Indeed, we note that the proof
of \thm{disjointness} actually proved the following generalization,
which lets us reduce $\OR$-composition lower bounds to
$\XOR$-composition lower bounds if $F$ has an all-zero
row or column.

\begin{theorem}\label{thm:all_zero}
	Let $F:\X\times\Y\to\B$ be a communication function.
	Suppose $F$ has an all-zero row or all-zero column
	(that is, an input $a\in\X$ such that $F(a,y)=0$ for all $y\in\Y$,
	or an input $b\in\Y$ such that $F(x,b)=0$ for all $b\in\X$). Then
	\begin{equation}\alogrank(\OR_n\circ F)=\Omega(\alogrank(\XOR_n\circ F)/\sqrt{n}).\end{equation}
\end{theorem}

\begin{proof}
	The proof is identical to the proof of \thm{disjointness}.
	We use $F$ in place of $\AND^{\CC}$, and set $G=\OR_n\circ F$.
	For $S\subseteq[n]$,
	we define $G_S$ to by the $\OR$ of the $F$-inputs in the set $S$
	(out of the $n$ given $F$-inputs).
	We need an all-zero column in order to define $A_S$ approximating $G_S$
	by deleting and duplicating columns of $A$;
	we do this by setting $(A_S)_{xy}:=A_{xy^S}$, as before, but
	we must be careful when defining $y^S$. Note that $y$
	is a Bob-input to $n$ copies of $F$, say $y_1,y_2,\dots,y_n$,
	and we wish to ``zero out''
	the copies that are outside of $S$. We can do this by setting
	$y^S_i=b$ for all $i\notin S$, and $y^S_i=y_i$ for all $i\in S$.
	The rest of the argument proceeds as before.
\end{proof}

We can then use \thm{xor} to immediately get the following
$\OR$-composition theorem.

\begin{corollary}\label{cor:clean_compose}
	Let $F:\X\times\Y\to\B$ be a communication problem with an all-zero row
	or column and let $|F|=|\X||\Y|$. Then we have
	\begin{align}
	&\log\arank(\OR_n\circ F)=\Omega(\sqrt{n}\log\agamma_2(F)), \textrm{ and}\\    
	&\log\agamma_2(\OR_n\circ F)=\Omega(\sqrt{n}\log\agamma_2(F))-O(\log\log|F|).
	\end{align}
\end{corollary}

\begin{proof}
	By \thm{all_zero}, we have
	\begin{equation}\alogrank(\OR_n\circ F)=\Omega(\alogrank(\XOR_n\circ F)/\sqrt{n})
	=\Omega(\log\agamma_2(\XOR_n\circ F)/\sqrt{n}).\end{equation}
	Applying \thm{xor} to this equation, we get
	\begin{equation}\alogrank(\OR_n\circ F)=\Omega(\sqrt{n}(\log\agamma_2(F)-c)).\end{equation}
	We now get rid of the $c$ term. If $\log\agamma_2(F)>2c$, we can
	of course remove the $c$ term and collapse the loss into
	the $\Omega$-notation. So suppose $\log\agamma_2(F)\le 2c$.
	
	Without loss of generality,
	suppose $F$ has an all-zero row, so $F(a,y)=0$ for all $y\in\Y$.
	If $F$ has any other row that has both a $1$ and a $0$,
	then $F$ has $\AND^{\cc}$ as a sub-problem. If not, then
	each row of $F$ is either all-ones or all-zeros,
	meaning the sign matrix of $F$
	has rank $1$. In other words, either $\gamma_2(F)\le 1$
	(which implies $\log\agamma_2(F)\le 0$), or else $F$
	contains $\AND^{\cc}$ as a sub-problem. In the former case,
	the statement $\alogrank(\OR_n\circ F)=\Omega(\sqrt{n}\log\agamma_2(F))$
	is trivial. In the latter case, \cor{NOTEQ} applies, and we get
	\begin{equation}\alogrank(\OR_n\circ F)=\Omega(\sqrt{n})=\Omega(c\sqrt{n})
	=\Omega(\sqrt{n}\log\agamma_2(F)).\end{equation}
	Finally, we appeal to \thm{gamma_rank} to get the desired result
	(noting that $\log\log|\OR_n\circ F|=\log n + \log\log|F|$).
\end{proof}

Although we have a general $\OR$ composition result for functions $F$ with an all-zero row or column, this does not cover all functions.
We modify the proof of \thm{disjointness} in the next section to cover all functions $F$. 

\subsection{\texorpdfstring{$\OR$}{OR} composition}
\label{sec:commOR}

We now prove our $\OR$ composition theorems for approximate
rank and approximate gamma $2$ norm for general functions.
We start with the following useful lemma on the composition
of a polynomial with matrices.

\begin{lemma}\label{lem:poly_matrix}
Let $p$ be a real  polynomial on $N$ variables.
Let $A_1,A_2,\dots,A_N$ be real matrices of the same size.
Let the matrix $B=p(A_1,A_2,\dots,A_N)$ be the result of plugging in
the matrices into $p$ and using the Hadamard product for the matrix product.
Then
\begin{equation}\log\gamma_2(B)=O\left(\deg(p)\max_{i\in[N]}\gamma_2(A_i) + \log C\right),\end{equation}
where $C$ is the sum of the absolute values of coefficients of $p$.
If $p$ is bounded inside $[-1,1]$ for all inputs in $\{-1,1\}^N$, we have $\log C=O(\deg(p)\log N)$.
\end{lemma}

\begin{proof}
Write $p=\sum_{m}\alpha_m\prod_{i\in m}x_i$
where $m$ ranges over monomials, each $\alpha_m$ is a real
coefficient, and $x_i$ are the variables.
 We have
 \begin{equation}B=\sum_{m}\alpha_m\prod_{i\in m}A_i,\end{equation}
 where the product is the Hadamard product.
 Recall that $\gamma_2$ satisfies $\gamma_2(X+Y)\le\gamma_2(X)+\gamma_2(Y)$,
 $\gamma_2(X\circ Y)\le\gamma_2(X)\gamma_2(Y)$, and
 $\gamma_2(\alpha X)=|\alpha|\gamma_2(X)$.
 Setting $d=\deg(p)$ and $M=\max_{i\in[N]}\gamma_2(A_i)$,
 we immediately get $\gamma_2(B)\le M^d\sum_{m}|\alpha_m|$. Taking logarithms on both sides yields the desired upper bound on $\log \gamma_2(B)$. The upper bound on the sum of absolute values of coefficients of a bounded polynomial is proved in \thm{coeff_bounded}.
\end{proof}

The above lemma, combined with the polynomial we get from
Belovs's algorithm, will turn any approximating matrix for
$\OR\circ F$ into an approximating matrix for $\XOR\circ F$
whose gamma $2$ norm is not too much larger. We can then
use Sherstov's $\XOR$ lemma to get a lower bound on
the gamma $2$ norm of $\OR\circ F$. This gives us the following result.

\begin{lemma}\label{lem:gamma_or}
Let $F:\X\times\Y\to\B$ be a communication problem. Then
\begin{equation}\log\agamma_2(\OR_n\circ F)\ge\Omega(\sqrt{n}\log\agamma_2(F))-O(n).\end{equation}
\end{lemma}

\begin{proof}
Let $c$ be the constant in \thm{xor}.
If $\log\agamma_2(F)<2+c$, the statement is trivial:
we just need an $\Omega(\sqrt{n})$ lower bound on
$\log\agamma_2(\OR_n\circ F)$, which follows from \cor{NOTEQ}
unless $F$ is equivalent to $\NOTEQ$ (and in the latter
case, $\log\agamma_2(F)\le 0$ anyway, so the theorem holds
trivially as the right hand side is less than $0$).
Therefore, suppose $\log\agamma_2(F)\ge 2+c$.
In particular, $F$ is not constant,
so let $(a,b)$ be an input to $F$ such that $F(a,b)=0$.

Let $G=\OR_n\circ F$. For each $S\subseteq[n]$, let $G_S$ be the
$\OR$ of the inputs in the set $S$, as in the proof of
\thm{disjointness}. Let $A$ be the best approximating matrix for $G$, so
$\gamma_2(A)=\agamma_2(G)$, $A_{xy}$
is within $2/3$ of $(-1)^{1-G(x,y)}$ if $(x,y)$ is in the promise of $G$,
and $A_{xy}\in[-1,1]$. For each $S\subseteq[n]$, we construct
the matrix $A_S$ as follows. First, consider the submatrix $A'_S$
of $A$ we get by restricting $A$ to rows $x$ and columns $y$
satisfying $x_i=a$ and $y_i=b$ for each $i\notin S$.
Next, let $J$ denote the all-ones matrix with the same dimensions
as $F$, and tensor product $A'_S$ with $J$ $n-|S|$ times (once for
each $i\notin S$). Intuitively, this adds $n-|S|$ inputs that
are always ignored. Finally, rearrange the rows and columns of the
resulting matrix so that the ignored inputs are in the positions
$i\notin S$. This final matrix is $A_S$.

It is not hard to see that $A_S$ approximates $G_S$ to the usual
error $2/3$ for all $S$. Moreover, since $\gamma_2$
is invariant under rearrangements of rows and columns and under
tensor products with $J$, and since it is nonincreasing under restriction
to a submatrix, we have $\gamma_2(A_S)\le\gamma_2(A)$ for all $S$.

From \thm{magic_poly_pm}, we have a polynomial $p$ of degree $d=O(\sqrt{n})$
on $2^n$ variables $\{z_S\}_{S\subseteq[n]}$ such that
if there are $n$ bits $w_1,w_2,\dots,w_n$ and if each
$z_S$ is instantiated to be within $1/3$ of $\vee_{i\in S} w_i$,
then $p(z)$ evaluates to within $1/3$ of $\oplus_{i\in[n]} w_i$.
We represent $p$ in the $\{-1,1\}$ basis, so that the inputs
and outputs of $p$ approximate $-1$ or $1$ instead of $0$ or $1$.
Each monomial $m$ of $p$ is a set of size at most $d$
of subsets $S\subseteq[n]$. 

We use \lem{poly_matrix} to plug the matrices $A_S$ into the polynomial
$p$. This gives us a matrix $B$ such that $B_{xy}$ is the result
of applying $p$ to $\{(A_S)_{xy}\}_S$. As we saw in the proof of
\thm{disjointness}, this means the matrix $B$ approximates $\XOR\circ F$
to error $2/3$. \thm{magic_poly_pm} also gives us an upper bound
on the sum of absolute coefficients. We therefore get
\begin{equation}\log\agamma_2(\XOR\circ F)\le\log\gamma_2(B)
=O(\sqrt{n}\log\agamma_2(\OR\circ F)+n^{1.5}).\end{equation}
By \thm{xor}, we have
$\log\agamma_2(\XOR_n\circ F)=\Omega(n\log\agamma_2(F))$,
from which we get $\log\agamma_2(G)\ge\Omega(\sqrt{n}\log\agamma_2(F))-O(n)$.
\end{proof}

Finally, we apply \lem{gamma_or} recursively to turn the additive
$O(n)$ loss into a multiplicative $\polylog n$ loss. The idea is as follows: for functions $F$
such that $\log\agamma_2(F)$ is sufficiently large (larger than $c\cdot \sqrt{n}$ for some constant $c$),
we can get rid of the additive $O(n)$ factor
simply because $\sqrt{n}\log\agamma_2(F)-O(n)=\Omega(\sqrt{n}\log\agamma_2(F))$.
So the only trouble is with functions $F$ that whose
approximate gamma $2$ norm is small compared to $n$
(the size of the desired $\OR$). The key insight
is to use the fact that $\OR_n$
is the same function as
$\OR_{\sqrt{n}}\circ\OR_{\sqrt{n}}$, and therefore
$\OR_n\circ F=\OR_{\sqrt{n}}\circ(\OR_{\sqrt{n}}\circ F)$.
The inner function $\OR_{\sqrt{n}}\circ F$ will then
intuitively have
large approximate gamma $2$ norm compared with
the outer $\OR$ of size $\sqrt{n}$, so we can use
\lem{gamma_or} on the outer composition.
Finally, to show that the approximate gamma $2$ norm
of the inner function $\OR_{\sqrt{n}}\circ F$
is as large as suspected, we recurse the argument.

We prove this recursive argument formally in the
following technical lemma,
from which the desired result will
directly follow. We phrase this lemma in a very
general setting (making no direct reference to $\OR$
or to communication complexity or gamma $2$ norm)
because in \sec{commPrOR} we will need to apply this lemma
to the $\PrOR$ function instead of the $\OR$
function (and it will have to compose with partial
functions instead of total functions).

\begin{lemma}\label{lem:recurse}
Let $\{\alpha_n\}_{n\in\mathbb{N}}$ be a family
of functions $\alpha_n:A\to A$ on some domain $A$,
with the property that $\alpha_n\circ\alpha_m=\alpha_{nm}$
for all $n,m\in N$. Let $M:A\to\mathbb{R}$
be a function satisfying
$M(\alpha_n(G))\ge M(\alpha_m(G))$
for all $G\in A$ whenever $n\ge m$.

Fix a positive integer $n$
and a domain element $F\in A$.
Let $S:=\{\alpha_k(F):k\le n\}$.
Let $a$, $b$ and $c$
be positive integers such that
for all $k\le n$ and $G\in S$, we have
\begin{equation}
M(\alpha_k(G))\ge\frac{\sqrt{k}}{c}-d\qquad
\mbox{and}\qquad
M(\alpha_k(G))\ge\frac{\sqrt{k}M(G)}{a}-bk.
\end{equation}
Then we also have
\begin{equation}
M(\alpha_n\circ F)
\ge\frac{\sqrt{n}M(F)}{16cba^4+5a(\log n)^{2\log 2a}}-d.
\end{equation}
\end{lemma}

Before proving this lemma, we will show how it
implies \thm{approxrank},
which we restate here for convenience.

\approxrank*

\begin{proof}
First, we observe that when $F$ is equivalent
to $\NOTEQ$, we have $\log\agamma_2(F)\le 0$ so the
theorem follows trivially. We focus on the case
where $F$ is not equivalent to $\NOTEQ$.

We use \lem{recurse} with
$A=\{\OR_k\circ F:k\in\mathbb{N}\}$,
$\alpha_k$ being the composition-with-$\OR_k$
operator (so $\alpha_k(G)=\OR_k\circ G$ for all
$G\in A$),
and $M=\log\agamma_2(\cdot)$.
It is clear that $\alpha_n\circ\alpha_m=\alpha_{nm}$
(from the associativity of composition).
The property that $M(\alpha_n(G))\ge M(\alpha_m(G))$
when $n\ge m$, that is, the property that
$\log\agamma_2(\OR_n\circ G)
\ge\log\agamma_2(\OR_m\circ G)$,
follows from the fact that the sign matrix
of $\OR_m\circ G$ is a submatrix of the sign
matrix of $\OR_n\circ G$, and $\gamma_2$ is non-increasing
under submatrices (and hence so is $\agamma_2$).

From \cor{NOTEQ}, since $F$ is not equivalent to
$\NOTEQ$, we have
$\log\agamma_2(\OR_n\circ F)=\Omega(\sqrt{n})$.
Hence we can pick $c$ and $d$ in \lem{recurse}
to be universal constants independent of $F$ and $n$.
Moreover, from \lem{gamma_or}, we can also pick
$a$ and $b$ to be universal constants independent of
$F$ and $n$. Since $a$, $b$, $c$, and $d$
are all constants, \lem{recurse} gives us
$\log\agamma_2(\OR_n\circ F)=\tOmega(\sqrt{n}\log\agamma_2(F))$, as desired.
\end{proof}

We now prove the lemma.
\begin{proof}[Proof of \protect{\lem{recurse}}]
We compare $M(F)$ to $b(2a)^4$.
If $M(F)<b(2a)^4$, then
\begin{equation}
M(\alpha_n(F))\ge\frac{\sqrt{n}}{c}-d
\ge\frac{\sqrt{n}M(F)}{cb(2a)^4}-d
\end{equation}
from which the desired result follows.

For the rest of the proof, we assume $M(F)\ge b(2a)^4$.
We prove by induction on $t$ that
$M(\alpha_{c_t}(F))\ge \sqrt{c_t}M(F)/(2a)^t$,
where $c_t$ is defined by $c_{t+1}=c_t^2/(2a)^{2t}$
and $c_1=(2a)^6$.
The base case follows from
\begin{equation}M(\alpha_{(2a)^6}(F))
\ge(1/a)(2a)^3 M(F)-b(2a)^6
\ge (1/2a)(2a)^3 M(F),\end{equation}
where we used $M(F)\ge b(2a)^4$.

For the induction step, we have
\begin{equation}
M(\alpha_{c_{t+1}}(F))=
M(\alpha_{c_t/(2a)^{2t}}(\alpha_{c_t}(F)))
\ge (1/a)\sqrt{c_t/(2a)^{2t}}M(\alpha_{c_t}(F))
-bc_t/(2a)^{2t}
\end{equation}
where we used the property in the lemma of the constants
$a$ and $b$ to remove the outer function
$\alpha_{c_t/(2a)^{2t}}$.
Next, apply the induction hypothesis for $c_t$
to get $M(\alpha_{c_t}(F))\ge\sqrt{c_t}M(F)/(2a)^t$. This
gives
\begin{equation}
M(\alpha_{c_{t+1}}(F))
\ge (1/a)c_t M(F)/(2a)^{2t}-bc_t/(2a)^t
=(1/a)\sqrt{c_{t+1}}M(F)/(2a)^t-(b/(2a)^t)\sqrt{c_{t+1}}.
\end{equation}
Since $M(F)\ge2ab$, the subtracted term is at most
half the first term, so we get
\begin{equation}
M(\alpha_{c_{t+1}}(F))
\ge (1/2a)\sqrt{c_{t+1}}M(F)/(2a)^t
=\sqrt{c_{t+1}}M(F)/(2a)^{t+1},
\end{equation}
finishing the induction step.

Now, the recursion for $c_t$ is
$c_{t+1}=c_t^2/(2a)^{2t}$. Using the identity
$1.5^t\ge t$ for all positive integers $t$,
it is not hard to prove by induction
that $c_t\ge (2a)^{4\cdot 1.5^t}$. It is also clear that
$c_t\le (2a)^{4\cdot 2^t}$.

Consider the largest $t$ such that $c_t\le n$.
Then $t\le \log_{1.5}\log_{(2a)^4} n\le\log_{1.5}\log n\le 2\log\log n$, and
$n< c_{t+1}=c_t^2/(2a)^{2t}$, so
$(2a)^t\sqrt{n}< c_t\le n$.
If $c_t\ge n/2$, we have 
\begin{equation}
M(\alpha_n(F))\ge M(\alpha_{c_t}(F))
\ge\sqrt{c_t}M(F)/(2a)^t
=(1/\sqrt{2})(2a)^{-t}\sqrt{n}M(F).
\end{equation}
Otherwise, $\lfloor n/c_t\rfloor\ge n/2c_t$, and we have
\begin{equation}
    M(\alpha_n(F))\ge
    M(\alpha_{\lfloor n/c_t\rfloor}(\alpha_{c_t}(F)))
    \ge(1/a)\sqrt{n/2c_t}M(\alpha_{c_t}(F)) -bn/c_t
\end{equation}
\begin{equation}
\ge(1/a)\sqrt{n/2c_t}(2a)^{-t}\sqrt{c_t}M(F) -bn/c_t
=(1/\sqrt{2}a)(2a)^{-t}\sqrt{n}M(F) -bn/c_t.
\end{equation}
Since $c_t\ge\sqrt{n}(2a)^t$,
we have $n/c_t\le \sqrt{n}(2a)^{-t}$.
Also, since $M(F)\ge 2ab$, we have
$b\le (1/2a)M(F)$. Hence the subtracted term above
is at least a factor of $\sqrt{2}$ smaller than the
first term, which means that subtracting it off
decreases the first term by a factor of at most
$\sqrt{2}-1$. Using $(\sqrt{2}-1)/\sqrt{2}\ge 1/5$,
we get
\begin{equation}
M(\alpha_n(F))\ge(1/5a)(2a)^{-t}\sqrt{n}M(F).
\end{equation}
Note that since $\sqrt{2}<5$, the above inequality
is also satisfied in the first case, where $c_t\ge n/2$.
Finally, using $t\le 2\log\log n$, we get
\begin{equation}
M(\alpha_n(F))
\ge\frac{\sqrt{n}M(F)}{5a(2a)^{2\log \log n}}
=\frac{\sqrt{n}M(F)}{5a(\log n)^{2\log 2a}},
\end{equation}
from which the desired result follows.
\end{proof}

\subsection{\texorpdfstring{$\PrOR$}{PrOR} composition}
\label{sec:commPrOR}

In this section, we extend the composition result to $\PrOR$.
To do so, we first have to define the notion of partial functions
in communication complexity. As in the query setting,
a partial communication problem will be a function
$F:\X\times\Y\to\{0,1,*\}$. We use $\Dom(F)$ to denote the set
of pairs $(x,y)$ with $F(x,y)\ne *$.
We associate a communication
matrix with a sign matrix, which will have $\{-1,1,*\}$ entries.
We say that a real matrix $A$ approximates (the sign matrix of) $F$
to error $\epsilon$ if $A_{xy}$ is within $\epsilon$ of $(-1)^{1-F(x,y)}$
for all $(x,y)\in\Dom(F)$, and in addition, $|A_{xy}|\le 1$
for all $(x,y)\in\X\times\Y$ (even those outside the promise of $F$).
The measures $\arank(F)$ and $\agamma_2(F)$ are then defined as
they were previously, minimizing over all matrices approximating $F$
under this new definition of approximation.

Finally, we define composition for partial functions. A partial Boolean
function $f:\B^n\to\Ba$ composed with a partial communication function
$F:\X\times\Y\to\Ba$ is a partial communication function
$f\circ F:\X^n\times\Y^n\to\Ba$. On inputs where $(x_1,y_1),\dots,(x_n,y_n)$
are all inside $\Dom(F)$, the output of $f\circ F$ is given by the normal
function composition -- that is, it is equal to $f(F(x_1,y_1),\dots,F(x_n,y_n))$.
On inputs where some $(x_i,y_i)$ is not in the domain of $F$, the
output of $f\circ F$ is simply $*$. Another way to view it is
to extend $f$ to take inputs from $\Ba^n$, where $f(x)$ for any $x\notin\B^n$
is defined to be $*$, and then let $f\circ F$ be the usual function
composition.

We are now ready to prove a composition theorem for $\PrOR$.
To do so, we need the a version of \thm{multilinear_magic_poly}
showing the existence of a slightly robust polynomial
for the $\SCGT$ function that is also multilinear.
This polynomial will be in the $\{-1,1\}$ basis, and
we abuse notation by assuming $\XOR$ and $\SCGT$ have input
and output bits in $\{-1,1\}$ instead of $\B$.

\begin{lemma}\label{lem:multilinear_poly}
There is a real polynomial $p$ of degree
$O(\sqrt{n})$ acting on $2^n$ variables
$\{z_S\}_{S\subseteq[n]}$
and a constant $c\ge 10^{-5}$ such that for any input
$z\in\{-1,1\}^{2^n}$ with $\XOR_n\circ\SCGT_{2^n}(z)\ne *$
and any $\Delta\in[-c/n,c/n]^{2^n}$,
\begin{equation}
    |p(z+\Delta)-\XOR_n\circ\SCGT_{2^n}(z)|\le 2/3,
\end{equation}
and for all $z\in\B^{2^n}$,
we have $p(z)\in[-1,1]$. In addition, $p$ is
multilinear.
\end{lemma}

\begin{proof}
The existence of this polynomial follows directly
from \thm{multilinear_magic_poly}. All we need to do
is change bases from $\B$ to $\{-1,1\}$, which we can do
by applying the variable substitution $x\to 2x-1$ to all
input variables and to the output.
\end{proof}

The proof of the following composition lemma for $\PrOR$
will closely follow the proofs of \lem{gamma_or} (the analogous
composition theorem for $\OR$) and of \thm{approxdegPrOR}
(the composition theorem for $\PrOR$ in query complexity).

\begin{lemma}\label{lem:gamma_pror}
	Let $F:\X\times\Y\to\{0,1,*\}$ be a partial communication problem.
	Then
	\begin{equation}\log\agamma_2(\PrOR_n\circ F)\ge\Omega(\sqrt{n}\log\agamma_2(F)/\log n)-O(n/\log n).\end{equation}
\end{lemma}

\begin{proof}
Let $C$ be the constant from \thm{xor}. If
$\log\agamma_2(F)\le C+1$, the theorem holds trivially
by picking the constant in the $O(n)$ term to be larger than the
constant in the $\Omega(\sqrt{n}\log\agamma_2(F)/\log n)$ term
by a factor of at least $C+1$. Therefore, it suffices
to prove the result for functions $F$ satisfying
$\log\agamma_2(F)> C+1$.

We follow the proof of \lem{gamma_or} fairly closely.
Fix $F$ with $\log\agamma_2(F)> C+1$.
$F$ cannot be constant, so there is
an input pair $(a,b)$ such that $F(a,b)=0$.
Let $G=\PrOR_n\circ F$. For each $S\subseteq[n]$, let
$G_S$ be the function $\PrOR$ applied to the copies of
$F$ in the set $S$.
Let $A$ be a matrix approximating
$G$ to error $c/n$ with minimum $\gamma_2$ norm,
where $c$ is the constant from \lem{multilinear_poly}.
That is, $A_{xy}\in[-1,1]$ for all
$x\in\X,y\in\Y$, and $A_{xy}$ is within $c/n$
of $(-1)^{1-G(x,y)}$ for $(x,y)\in\Dom(G)$.

We take a moment to upper bound $\gamma_2(A)$.
Note that  $\log\agamma_2(\cdot)$ can be amplified
by plugging the matrix into a univariate amplification
polynomial (see \lem{amp}). A univariate polynomial
can amplify constant error to error $\epsilon$ using degree
$O(\log 1/\epsilon)$ (see footnote \ref{footnote:amppoly}). 
Moreover, if we plug in
a matrix with gamma $2$ norm $\ell$
into a polynomial of degree $d$ (using the Hadamard product),
the gamma $2$ norm of the result will be at most $\ell^d$
times the total sum of the absolute values of the coefficients
of the polynomial. Since a bounded univariate
polynomial of degree $d$
has coefficients that are at most $4^d$ (see, for example,
\cite{She13b}), we conclude the gamma $2$ norm of the result
is at most $d(4\ell)^d$. In our case, $d=O(\log n)$,
so we conclude that
$\log\gamma_2(A)=O(\log\agamma_2(G)\cdot\log n)$
(since one way to construct a good approximation matrix
like $A$ is to start with a $2/3$-approximation matrix
with logrank equal to $\log\agamma_2(G)$ and then amplify
it at $O(\log n)$ cost).

As before, we construct the matrices $A_S$ for each
$S\subseteq[n]$ as follows. First, let $A'_S$ be the submatrix
of $A$ we get by restricting $A$ to rows $x$ and columns $y$
satisfying $x_i=a$ and $y_i=b$ for all $i\notin S$. This effectively fixes the values on the copies
of $F$ outside the set $S$ to always be $0$. Next,
let $J$ denote the all-ones matrix with the same dimensions
as $F$, and tensor product $A'_S$ with $J$ $n-|S|$ times
(once for each $i\notin S$). This effectively adds $n-|S|$
inputs to $F$ that are always ignored. Finally, rearrange
the rows and columns of the resulting matrix so that the
ignored inputs are in the positions $i\notin S$. The final
matrix is $A_S$.

It is not hard to see that $A_S$ approximates $G_S$ to the
usual $2/3$ error for all $S$. Moreover, since $\gamma_2$
is invariant under rearrangements of rows and columns and
under tensor products with $J$, and since it is nonincreasing
under restrictions to a submatrix, we have
$\gamma_2(A_S)\le\gamma_2(A)$ for all $S$.

From \lem{multilinear_poly}, we have a polynomial $p$
of degree $d=O(\sqrt{n})$ that
(approximately) computes $\XOR_n\circ\SCGT_{2^n}$
with robustness $c/n$. We plug in the matrices $A_S$
into the variables $z_S$ of $p$ using the Hadamard
product as the matrix product, to get a matrix $B$,
and apply \lem{poly_matrix}. This gives $\log\agamma_2(B)=O(\deg(p)\log\agamma_2(A)+\deg(p)\log(2^n))
=O(\sqrt{n}\log n\log\agamma_2(G)+n^{1.5})$.

We now wish to show that $B$ approximates the sign matrix
of $\XOR_n\circ F$. That is, we need to show each entry
of $B$ is within $2/3$ of the corresponding entry of the sign
matrix. Fix an entry $(x,y)$ of $B$ with $x=(x_1,x_2,\dots,x_n)$
and $y=(y_1,y_2,\dots,y_n)$. Then $B_{xy}$ is the result
of applying $p$ to the variables $z_S=(A_S)_{xy}$.
We know that $(A_S)_{xy}$ approximates $G_S(x,y)$ for
$S$ such that $(x,y)\in\Dom G_S$. For $S$ such that $(x,y)\notin\Dom(G_S)$, we have $(A_S)_{xy}\in[-1,1]$.
Call the $S$ such that $(x,y)\in\Dom(G_S)$ \emph{good}
and the rest \emph{bad}.
As in the proof of \thm{approxdegPrOR},
we can the vector $(v_S)_{S\subseteq[n]}$ with $v_S=(A_S)_{xy}$
as a convex combination
of vectors that agree with $z_S$ on all good $S$ and take
$\{-1,1\}$ values on all bad $S$. The evaluation
of $p$ on each of these support vectors will then be
within $2/3$ of the sign matrix of $\XOR_n\circ F$
at the entry $(x,y)$, since these support vectors are within
$c/n$ of being integer points that satisfy the $\SCGT$
promise. Since $p$ is multilinear, it follows that it evaluates
to within $2/3$ of the right value on $v$ itself as well.
Hence $B_{xy}$ approximates the sign matrix of $\XOR_n\circ F$,
as we hoped.

\thm{xor} tells us that
$\log\gamma_2(B)=\Omega(n\log\agamma_2(F))$. But we also know
that
$\log\agamma_2(B)=O(\sqrt{n}\log n\log\agamma_2(G)+n^{1.5})$,
from which the desired result follows.
\end{proof}

As before, we can replace the additive $O(n)$ with a multiplicative term;
this time, we lose a quasi-polylogarithmic factor in $n$ instead of merely
a polylogarithmic factor.

\begin{theorem}\label{thm:pror_gamma}
	For all $F:\X\times\Y\to\B$,
	\begin{equation}\log\agamma_2(\PrOR_n\circ F)
	\ge n^{1/2-o(1)}\log\agamma_2(F).\end{equation}
\end{theorem}

We note that we do not prove this theorem for partial functions
due to a technicality: we cannot handle partial functions which do not
contain $\AND^{\cc}$ as a subproblem yet are not equivalent to
$\NOTEQ$. It is possible that some such function exists that has
$\agamma_2(F)>1$ and yet behaves like $\NOTEQ$ in that
$\log\agamma_2(\OR_n\circ F)=O(1)$. However, we note that our proof
\emph{does} apply to partial functions $F$ that contain the two-bit
$\AND^{\cc}$ as a subproblem or can otherwise be shown to obey
$\agamma_2(\PrOR_n\circ F)=\Omega(\sqrt{n})$.

\begin{proof}
    We use \lem{recurse}.
	The key to the proof is the observation that
	\begin{equation}\PrOR_s\circ\PrOR_t=\PrOR_{st}.\end{equation}
	This equality holds exactly: the function on the left has the same
	domain as the function on the right, and the two give the same values
	on that domain.
	
	This means we can apply \lem{recurse} in the same way
	as we did for
	\thm{polylog_gamma}; the only difference is that the term $a$
	is no longer constant, but instead as large as $O(\log n)$ due to
	the loss in \lem{gamma_pror}. Since the final loss in
	\thm{polylog_gamma} was $(\log n)^{O(\log a)}$, this means we lose
	a quasi-polylogarithmic factor in $n$.
	
	An additional requirement for applying \lem{recurse}
	is that $\log\agamma_2(\PrOR_n\circ F)=\tOmega(\sqrt{n})$.
	For this, note that if $F$ contains the two-bit
	$\AND^{\cc}$ function as a subproblem, then
	$\PrOR_n\circ F$ contains unique disjointness
	as a subproblem, which satisfies the desired lower bound
	(one can show this directly by a modification
	of \thm{disjointness}, losing a log factor,
	or using the known lower bound for unique set disjointness~\cite{Raz03,She11}).
	If $F$ does not contain the two-bit $\AND^{\cc}$ function
	as a subproblem,
	and if $F$ is a total function, then $F$ is equivalent
	to $\NOTEQ$. In this case, $\log\agamma_2(F)\le 0$
	and the theorem follows trivially.
\end{proof}

\subsection{Symmetric function composition}
\label{sec:commsym}
We have reached the final result of this section:
a composition theorem for approximate gamma $2$ norm
with arbitrary symmetric functions on the outside.

\begin{theorem}\label{thm:sym_gamma}
	For all $F:\X\times\Y\to\B$ and all symmetric functions
	$g:\B^n\to\B$,
	\begin{equation}\log\agamma_2(g\circ F)
	\ge \adeg(g)^{1-o(1)}\log\agamma_2(F).
	\end{equation}
\end{theorem}

We note that as for \thm{pror_gamma}, this theorem generalizes to partial
functions $F$ so long as $F$ contains the two-bit $\AND^{\cc}$
and two-bit $\OR^{\cc}$
as subproblems or can otherwise be shown to obey
$\agamma_2(\PrOR_n\circ F)=\Omega(\sqrt{n})$
and $\agamma_2(\PrAND_n\circ F)=\Omega(\sqrt{n})$.
We do not analyze partial functions $g$ at all, and leave this for future work.

\begin{proof}
	The proof mirrors that of \thm{approxdegsym}. Summarizing,
	we use a theorem of Paturi to characterize $\adeg(g)$ in terms
	of the most central Hamming layer $k$ such that $g$ behaves
	differently on Hamming layers $k$ and $k+1$.
	We then restrict $g$ to those Hamming layers; this reduces the
	problem to showing a composition theorem with $\PrTH_n^k$
	the outside.
	By negating the function if necessary, we
	assume without loss of generality that $k\le n/2$.
	A further restriction to a promise reduces the problem
	to showing a composition theorem with $\PrTH_{2k}^k\circ\PrOR_{n/2k}$
	on the outside. The composition theorem for $\PrOR_{n/2k}$
	follows from \thm{pror_gamma}. Finally, the composition
	theorem for $\PrTH_{2k}^k$ follows from the fact that
	$\adeg(\PrTH_{2k}^k)=\Omega(k)$ (as shown by Paturi) together with
	Sherstov's composition theorem \cite{She12}
	that applies when the approximate
	degree of the outer function is linear, a communication
	analogue of \thm{Sherstov_compose}.
\end{proof}

\section{Extensions of our results}

In this section we show two extensions of our main results.

\subsection{Extension to \texorpdfstring{$\OR$}{OR} of different functions}
\label{sec:unbalanced}

In \sec{query}, we showed a lower bound on the approximate degree of the OR of $n$ copies of a function $f$ (\thm{approxdeg}).
In this section, we will extend this result to a non-uniform version of this theorem for the case when the functions $f_i$ in the $\OR$ are possibly different. 

For $n$ Boolean functions $f_1,\ldots,f_n$, let $\OR_n \circ \, (f_1,f_2,\ldots,f_n)$ denote the function $\bigvee_{i=1}^n f_i(x_i)$, where $x_i$ is an input to $f_i$ and $x_i$'s have disjoint variables for different $i$. These functions may be on different input sizes, hence the $x_i$ may be of different sizes. We completely characterize the approximate degree of this function.

\begin{restatable}{theorem}{approxdegnu}
\label{thm:approxdegnu}
For any Boolean functions $f_1,f_2,\ldots,f_n$, we have 
\begin{equation}
    \adeg\Bigl(\OR_n \circ \, (f_1,f_2,\ldots,f_n)\Bigr) 
    = \Theta\Bigl(\sqrt{{\textstyle\sum_i} \adeg(f_i)^2}\Bigr).
\end{equation}
\end{restatable}

\begin{proof}
Our proof will only the following two facts:
\begin{align}
 &\adeg(\OR_n \circ \, (f_1,f_2,\ldots,f_n))=\Omega(\sqrt{n} \min_i \adeg(f_i)), ~ \mathrm{and} \label{eq:nulower}\\
 &\adeg(\OR_n \circ \, (f_1,f_2,\ldots,f_n))=O(\sqrt{n} \max_i \adeg(f_i)). \label{eq:nuupper}
\end{align}
The first equation is the same as \eq{ORmin}, proved at the end of \sec{query}. The second equation follows from Sherstov's robust polynomial construction (\thm{Sherstovrobust}), since we can take a robust polynomial for $\OR_n$ of degree $O(\sqrt{n})$ and plug in for the \th{i} variable the approximating polynomial for $f_i$.

To see how these imply the claim, consider the function $F= \OR_n \circ \, (f_1,f_2,\ldots,f_n)$.
Let 
\begin{equation}
    d_i = \adeg(f_i)^2.
\end{equation}
Now let $k$ be the least common multiple of the numbers $d_i$, which is well defined as the each $d_i$ is a positive integer. Consider the function $G=\OR_k \circ F$. By \eq{nulower} and \eq{nuupper}, we have that
\begin{equation}
\adeg(G) = \Theta(\sqrt{k}\,\adeg(F))\label{eq:Gupper}.
\end{equation}
Now since $\OR$ is associative, we can also write $G$ as 
\begin{equation}
G= \OR_{k}\circ F = \OR_{nk}\circ (\underbrace{f_1,\ldots,f_1}_\text{$k$ times},\underbrace{f_2,\ldots,f_2}_\text{$k$ times},\ldots, \underbrace{f_n,\ldots, f_n}_\text{$k$ times}).
\end{equation}

Now let us regroup the $k$ copies of $f_i$ into smaller $\OR$s such that each group has the same total weight, where we think of $f_i$ having weight $d_i$. 
So the $k$ copies of $f_i$ is regrouped into $d_i$ groups, each containing $k/d_i$ copies of $f_i$. 
Let $\ell = \sum_{i=1}^{n} d_i$. We can rewrite $G$ as
\begin{align}
G = \OR_{\ell}\circ \Big(&\overbrace{\OR_{k/d_1}\circ(f_1,\ldots f_1),\ldots,\OR_{k/d_1}\circ (f_1,\ldots f_1)}^\text{$d_1$ times}, \nonumber \\
&\ldots, \underbrace{\OR_{k/d_n}\circ (f_n,\ldots f_n),\ldots,\OR_{k/d_n}\circ (f_n,\ldots f_n)}_\text{$d_n$ times} \Big). \label{eq:group}
\end{align}

For example, if $d_1=1$, $d_2=2$, and $d_3=3$, then $k=6$ and this would be 
\begin{align*}
G &= \OR_{18}\circ (f_1,f_1,f_1,f_1,f_1,f_1,f_2,f_2,f_2,f_2,f_2,f_2, f_3,f_3,f_3,f_3,f_3,f_3)\\
&= \OR_{6}\circ \left( \OR_{6} (f_1,\ldots,f_1),\OR_{3}(f_2,f_2,f_2), \OR_3(f_2,f_2,f_2), \OR_2(f_3,f_3), \OR_2(f_3,f_3),\OR_2(f_3,f_3)\right). 
\end{align*}

Now by \eq{nulower} and \eq{nuupper}, we can compute the approximate degree of each of these groups. We have for all $i\in[n]$,
\begin{equation}
\adeg(\OR_{k/d_i} \circ f_i) = \Theta \left(\sqrt{k/d_i} \adeg(f_i)\right) = \Theta\bigl(\sqrt{k}\bigr). \label{eq:a22}
\end{equation}
Using this, we can apply \eq{nulower} and \eq{nuupper} to \eq{group}, and get
\begin{equation}
\adeg(G)= \Theta(\sqrt{\ell k})\label{eq:Glower}.
\end{equation}
Combining equations \eq{Gupper} and \eq{Glower}, we get
\begin{equation}
\adeg(F) = \Theta\left(\sqrt{\textstyle\sum_i \adeg(f_i)^2}\right).\qedhere
\end{equation}
\end{proof}

\subsection{Extension to quantum information complexity}
\label{sec:QIC}

In this section we prove an $\OR$-composition theorem for quantum information complexity and in particular, we establish \thm{QIC}, which we state below: 
\begin{restatable}{theorem}{QICdisj}
\label{thm:QIC}
Let $\DISJ_n$ be the set disjointness function and let $\QIC(F,\epsilon)$ denote the $\epsilon$-error distribution-free quantum information complexity of $F$. Then
\begin{equation}
\QIC (\DISJ_n,1/3) = \Omega(\sqrt{n}/\log n).    
\end{equation}
\end{restatable}

\subsubsection{Preliminaries}

In this section we assume the reader is familiar with quantum information complexity. 
The reader is referred to \cite{Touchette} for more details.

Let us start by recalling the definition of quantum information complexity.
Let $\Pi$ be a quantum communication protocol and let $\mu$ be a probability distribution over  inputs $(x,y)$.
At the end of round $i$ of the protocol, we assume there are three registers $A_i$, $B_i$, and $C_i$.
$A_i$ is with Alice and $B_i$ is with Bob and $C_i$ was sent as a message in the $i^{\text{th}}$ round by either Alice or Bob depending on whether $i$ is odd or even, respectively. 
Let $\ket{\psi^{i,x,y}}_{A_i,B_i,C_i}$ be the joint state on the three registers on input $(x,y)$. 
The total purified state at the end of round $i$ is given by
\begin{equation}
\ket{\psi^i}_{X,Y,R, A_i,B_i,C_i} = \sum_{x,y} \sqrt{\mu(x,y)} \ket{x}_X \ket{y}_Y \ket{x,y}_R \ket{\psi^{i,x,y}}_{A_i,B_i,C_i}.
\end{equation}
Then the quantum information cost of $\Pi$ on the distribution $\mu$, $\QIC(\Pi,\mu)$, is given by
\begin{equation}
\QIC(\Pi, \mu) = \sum_{i \: \text{odd}} I(R; C_i| Y, B_i)_{\psi^i} + \sum_{i \: \text{even}} I(R; C_i| X, A_i)_{\psi^i}.
\end{equation}
Note that our definition of quantum information cost is off by a factor of $2$ from the definition in \cite{Touchette}. However, for simplicity we ignore this factor, since our lower bounds ignore constant factors anyway.

Now the quantum information complexity of a function $F$ on distribution $\mu$ with error $\epsilon$, denoted $\QIC(F,\mu, \epsilon)$, is defined as 
\begin{equation}
\QIC(F,\mu,\epsilon) = \inf_{\substack{\Pi:~\Pi \textrm{ computes } F\\ \textrm{with error at most } \epsilon}} \QIC(\Pi,\mu),
\end{equation}
the infimum of quantum information costs of protocols $\Pi$ (with respect to $\mu$) which compute $F$ with error at most $\epsilon$. 
Finally, the distribution-free quantum information complexity of $F$, $\QIC(F, \epsilon)$, is defined as
\begin{equation}
\QIC(F, \epsilon) = \max_{\mu} \QIC(F,\mu, \epsilon),
\end{equation}
the maximum over distributions $\mu$ of $\QIC(F,\mu, \epsilon)$. If $\epsilon$ is unspecified, it is taken to be $1/3$, and we use $\QIC(F)$ to mean $\QIC(F,1/3)$.

We will need the following basic lemmas about quantum information cost and quantum information complexity. 
The first lemma is about switching the quantifiers in the definition of quantum information complexity.

\begin{lemma}[\cite{BGKMT15}]\label{lem:minimax} Consider the following alternate definition of quantum information complexity:
\begin{equation}
\widetilde{\textnormal{QIC}}(F, \epsilon) = \inf_{\substack{\Pi:~\Pi \textrm{ computes } F\\ \textrm{with error at most } \epsilon}} \: \max_{\mu} \: \QIC(\Pi, \mu).
\end{equation}
Then 
\begin{equation}
\widetilde{\textnormal{QIC}}(F, 2\epsilon) \le 2 \cdot \QIC(F, \epsilon).    
\end{equation}
\end{lemma}

The next lemma is about cleaning up a protocol, using the standard ``uncomputation'' trick, so that all the registers except an output bit are essentially returned to their original states. The proofs in~\cite{cleve1999quantum, hoza2017quantum} are for quantum communication cost but the same proofs work for quantum information cost as well~\cite{TL17}.

\begin{lemma}\label{lem:clean}
Suppose $\Pi$ is a quantum protocol such that $\QIC(\Pi, \mu) \le I$ for all $\mu$ and that computes a Boolean function $f$ with error at most $ \epsilon$. Then there is a cleaned up version of $\Pi$, $\Pi'$ which satisfies $\QIC(\Pi', \mu) \le 2I$ and the following property. Suppose $\ket{\psi}_{A B}$ is the entanglement shared at the start of $\Pi'$ and $\ket{\phi^{x,y}}_{A B B_{\textnormal{out}}}$ be the final state of the protocol $\Pi'$ on inputs $x,y$ with register $A$ held by Alice and registers $B, B_{\textnormal{out}}$ held by Bob. Also let  $\ket{\widetilde{\phi^{x,y}}}_{A B B_{\textnormal{out}}}$ be the state $\ket{\psi}_{A B} \otimes \ket{f(x,y)}_{B_{\textnormal{out}}}$. Then it holds that
\begin{equation}
\|\phi^{x,y} - \widetilde{\phi^{x,y}}\|_{\textnormal{tr}} \le 16 \sqrt{\epsilon}
\end{equation}
for all inputs $x,y$.
\end{lemma}

The next lemma is about reducing the error probability and its effect on quantum information complexity \cite{BGKMT15}.

\begin{lemma}
\label{lem:error_reduction}
For any problem $F$, $\QIC(F,\epsilon) = O(\QIC(F,1/3) \cdot \log(1/\epsilon))$.
\end{lemma}

The next proposition states an elementary fact about quantum conditional mutual information which can proven by combining Uhlmann's theorem and unitary equivalence of quantum conditional mutual information (e.g. see \cite{BGKMT15}).

\begin{proposition}\label{prop:remove_pur}
Let $\ket{\psi}_{R,A,B,C}$ and $\ket{\phi}_{\widetilde{R},A,B,C}$ be two pure states s.t. the marginals of $\psi$ and $\phi$ on the registers $A,B,C$ are equal i.e. $\psi_{A,B,C} = \phi_{A,B,C}$. Then
\begin{equation}
I(R; C|B)_{\psi} = I(\widetilde{R}; C|B)_{\phi}    
\end{equation}
\end{proposition}

The next lemma is a direct sum theorem for quantum information complexity. 

\begin{lemma}[\cite{Touchette}]\label{lem:directsum}
Let $F$ be a (possibly partial) Boolean function and $\mu$ be a distribution. Then 
\begin{equation}
\QIC(F^n, \epsilon) \ge n \cdot \QIC(F, \epsilon).    
\end{equation}
\end{lemma}

The direct sum usually gives an equality but there are two differences between the above lemma and the usual direct sum statement: absence of a prior and different error model (where we demand that all the copies are solved except w.p.\ $\le \epsilon$). But it is easy to see that the usual direct sum statement (in \cite{Touchette}) implies the above inequality.

\subsubsection{\texorpdfstring{$\OR$}{OR} composition}

Now we are ready to prove our general result, an $\OR$-composition theorem for quantum information complexity.

\begin{theorem}\label{thm:QIC_OR}
Let $F: \mathcal{X} \times \mathcal{Y} \rightarrow \{0,1\}$ be a Boolean function that has an all-zeroes column. Then
\begin{equation}
\QIC(\OR_n\circ F, 1/3) \ge \Omega\left(\frac{\sqrt{n}}{\log(n)} \cdot \QIC(F, 1/3)\right).    
\end{equation}
\end{theorem}

The above theorem implies the following OR lemma for general functions as a corollary. The proof (of the reduction from general functions to functions with an all-zeroes column) is similar to the self-reducibility arguments in the derivation of \thm{polylog_gamma} and  \lem{gamma_or} from \cor{clean_compose}. So we skip it to avoid repetition of the same arguments.

\begin{corollary}
Let $F: \mathcal{X} \times \mathcal{Y} \rightarrow \{0,1\}$ be a Boolean function. Then
\begin{equation}
\QIC(\OR_n\circ F, 1/3) \ge n^{1/2 - o(1)} \cdot \left( \QIC(F, 1/3) - n^{o(1)} \right).    
\end{equation}
\end{corollary}

\thm{QIC_OR} also implies \thm{QIC}, which lower bounds the quantum information complexity of disjointness. As mentioned before, this was already known \cite{BGKMT15} (and in fact without the log factor) but our proof is simpler and more intuitive. 

\QICdisj*

This follows from \thm{QIC_OR}  by plugging in $F = \text{AND}$, which does have an all-zeroes column, and observing that $\QIC(\text{AND}, 1/3) \ge \Omega(1)$ because the protocol needs to learn the value of $\text{AND}$ (with some accuracy). Note that here quantum information complexity is measured w.r.t. an arbitrary distribution in contrast with the quantum information complexity w.r.t. distributions having tiny mass on inputs evaluating to $1$ which arise in the study of disjointness (see \cite{BGKMT15} and the references therein) where quantum information complexity can approach zero (for unbounded round protocols).

Now we are ready to prove the general result of this section, \thm{QIC_OR}.

\begin{proof}[Proof of \protect{\thm{QIC_OR}}]
Suppose $\QIC(\OR_n\circ F, 1/3) = I$. 
By \lem{error_reduction}, $\QIC(\OR_n\circ F, 1/n^4) \le O(I \cdot \log(n))$. By the combination of \lem{minimax} and \lem{clean}, there is an almost clean protocol $\Pi$ s.t. $\QIC(\Pi, \mu, 2/n^4) \le O(I \cdot \log(n))$ for all distributions $\mu$. 

We will use this protocol $\Pi$ as a black box along with Belovs' algorithm (with error $1/5$) for combinatorial group testing (\thm{BelovsCGT}) to design a protocol $\tau$ for solving $n$ copies of $F$, with quantum information cost $\le  O(I \cdot \sqrt{n} \log(n))$, which will imply the lower bound we need using \lem{directsum}. 

Suppose Alice and Bob want to solve $n$ copies of $F$ on inputs $X_1,\ldots, X_n$ and $Y_1,\ldots, Y_n$ jointly distributed according to some distribution $\mu$ (with bounded probability of error). Bob will run Belovs' algorithm (with the goal of learning the string $F(X_1,Y_1),\ldots, F(X_n, Y_n)$) with a query to $\vee_{i \in S} F(X_i, Y_i)$ simulated by running $\Pi$ with Alice where Alice's input is $X_1,\ldots, X_n$ and Bob's input is $(Y_i)_{i \in S}$ and fixed to $y^*$ outside $S$, where $y^*$ is the input corresponding to the all-zeroes column in $F$, i.e. $F(x,y^*) = 0$ for all $x$. It is not hard to see that this protocol $\tau$ allows Bob to predict $F(X_1, Y_1),\ldots, F(X_n, Y_n)$ except with probability $O(1/n) + 1/5$ (using the fact that $\Pi$ is almost clean). Also one can prove that the quantum information cost of $\tau$ is $\le O(\sqrt{n} \cdot I \log(n))$. To see this, suppose the (unnormalized) state at the end of the $i^{\text{th}}$ round of the $j^{\text{th}}$ simulation of $\Pi$ (while running $\tau$) be
\begin{equation}
\ket{\psi^{j, i}}_{X,Y,\widetilde{Y}, R, \widetilde{A}^j, \widetilde{B}^j, A^j_i, C^j_i, B^j_i} = \sum_{x,y} \ket{x}_X \ket{y}_Y \ket{\widetilde{y}}_{\widetilde{Y}} \ket{x, y}_R \ket{\phi^{j,x,y}}_{\widetilde{A}^j, \widetilde{B}^j} \ket{\psi^{j,i,x,\widetilde{y}}}_{A^j_i, C^j_i, B^j_i}
\end{equation}
Here $X,Y$ registers contain the actual inputs to protocol $\tau$ while the registers $X, \widetilde{Y}$ contain the inputs to $\Pi$ in the $j^{\text{th}}$ simulation. The registers $\widetilde{A}^j, \widetilde{B}^j$ contain the garbage left from previous simulations of $\Pi$ including the previous answers, Bob's query register $S$ for the current round etc. Note that in the $j^{\text{th}}$ simulation of $\Pi$, Alice and Bob apply a sequence of unitaries on registers $A^j_i, C^j_i, B^j_i$ controlled on the registers $X, \widetilde{Y}$ without touching the registers $Y, \widetilde{A}^j, \widetilde{B}^j$. This will be crucial. Now the quantum information cost of $\tau$ is the following:
\begin{equation}
\sum_{j=1}^q \left(\sum_{i \: \textnormal{odd}} I(R; C^j_i| Y, \widetilde{Y}, B^j_i, \widetilde{B}^j)_{\psi^{j,i}} + \sum_{i \: \textnormal{even}} I(R; C^j_i| X, A^j_i, \widetilde{A}^j)_{\psi^{j,i}} \right)
\end{equation}
where $q$ is the query cost of Belovs' algorithm ($q \le O(\sqrt{n})$). We will prove that for every $j$,
\begin{align}
\sum_{i \: \textnormal{odd}} I(R; C^j_i| Y, \widetilde{Y}, B^j_i, \widetilde{B}^j)_{\psi^{j,i}} + \sum_{i \: \textnormal{even}} I(R; C^j_i| X, A^j_i, \widetilde{A}^j)_{\psi^{j,i}} \le O(I \cdot \log(n)) \label{eq:aeqn10}
\end{align}
which will complete the proof. We will in fact prove a stronger statement:
\begin{align}
    \sum_{i \: \textnormal{odd}} I(R, Y, \widetilde{A}^j, \widetilde{B}^j; C^j_i| \widetilde{Y}, B^j_i)_{\psi^{j,i}} + \sum_{i \: \textnormal{even}} I(R, Y, \widetilde{A}^j, \widetilde{B}^j; C^j_i| X, A^j_i)_{\psi^{j,i}} \le O(I \cdot \log(n)) \label{eq:aeqn11}
\end{align}
from which \eq{aeqn10} follows by applying the chain rule and positivity of quantum conditional mutual information. Now note that the marginal state on registers $X, \widetilde{Y}$ in $\psi^{j,i}$ is a classical distribution since there is a copy of $x$ in $R$ and $\widetilde{y}$ is a deterministic function of $y$ and the value in query register $S$ inside $\widetilde{B}^j$. Denote this distribution by $\nu$. Consider the following alternate states:
\begin{equation}
\ket{\phi^{j,i}}_{X, \widetilde{Y}, \widetilde{R}, A^j_i, C^j_i, B^j_i} = \sum_{x,\widetilde{y}} \sqrt{\nu(x, \widetilde{y})} \ket{x}_X  \ket{\widetilde{y}}_{\widetilde{Y}} \ket{x, \widetilde{y}}_{\widetilde{R}}  \ket{\psi^{j,i,x,\widetilde{y}}}_{A^j_i, C^j_i, B^j_i}
\end{equation}
Because $\QIC(\Pi, \nu) \le O(I \cdot \log(n))$, we get that
\begin{align}
     \sum_{i \: \textnormal{odd}} I(\widetilde{R}; C^j_i| \widetilde{Y}, B^j_i)_{\phi^{j, i}} + \sum_{i \: \textnormal{even}} I(\widetilde{R}; C^j_i| X, A^j_i)_{\phi^{j,i}} = \QIC(\Pi, \nu) \le O(I \cdot \log(n)) \label{eq:a12}
\end{align}
Also note that the marginal states on registers $X, \widetilde{Y}, A^j_i, C^j_i, B^j_i$ are the same in the two states $\phi^{j,i}$ and $\psi^{j,i}$. This along with \prop{remove_pur} implies that
\begin{align}
    &\: \: \: \: \: \: \sum_{i \: \textnormal{odd}} I(\widetilde{R}; C^j_i| \widetilde{Y}, B^j_i)_{\phi^{j, i}} + \sum_{i \: \textnormal{even}} I(\widetilde{R}; C^j_i| X, A^j_i)_{\phi^{j,i}} \nonumber \\
    &= \sum_{i \: \textnormal{odd}} I(R, Y, \widetilde{A}^j, \widetilde{B}^j; C^j_i| \widetilde{Y}, B^j_i)_{\psi^{j,i}} + \sum_{i \: \textnormal{even}} I(R, Y, \widetilde{A}^j, \widetilde{B}^j; C^j_i| X, A^j_i)_{\psi^{j,i}} \label{eq:a14}
\end{align}
Combining \eq{a12} and \eq{a14} gives us \eq{aeqn11} which completes the proof.
\end{proof}

\section*{Acknowledgements}
We would like to thank Mark Bun and Justin Thaler for helpful discussions and feedback on an early draft of this work. 
We would also like to thank Harry Buhrman for bringing reference~\cite{buhrman1998lower} to our attention. R.K.\ would like to thank Jeongwan Haah for helpful discussions regarding the proof of \thm{coeff_bounded}.

Some of this work was performed while the first two authors were students at the Massachusetts Institute of Technology and the last author was a postdoctoral associate at the Massachusetts Institute of Technology.
This work was partially supported by ARO grant W911NF-12-1-0541, NSF grant CCF-1410022, NSF grant CCF-1629809, and a Vannevar Bush faculty fellowship.

\appendix
\section{Coefficients of bounded polynomials}
\label{app:coeff_bounded}

In this section, we prove an elementary result about the maximum absolute value of coefficients of (not necessarily multilinear) multivariate polynomials that are bounded inside the unit cube. Our main result is the following theorem.

\begin{theorem}
\label{thm:coeff_bounded}
Let $p$ be a polynomial with real coefficients on $n$ variables with degree $d$ such that for all $x\in[0,1]^n$, $|p(x)|\leq 1$. Then the magnitude of any coefficient of $p$ is at most $(2 d)^{3d}$,
and the sum of magnitudes of all coefficients of $p$ is at most $(2 (n+d))^{3d}$.
\end{theorem}

For univariate polynomials a similar theorem is known; see~\cite{She13b} for an elementary proof of this fact.
Our \thm{coeff_bounded} follows from the following lemma, which at first looks like a weaker result.

\begin{lemma}
\label{lem:coeff_bounded} 
Let $p$ be a polynomial with real coefficients on $n$ variables with degree $d$ such that for all $x\in[0,1]^n$, $|p(x)|\leq 1$. Then the magnitude of any coefficient of $p$ is at most $(2nd(n+d))^d$.
\end{lemma}

Let us first prove the main result (\thm{coeff_bounded}) from this lemma.

\begin{proof}[Proof of \protect{\thm{coeff_bounded}}]
Suppose we want to bound the coefficient of some monomial $M$. Since $M$ has degree at most $d$, at most $d$ variables appear in $M$. 
We can set all remaining variables (that do not appear in $M$) to $0$ and this reduces the number of variables to at most $d$. 
This new polynomial is also bounded on the cube and hence we can apply \lem{coeff_bounded} with $n=d$ to get the desired bound.

To bound the total magnitude of all coefficients,
we first need to count the number of coefficients.
The number of monomials of degree $d$
on $n$ variables is $\binom{n+d-1}{d}$.
This is smaller than $\binom{n+d}{d}$.
We upper bound the latter by
$\left(\frac{e(n+d)}{d}\right)^d$,
which is a standard inequality for binomial coefficients. Multiplying by the bound
from \lem{coeff_bounded}, we get
an upper bound of $(2en(n+d)^2)^d$
on the magnitude of all coefficients of monomials
of degree $d$. Some monomials may have degree
smaller than $d$; we therefore need to sum the
above over degrees from $d$ down to $0$. But this
series decreases faster than a geometric series --
indeed, assuming $n,d\ge 1$, it decreases strictly
faster than a geometric series that decreases by
a factor of $8e\ge 21$ each term. Such a geometric
series has a sum at most $21/20$ times
the largest term. Replacing $2e$ by $8$ in the equation is an increase by a factor larger
than $21/20$, so we can safely upper bound the
total magnitude of all coefficients by
$(8n(n+d)^2)^d\le (2(n+d))^{3d}$.
\end{proof}

The remainder of this section is devoted to proving \lem{coeff_bounded}. The main idea is to use multivariate Lagrange interpolation similar to the proof for univariate case in \cite{She13b}. We first develop some Lagrange interpolators and prove some basic properties about them.

Consider the set of points 
\begin{equation}
S = \{\alpha \in \mathbb{R}^n: d \alpha \in \{0,1,\ldots, d\}^n \: \text{and} \: \sum_i \alpha_i \le 1\}.    
\end{equation}
Note that the elements of $S$ are in one-to-one correspondence with degree $\le d$ monomials in $n$ variables and $|S| = \binom{n+d}{d}$.
We first define Lagrange interpolators with respect to the points in $S$.

\begin{proposition}
\label{prop:p_alpha_definition}
For every $\alpha \in S$, there is a degree $d$ polynomial $p_{\alpha}$ s.t. $p_{\alpha}(\alpha) = 1$ and $p_{\alpha}(\beta) = 0$ for all $\beta \in S$ with $\beta \neq \alpha$.
\end{proposition}

\begin{proof}
We first define polynomials $q_{\alpha}$ that satisfy all the properties of $p_\alpha$, except that $q_{\alpha}(\alpha) \neq 0$ instead of $q_{\alpha}(\alpha)=1$. We then define $p_{\alpha}(x)$ to be $q_{\alpha}(x)/q_{\alpha}(\alpha)$. Let us first construct $q_{\alpha}$ for all $\alpha$ such that $\sum_i \alpha_i = 1$. We define $q_{\alpha}$ as
\begin{equation}
q_\alpha(x) = \prod_{i=1}^n \prod_{j_i = 0}^{d\alpha_i - 1} (x_i - j_i/d).    
\end{equation}
Clearly $q_{\alpha}$ has degree $d$ and $q_{\alpha}(\alpha) \neq 0$ because all the terms of the form $(\alpha_i - j_i/d)$ are nonzero by construction. But why is $q_\alpha(\beta) = 0$ for all $\beta \in S$ with $\beta \neq \alpha$? Since $\sum_i \alpha_i = 1$ and $\sum_i \beta_i \le 1$ as well as $\beta \neq \alpha$, it follows that there is some $i\in [n]$ s.t. $\beta_i < \alpha_i$ and hence $\beta_i = j_i/d$ for some $j_i \le d \alpha_i - 1$ (due to the integrality of $d \beta_i$ and $d \alpha_i$). For that choice of $i$, the term $(\beta_i - j_i/d)$ in the product above will be zero.

Now let us construct the polynomial corresponding to $\alpha=0^n$, $q_{0^n}$. Here the polynomial is quite simple as well.
\begin{equation}
q_{0^n}(x) = \prod_{j=1}^d \left(\sum_{i=1}^n x_i - j/d\right).    
\end{equation}
It is clear that $q_{0^n}(\beta) = 0$ for all $\beta \in S$ with $\beta \neq 0$. Also $q_{0^n}$ has degree $d$ and $q_{0^n}(0) \neq 0$.

Now let us construct the polynomial in the general case, which will be a combination of the two cases discussed above. Suppose $\sum_{i} \alpha_i = k/d$. Then the polynomial is as follows.
\begin{equation}
q_\alpha(x) = \prod_{j=k+1}^d \left(\sum_{i=1}^n x_i - j/d\right) \prod_{i=1}^n \prod_{j_i = 0}^{d\alpha_i - 1} (x_i - j_i/d)    
\end{equation}
It is clear that the degree of the polynomial is $d$ and that $q_\alpha(\alpha) \neq 0$. The property that $q_\alpha(\beta) = 0$ for $\beta \in S$, $\beta \neq \alpha$ follows from a combination of arguments given above.
\end{proof}

We now list some properties of the polynomials $p_{\alpha}$ which will be useful for us.

\begin{proposition}\label{prop:p_alpha_properties} The following statements hold for the polynomials constructed in \prop{p_alpha_definition}.
\begin{enumerate}
    \item The collection of polynomials $\{p_{\alpha}\}_{\alpha \in S}$ forms a basis for the space of $n$-variate degree-$d$ polynomials.
    \item For each $\alpha \in S$, the magnitude of coefficients of $p_{\alpha}$ is at most $d^d (2n)^d$.
\end{enumerate}
\end{proposition}

\begin{proof}
For item $1$, it suffices to prove that the polynomials $p_\alpha$ are linearly independent since the number of polynomials is the same as the dimension of the space. Suppose we have a linear combination of these polynomials which is $0$:
$\sum_{\alpha \in S} c_{\alpha} p_{\alpha} = 0$.
Evaluating this linear combination at all $\beta \in S$, we get
\begin{equation}
c_\beta = c_{\beta} p_{\beta}(\beta) = 0,    
\end{equation}
which proves their linear independence.

To bound the coefficients of $p_\alpha(x) = q_\alpha(x)/q_\alpha(\alpha)$, first note that we can instead scale $q_\alpha$ up by a factor of $d^d$. This scaled up version of $q_\alpha$ is the following:
\begin{equation}
d^d q_\alpha(x) = \prod_{j=k+1}^d \left(d\sum_{i=1}^n x_i - j\right) \prod_{i=1}^n \prod_{j_i = 0}^{d\alpha_i - 1} (d x_i - j_i).    
\end{equation}
It's easy to see that $|d^d q_\alpha(\alpha)| \geq 1$, since each term in the product is a nonzero integer. Hence to upper bound the coefficients of $p_\alpha$, it is sufficient to upper bound the coefficients of $d^d q_\alpha$. 

A crude upper bound on the magnitude of each coefficient in the above polynomial is $d^d (n+1)^d \le d^d (2n)^d$. One way to see this is that the polynomial is a product of degree-$1$ polynomials, which are of the form $(d\sum_{i=1}^n x_i - j)$ or $(d x_i - j_i)$. For each of these, the sum of magnitudes of all coefficients is at most $d(n+1)$. When taking the product of several polynomials, the sum of magnitudes of all coefficients is submultiplicative, hence the sum of magnitudes of all coefficients in the above polynomial is $d^d (n+1)^d \le d^d (2n)^d$.
\end{proof}

Now we are ready to prove \lem{coeff_bounded}. 

\begin{proof}[Proof of \protect{\lem{coeff_bounded}}]
Since $\{p_\alpha\}_{\alpha \in S}$ is a basis for the space of degree-$d$ polynomials on $n$ variables (by \prop{p_alpha_properties}), we know that the given bounded polynomial $p$ can be written as a linear combination of these polynomials
\begin{equation}
p = \sum_{\alpha \in S} c_\alpha p_\alpha,    
\end{equation}
for some real numbers $c_\alpha$. By evaluating the above expression at $\beta \in S$, we can compute the coefficient $c_\beta$ and we get
\begin{equation}
|c_\beta| = |p(\beta)| \le 1.    
\end{equation}
Since $p$ is a linear combination of the $\binom{n+d}{d}$ polynomials $p_\alpha$ with coefficients $|c_\alpha|\leq 1$, we get that the magnitude of the largest coefficient of $p$ is at most $\binom{n+d}{d} d^d (2n)^d \le (2nd(n+d))^d$, where we used the fact that the largest coefficient of any $p_\alpha$ is at most $d^d (2n)^d$ (using \prop{p_alpha_properties}).
\end{proof}

\section{Quantum query complexity of \texorpdfstring{$\SCGT$}{SCGT}}
\label{app:belovsSCGT}

In this appendix we formally prove \thm{strong_Belovs}, which follows from \cite{Bel15}.

\strongBelovs*

Before proving this, let us recall the definition of $\SCGT$ (\defn{SCGT}):

\scgt*

\begin{proof}
The lower bound on the complexity of $\SCGT$ follows from \thm{BelovsCGT}, since $\CGT$ is a special case of $\SCGT$.
To show the upper bound, we use the dual of the adversary bound~\cite{Bel15,Rei11,LMR+11}, which we now restate using notation convenient for this problem.

For each $S\subseteq [n]$, let $X_S$ be a $D\times D$ matrix, where $D$ is the domain of $\SCGT$ as defined above.
Then the quantum query complexity of $\SCGT_{2^n}$ is given by the following semidefinite program (SDP):
\begin{align}
    \text{minimize} \qquad &  \max_{z \in D} \displaystyle\sum_{S\subseteq[n]} X_S\llbracket z,z\rrbracket & \qquad \\
    \text{s.t.} \qquad & \displaystyle\sum_{S:z_S \neq z'_S} X_S\llbracket  z,z'\rrbracket =1 & \forall z,z' \in D \text{ with } x(z)\neq x(z')\\
    \text{and} \qquad  & X_S \succcurlyeq 0  & \forall S\subseteq [n] 
\end{align}
Following Belovs, we used the notation $X_S\llbracket   z,z'\rrbracket$ for $z,z'\in D$ to denote the $(z,z')$ entry of the matrix $X_S$. 
We will also talk about matrices where each entry is a function of some parameter $p \in \mathbb{R}$. 
For this usage, we will use the notation $X_S(p)$ to talk about the matrix, and $X_S(p)\llbracket   z,z'\rrbracket$ to talk about a particular entry of the matrix.

So to prove an upper bound on $Q(\SCGT)$ it suffices to exhibit $X_S$ obeying the above constraints whose objective value is $O(\sqrt{n})$. 
We construct these matrices in steps. We start by constructing vectors $\psi_S(p)$ of length $D$ as follows. For all $S\subseteq[n]$, and $p\in[0,1]$, define the vector $\psi_S(p)$ as
\begin{equation}
\psi_S(p)\llbracket z \rrbracket = \frac{1}{(1-p)^{|x(z)|/2}} \times \begin{cases} \sqrt[4]{np/(1-p)} & \sum_{i \in S} x(z)_i=0\\
\sqrt[4]{(1-p)/np} & \sum_{i \in S} x(z)_i=1\\
0 & \text{otherwise}
\end{cases}.    
\end{equation}
Now consider the following $D\times D$ rank-one matrices parameterized by $p\in[0,1]$:
\begin{equation}
Y_S(p)=\frac{p^{|S|}(1-p)^{n-|S|}}{2p} \psi_S(p) \psi^*_S(p).
\end{equation}
Finally, we define $X_S$ as
\begin{equation}
X_S = \int_0^1 Y_S(p) dp.
\end{equation}
We now claim these $X_S$ are a solution to the above SDP with objective value $O(\sqrt{n})$.

To see this, first note that since the $Y_S(p)$ are all rank 1 and hence positive semidefinite (PSD), and the PSD matrices form a convex cone, the $X_S$ matrices are PSD as well. This satisfies ones of the constraints of the SDP.

Next we show the objective value is $O(\sqrt{n})$. To see this, note that for any $z\in D$, we have that
\begin{align*}
Y_S(p)\llbracket  z,z\rrbracket  = \frac{p^{|S|}(1-p)^{n-|S|}}{2p(1-p)^{|x(z)|}} \times \begin{cases} \sqrt{np/(1-p)} & \sum_{i \in S} x(z)_i=0\\
\sqrt{(1-p)/np} & \sum_{i \in S} x(z)_i=1\\
0 & \text{o.w.}
\end{cases}
\end{align*}

Consider the Bernoulli distribution $P$ on $[n]$, and note that the probability of obtaining $S$ under distribution, denoted $P(S)$, is $p^{|S|}(1-p)^{n-|S|}$. One can see that summing the above over $S\subseteq [n]$ gives
\begin{align}
&\sum_{S\subseteq [n]} Y_S(p)\llbracket z,z\rrbracket  \nonumber \\
&= 
\frac{1}{2p(1-p)^{|x(z)|}} \left( \Pr_{S\sim P}\left[\sum_{i \in S} x(z)_i=0\right]\sqrt{\frac{np}{(1-p)}} + \Pr_{S\sim P}\left[\sum_{i \in S} x(z)_i=1\right]\sqrt{\frac{(1-p)}{np}} \right) \\
&= \frac{1}{2p(1-p)^{|x(z)|}} \left( (1-p)^{|x(z)|}\sqrt{\frac{np}{(1-p)}} + |x(z)|p(1-p)^{|x(z)|-1}\sqrt{\frac{(1-p)}{np}} \right) \\
&=\frac{1}{2} \left( \sqrt{\frac{n}{p(1-p)}} + |x(z)|\sqrt{\frac{1}{np(1-p)}} \right) \leq \sqrt{\frac{n}{p(1-p)}}.
\end{align}
Hence for all $z\in D$, we have that
\begin{equation}
\displaystyle\sum_{S\subseteq[n]} X_S\llbracket z,z\rrbracket \leq \int_0^1 dp \sqrt{\frac{n}{p(1-p)}} = \pi\sqrt{n}, 
\end{equation}
as desired. 

Finally we show these $X_S$ satisfy the remaining constraint. Suppose $z,z'\in D$ such that $x(z)\neq x(z')$, and consider $Y_S(p)\llbracket z,z'\rrbracket$. Since the $\psi_S$ vectors only have mass on points for which $\sum_{i \in S} x(z)_i\in\{0,1\}$, the only sets $S$ for which the value of $Y_S(p)\llbracket z,z'\rrbracket$ is nonzero are those for which either $\sum_{i \in S} x(z)_i=\sum_{i \in S} x(z')_i=0$, $\sum_{i \in S} x(z)_i=\sum_{i \in S} x(z')_i=1$, or $\sum_{i \in S} x(z)_i+\sum_{i \in S} x(z')_i=1$. 

Now by the definition of $\SCGT$, if $\sum_{i \in S} x(z)_i \in \{0,1\}$, then $z_S=\sum_{i \in S} x(z)_i$. Hence if $\sum_{i \in S} x(z)_i=\sum_{i \in S} x(z')_i=0$ or $\sum_{i \in S} x(z)_i=\sum_{i \in S} x(z')_i=1$, then we have that $z_S=z'_S$. Therefore if one considers the sum 
\begin{equation}\displaystyle\sum_{S:z_S \neq z'_S} Y_S(p)\llbracket z,z'\rrbracket\end{equation}
Then the only nonzero terms are those for which $\sum_{i \in S} x(z)_i+\sum_{i \in S} x(z')_i=1$.
Hence we have that 
\begin{align}
    \sum_{S:z_S \neq z'_S} Y_S(p)\llbracket z,z'\rrbracket &= \frac{\Pr_{S\sim P}[\sum_{i \in S} x(z)_i+\sum_{i \in S} x(z')_i=1]}{2p(1-p)^{|x(z)|+|x(z')|/2}} \\
    &=\frac{|x(z) \oplus x(z')|p (1-p)^{|x(z)\vee x(z')|-1}}{2p(1-p)^{|x(z)|+|x(z')|/2}} \\
    &=\frac{|x(z) \oplus x(z')|}{2} (1-p)^{|x(z)\vee x(z')|-1 - (|x(z)|+|x(z')|/2)} \\
    &= \frac{|x(z) \oplus x(z')|}{2} (1-p)^{\frac{|x(z) \oplus x(z')|}{2}-1 } 
\end{align}
Where $x(z) \oplus x(z')$ denotes the bitwise $\XOR$ of the strings and $x(z)\vee x(z')$ denotes their bitwise $\OR$.
Hence for all $z,z'$ such that $x(z)\neq x(z')$,
\begin{align}
    \sum_{S:z_S \neq z'_S} X_S\llbracket z,z'\rrbracket &= \int_0^1 dp \frac{|x(z) \oplus x(z')|}{2} (1-p)^{\frac{|x(z) \oplus x(z')|}{2}-1 } =1
\end{align}
Where we used that fact that $|x(z) \oplus x(z')|\geq 1$ as $x(z)\neq x(z')$, and that for any positive real $a>0$, we have $\int_0^1 a(1-p)^{a-1} dp =1$.
\end{proof}

\bibliographystyle{alphaurl}
\phantomsection\addcontentsline{toc}{section}{References} 
\renewcommand{\UrlFont}{\ttfamily\small}
\newcommand{\eprint}[1]{\small \upshape \tt \href{http://arxiv.org/abs/#1}{#1}}
\let\oldpath\path
\renewcommand{\path}[1]{\small\oldpath{#1}}
\bibliography{query}

\newcommand{\etalchar}[1]{$^{#1}$}
\begin{thebibliography}{BNRdW07}

\bibitem[AA03]{AA03}
Scott Aaronson and Andris Ambainis.
\newblock Quantum search of spatial regions.
\newblock In {\em Proceedings of the 44th Annual IEEE Symposium on Foundations
  of Computer Science (FOCS 2003)}, pages 200--209, 2003.
\newblock \href {http://dx.doi.org/10.1109/SFCS.2003.1238194}
  {\path{doi:10.1109/SFCS.2003.1238194}}.

\bibitem[Aar05]{aaronsonpostbqp}
Scott Aaronson.
\newblock Quantum computing, postselection, and probabilistic polynomial-time.
\newblock In {\em Proceedings of the Royal Society of London A: Mathematical,
  Physical and Engineering Sciences}, volume 461:2063, pages 3473--3482, 2005.
\newblock \href {http://dx.doi.org/10.1098/rspa.2005.1546}
  {\path{doi:10.1098/rspa.2005.1546}}.

\bibitem[Aar16]{aaronsonpnpsurvey}
Scott Aaronson.
\newblock {P=?NP}.
\newblock In {\em Open Problems in Mathematics}, pages 1--122. Springer, 2016.
\newblock \href {http://dx.doi.org/10.1007/978-3-319-32162-2_1}
  {\path{doi:10.1007/978-3-319-32162-2_1}}.

\bibitem[AM14]{AM14}
Andris Ambainis and Ashley Montanaro.
\newblock Quantum algorithms for search with wildcards and combinatorial group
  testing.
\newblock {\em Quantum Information \& Computation}, 14(5\&6):439--453, April
  2014.
\newblock URL: \url{http://dl.acm.org/citation.cfm?id=2638661.2638665}.

\bibitem[Amb05]{ambainis2005polynomial}
Andris Ambainis.
\newblock Polynomial degree and lower bounds in quantum complexity: Collision
  and element distinctness with small range.
\newblock {\em Theory of Computing}, 1(1):37--46, 2005.
\newblock \href {http://dx.doi.org/10.4086/toc.2005.v001a003}
  {\path{doi:10.4086/toc.2005.v001a003}}.

\bibitem[BBBV97]{BBBV97}
Charles~H. Bennett, Ethan Bernstein, Gilles Brassard, and Umesh Vazirani.
\newblock Strengths and weaknesses of quantum computing.
\newblock {\em SIAM Journal on Computing}, 26(5):1510--1523, 1997.
\newblock \href {http://dx.doi.org/10.1137/S0097539796300933}
  {\path{doi:10.1137/S0097539796300933}}.

\bibitem[BBC{\etalchar{+}}01]{BBC+01}
Robert Beals, Harry Buhrman, Richard Cleve, Michele Mosca, and Ronald de~Wolf.
\newblock Quantum lower bounds by polynomials.
\newblock {\em Journal of the ACM}, 48(4):778--797, 2001.
\newblock \href {http://arxiv.org/abs/quant-ph/9802049}
  {\path{arXiv:quant-ph/9802049}}, \href
  {http://dx.doi.org/10.1145/502090.502097} {\path{doi:10.1145/502090.502097}}.

\bibitem[BCW98]{BCW98}
Harry Buhrman, Richard Cleve, and Avi Wigderson.
\newblock Quantum vs. classical communication and computation.
\newblock In {\em Proceedings of the Thirtieth Annual ACM Symposium on Theory
  of Computing (STOC 1998)}, pages 63--68, 1998.
\newblock \href {http://dx.doi.org/10.1145/276698.276713}
  {\path{doi:10.1145/276698.276713}}.

\bibitem[BdW98]{buhrman1998lower}
Harry Buhrman and Ronald de~Wolf.
\newblock Lower bounds for quantum search and derandomization.
\newblock {\em arXiv preprint \eprint{arXiv:quant-ph/9811046}}, 1998.

\bibitem[BdW01]{BdW01}
Harry Buhrman and Ronald de~Wolf.
\newblock Communication complexity lower bounds by polynomials.
\newblock In {\em Proceedings 16th Annual IEEE Conference on Computational
  Complexity}, pages 120--130, 2001.
\newblock \href {http://dx.doi.org/10.1109/CCC.2001.933879}
  {\path{doi:10.1109/CCC.2001.933879}}.

\bibitem[Bel15]{Bel15}
Aleksandrs Belovs.
\newblock Quantum algorithms for learning symmetric juntas via the adversary
  bound.
\newblock {\em Computational Complexity}, 24(2):255--293, 2015.
\newblock \href {http://dx.doi.org/10.1007/s00037-015-0099-2}
  {\path{doi:10.1007/s00037-015-0099-2}}.

\bibitem[BGK{\etalchar{+}}15]{BGKMT15}
Mark Braverman, Ankit Garg, Young~Kun Ko, Jieming Mao, and Dave Touchette.
\newblock Near-optimal bounds on bounded-round quantum communication complexity
  of disjointness.
\newblock In {\em Proceedings of the 2015 IEEE 56th Annual Symposium on
  Foundations of Computer Science (FOCS 2015)}, pages 773--791, 2015.
\newblock \href {http://dx.doi.org/10.1109/FOCS.2015.53}
  {\path{doi:10.1109/FOCS.2015.53}}.

\bibitem[BKT18]{BKT18}
Mark Bun, Robin Kothari, and Justin Thaler.
\newblock The polynomial method strikes back: Tight quantum query bounds via
  dual polynomials.
\newblock In {\em Proceedings of the 50th Annual ACM SIGACT Symposium on Theory
  of Computing}, STOC 2018, pages 297--310. ACM, 2018.
\newblock \href {http://dx.doi.org/10.1145/3188745.3188784}
  {\path{doi:10.1145/3188745.3188784}}.

\bibitem[BNRdW07]{buhrman2007robust}
Harry Buhrman, Ilan Newman, Hein Rohrig, and Ronald de~Wolf.
\newblock Robust polynomials and quantum algorithms.
\newblock {\em Theory of Computing Systems}, 40(4):379--395, 2007.
\newblock \href {http://dx.doi.org/10.1007/s00224-006-1313-z}
  {\path{doi:10.1007/s00224-006-1313-z}}.

\bibitem[BRS95]{beigel1995pp}
Richard Beigel, Nick Reingold, and Daniel Spielman.
\newblock {PP} is closed under intersection.
\newblock {\em Journal of Computer and System Sciences}, 50(2):191--202, 1995.
\newblock \href {http://dx.doi.org/10.1006/jcss.1995.1017}
  {\path{doi:10.1006/jcss.1995.1017}}.

\bibitem[BT13]{BT13}
Mark Bun and Justin Thaler.
\newblock Dual lower bounds for approximate degree and {Markov-Bernstein}
  inequalities.
\newblock In {\em Automata, Languages, and Programming: 40th International
  Colloquium, ICALP 2013}, pages 303--314, 2013.
\newblock \href {http://dx.doi.org/10.1007/978-3-642-39206-1_26}
  {\path{doi:10.1007/978-3-642-39206-1_26}}.

\bibitem[BT15]{BT15}
Mark Bun and Justin Thaler.
\newblock Hardness amplification and the approximate degree of constant-depth
  circuits.
\newblock In {\em Automata, Languages, and Programming: 42nd International
  Colloquium, ICALP 2015}, pages 268--280, 2015.
\newblock \href {http://dx.doi.org/10.1007/978-3-662-47672-7_22}
  {\path{doi:10.1007/978-3-662-47672-7_22}}.

\bibitem[BV97]{BV97}
Ethan Bernstein and Umesh Vazirani.
\newblock Quantum complexity theory.
\newblock {\em SIAM Journal on Computing}, 26(5):1411--1473, 1997.
\newblock \href {http://dx.doi.org/10.1137/S0097539796300921}
  {\path{doi:10.1137/S0097539796300921}}.

\bibitem[CvDNT13]{cleve1999quantum}
Richard Cleve, Wim van Dam, Michael Nielsen, and Alain Tapp.
\newblock Quantum entanglement and the communication complexity of the inner
  product function.
\newblock {\em Theoretical Computer Science}, 486:11--19, 2013.
\newblock \href {http://dx.doi.org/10.1016/j.tcs.2012.12.012}
  {\path{doi:10.1016/j.tcs.2012.12.012}}.

\bibitem[DHH00]{du2000combinatorial}
Dingzhu Du, Frank~K Hwang, and Frank Hwang.
\newblock {\em Combinatorial group testing and its applications}, volume~12.
\newblock World Scientific, 2000.
\newblock \href {http://dx.doi.org/10.1142/4252} {\path{doi:10.1142/4252}}.

\bibitem[DW11]{DW11}
Andrew Drucker and Ronald~{de} Wolf.
\newblock {\em Quantum Proofs for Classical Theorems}.
\newblock Number~2 in Graduate Surveys. Theory of Computing Library, 2011.
\newblock \href {http://dx.doi.org/10.4086/toc.gs.2011.002}
  {\path{doi:10.4086/toc.gs.2011.002}}.

\bibitem[For02]{For02}
J{\"u}rgen Forster.
\newblock A linear lower bound on the unbounded error probabilistic
  communication complexity.
\newblock {\em Journal of Computer and System Sciences}, 65(4):612 -- 625,
  2002.
\newblock Special Issue on Complexity 2001.
\newblock \href {http://dx.doi.org/10.1016/S0022-0000(02)00019-3}
  {\path{doi:10.1016/S0022-0000(02)00019-3}}.

\bibitem[GJPW17]{GJPW17}
Mika G{\"o}{\"o}s, T.~S. Jayram, Toniann Pitassi, and Thomas Watson.
\newblock {Randomized Communication vs. Partition Number}.
\newblock In {\em 44th International Colloquium on Automata, Languages, and
  Programming (ICALP 2017)}, volume~80 of {\em Leibniz International
  Proceedings in Informatics (LIPIcs)}, pages 52:1--52:15, 2017.
\newblock \href {http://dx.doi.org/10.4230/LIPIcs.ICALP.2017.52}
  {\path{doi:10.4230/LIPIcs.ICALP.2017.52}}.

\bibitem[GLS18]{gavinsky2018randomised}
Dmitry Gavinsky, Troy Lee, and Miklos Santha.
\newblock On the randomised query complexity of composition.
\newblock {\em arXiv preprint \eprint{arXiv:1801.02226}}, 2018.

\bibitem[Gro96]{Gro96}
Lov~K. Grover.
\newblock A fast quantum mechanical algorithm for database search.
\newblock In {\em Proceedings of the Twenty-eighth Annual ACM Symposium on
  Theory of Computing}, STOC '96, pages 212--219, 1996.
\newblock \href {http://dx.doi.org/10.1145/237814.237866}
  {\path{doi:10.1145/237814.237866}}.

\bibitem[HMdW03]{HMdW03}
Peter H{\o}yer, Michele Mosca, and Ronald de~Wolf.
\newblock Quantum search on bounded-error inputs.
\newblock In {\em Automata, Languages and Programming: 30th International
  Colloquium, ICALP 2003}, volume 2719 of {\em Lecture Notes in Computer
  Science}, pages 291--299, 2003.
\newblock \href {http://dx.doi.org/10.1007/3-540-45061-0_25}
  {\path{doi:10.1007/3-540-45061-0_25}}.

\bibitem[Hoz17]{hoza2017quantum}
William~M Hoza.
\newblock Quantum communication-query tradeoffs.
\newblock {\em arXiv preprint \eprint{arXiv:1703.07768}}, 2017.

\bibitem[Kah91]{Kah91}
Jean-Pierre Kahane.
\newblock Jacques hadamard.
\newblock {\em The Mathematical Intelligencer}, 13(1):23--29, Dec 1991.
\newblock \href {http://dx.doi.org/10.1007/BF03024068}
  {\path{doi:10.1007/BF03024068}}.

\bibitem[KdW03]{kerenidis2003exponential}
Iordanis Kerenidis and Ronald de~Wolf.
\newblock Exponential lower bound for 2-query locally decodable codes via a
  quantum argument.
\newblock In {\em Proceedings of the thirty-fifth annual ACM symposium on
  Theory of computing}, pages 106--115, 2003.
\newblock \href {http://dx.doi.org/10.1145/780542.780560}
  {\path{doi:10.1145/780542.780560}}.

\bibitem[KN06]{KN06}
E.~Kushilevitz and N.~Nisan.
\newblock {\em Communication Complexity}.
\newblock Cambridge University Press, 2006.
\newblock \href {http://dx.doi.org/10.1017/CBO9780511574948}
  {\path{doi:10.1017/CBO9780511574948}}.

\bibitem[LMR{\etalchar{+}}11]{LMR+11}
Troy Lee, Rajat Mittal, Ben~W. Reichardt, Robert {\v S}palek, and Mario
  Szegedy.
\newblock Quantum query complexity of state conversion.
\newblock In {\em Foundations of Computer Science (FOCS 2011)}, pages 344--353,
  2011.
\newblock \href {http://arxiv.org/abs/1011.3020} {\path{arXiv:1011.3020}},
  \href {http://dx.doi.org/10.1109/FOCS.2011.75}
  {\path{doi:10.1109/FOCS.2011.75}}.

\bibitem[LS09a]{LS09a}
Troy Lee and Adi Shraibman.
\newblock An approximation algorithm for approximation rank.
\newblock In {\em 24th Annual IEEE Conference on Computational Complexity},
  pages 351--357, 2009.
\newblock \href {http://dx.doi.org/10.1109/CCC.2009.25}
  {\path{doi:10.1109/CCC.2009.25}}.

\bibitem[LS09b]{LS09b}
Nati Linial and Adi Shraibman.
\newblock Lower bounds in communication complexity based on factorization
  norms.
\newblock {\em Random Structures \& Algorithms}, 34(3):368--394, 2009.
\newblock \href {http://dx.doi.org/10.1002/rsa.20232}
  {\path{doi:10.1002/rsa.20232}}.

\bibitem[LS{\v{S}}08]{LSS08}
Troy Lee, Adi Shraibman, and Robert {\v{S}}palek.
\newblock A direct product theorem for discrepancy.
\newblock In {\em 23rd Annual IEEE Conference on Computational Complexity (CCC
  2008)}, pages 71--80, 2008.
\newblock \href {http://dx.doi.org/10.1109/CCC.2008.25}
  {\path{doi:10.1109/CCC.2008.25}}.

\bibitem[LT17]{TL17}
Mathieu Lauri{\`e}re and Dave Touchette.
\newblock {The Flow of Information in Interactive Quantum Protocols: the Cost
  of Forgetting}.
\newblock In {\em 8th Innovations in Theoretical Computer Science Conference
  (ITCS 2017)}, volume~67 of {\em Leibniz International Proceedings in
  Informatics (LIPIcs)}, pages 47:1--47:1, Dagstuhl, Germany, 2017. Schloss
  Dagstuhl--Leibniz-Zentrum fuer Informatik.
\newblock \href {http://dx.doi.org/10.4230/LIPIcs.ITCS.2017.47}
  {\path{doi:10.4230/LIPIcs.ITCS.2017.47}}.

\bibitem[Mon14]{Mon14}
Ashley Montanaro.
\newblock A composition theorem for decision tree complexity.
\newblock {\em Chicago Journal of Theoretical Computer Science}, 2014(6), July
  2014.
\newblock \href {http://dx.doi.org/10.4086/cjtcs.2014.006}
  {\path{doi:10.4086/cjtcs.2014.006}}.

\bibitem[NS94]{nisan1994degree}
Noam Nisan and Mario Szegedy.
\newblock On the degree of boolean functions as real polynomials.
\newblock {\em Computational complexity}, 4(4):301--313, 1994.
\newblock \href {http://dx.doi.org/10.1007/BF01263419}
  {\path{doi:10.1007/BF01263419}}.

\bibitem[O'D04]{ODONNELL200468}
Ryan O'Donnell.
\newblock Hardness amplification within np.
\newblock {\em Journal of Computer and System Sciences}, 69(1):68 -- 94, 2004.
\newblock Special Issue on Computational Complexity 2002.
\newblock \href {http://dx.doi.org/10.1016/j.jcss.2004.01.001}
  {\path{doi:10.1016/j.jcss.2004.01.001}}.

\bibitem[Pat92]{Pat92}
Ramamohan Paturi.
\newblock On the degree of polynomials that approximate symmetric boolean
  functions (preliminary version).
\newblock In {\em Proceedings of the Twenty-fourth Annual ACM Symposium on
  Theory of Computing}, STOC '92, pages 468--474, 1992.
\newblock \href {http://dx.doi.org/10.1145/129712.129758}
  {\path{doi:10.1145/129712.129758}}.

\bibitem[Raz03]{Raz03}
Alexander~A Razborov.
\newblock Quantum communication complexity of symmetric predicates.
\newblock {\em Izvestiya: Mathematics}, 67(1):145, 2003.
\newblock \href {http://dx.doi.org/10.1070/IM2003v067n01ABEH000422}
  {\path{doi:10.1070/IM2003v067n01ABEH000422}}.

\bibitem[Rei11]{Rei11}
Ben~W Reichardt.
\newblock Reflections for quantum query algorithms.
\newblock In {\em Proceedings of the twenty-second annual ACM-SIAM symposium on
  Discrete Algorithms (SODA 2011)}, pages 560--569. SIAM, 2011.
\newblock \href {http://arxiv.org/abs/1005.1601} {\path{arXiv:1005.1601}},
  \href {http://dx.doi.org/10.1137/1.9781611973082.44}
  {\path{doi:10.1137/1.9781611973082.44}}.

\bibitem[San18]{San18}
Swagato Sanyal.
\newblock A composition theorem via conflict complexity.
\newblock {\em arXiv preprint \eprint{arXiv:1801.03285}}, 2018.

\bibitem[She11]{She11}
Alexander~A. Sherstov.
\newblock The pattern matrix method.
\newblock {\em SIAM Journal on Computing}, 40(6):1969--2000, 2011.
\newblock \href {http://dx.doi.org/10.1137/080733644}
  {\path{doi:10.1137/080733644}}.

\bibitem[She12]{She12}
Alexander~A. Sherstov.
\newblock Strong direct product theorems for quantum communication and query
  complexity.
\newblock {\em SIAM Journal on Computing}, 41(5):1122--1165, 2012.
\newblock \href {http://dx.doi.org/10.1137/110842661}
  {\path{doi:10.1137/110842661}}.

\bibitem[She13a]{She13a}
Alexander~A. Sherstov.
\newblock Approximating the {AND-OR} tree.
\newblock {\em Theory of Computing}, 9(20):653--663, 2013.
\newblock \href {http://dx.doi.org/10.4086/toc.2013.v009a020}
  {\path{doi:10.4086/toc.2013.v009a020}}.

\bibitem[She13b]{sherstov2009intersection}
Alexander~A. Sherstov.
\newblock The intersection of two halfspaces has high threshold degree.
\newblock {\em SIAM Journal on Computing}, 42(6):2329--2374, 2013.
\newblock \href {http://dx.doi.org/10.1137/100785260}
  {\path{doi:10.1137/100785260}}.

\bibitem[She13c]{She13b}
Alexander~A. Sherstov.
\newblock Making polynomials robust to noise.
\newblock {\em Theory of Computing}, 9(18):593--615, 2013.
\newblock \href {http://dx.doi.org/10.4086/toc.2013.v009a018}
  {\path{doi:10.4086/toc.2013.v009a018}}.

\bibitem[Shi02]{shi2002approximating}
Yaoyun Shi.
\newblock Approximating linear restrictions of boolean functions, 2002.
\newblock URL: \url{https://web.eecs.umich.edu/~shiyy/mypapers/linear02-j.ps}.

\bibitem[Tal13]{Tal13}
Avishay Tal.
\newblock Properties and applications of boolean function composition.
\newblock In {\em Innovations in Theoretical Computer Science (ITCS 2013)},
  pages 441--454, 2013.
\newblock \href{http://eccc.hpi-web.de/report/2012/163/}{\tt TR12-163}.
\newblock \href {http://dx.doi.org/10.1145/2422436.2422485}
  {\path{doi:10.1145/2422436.2422485}}.

\bibitem[Tou15]{Touchette}
Dave Touchette.
\newblock Quantum information complexity.
\newblock In {\em Proceedings of the Forty-seventh Annual ACM Symposium on
  Theory of Computing}, STOC '15, pages 317--326, 2015.
\newblock \href {http://dx.doi.org/10.1145/2746539.2746613}
  {\path{doi:10.1145/2746539.2746613}}.

\bibitem[Wil14]{WilliamsACC}
Ryan Williams.
\newblock Nonuniform {ACC} circuit lower bounds.
\newblock {\em Journal of the ACM}, 61(1):2:1--2:32, 2014.
\newblock \href {http://dx.doi.org/10.1145/2559903}
  {\path{doi:10.1145/2559903}}.

\bibitem[Yao82]{yao1982theory}
Andrew~C. Yao.
\newblock Theory and application of trapdoor functions.
\newblock In {\em Proceedings of the 23rd Annual Symposium on Foundations of
  Computer Science}, SFCS '82, pages 80--91, 1982.
\newblock \href {http://dx.doi.org/10.1109/SFCS.1982.95}
  {\path{doi:10.1109/SFCS.1982.95}}.

\end{thebibliography}

\end{document}